%% file: boundary_communities_BHS.tex
\documentclass[secnumarabic,amssymb, nobibnotes, aps, prd,showpacs,balancelastpage]{revtex4-1}

\usepackage{tikz}
\usepackage{pgfplots}
\pgfplotsset{compat=newest}
\usetikzlibrary{calc, positioning}
\usetikzlibrary{arrows, backgrounds, shapes.misc}
\usepackage{bm}

\include{declarations_for_revtex}

\usepackage{soul}
\usepackage[]{rotating}
\newcommand{\myturn}[1]{\begin{turn}{90}#1\end{turn}}

\begin{document}

\title{Incompatibility boundaries for properties of community partitions}%

\author{A.~Browet}
\affiliation{Universit\'e catholique de Louvain, ICTEAM institute}
\author{J.M.~Hendrickx}
\affiliation{Universit\'e catholique de Louvain, ICTEAM institute}
\author{ A.~Sarlette}
\affiliation{INRIA Paris (QUANTIC lab) \& Ghent University (Data Science Lab)}

\begin{abstract}
We prove the incompatibility of certain desirable properties of community partition quality functions. Our results generalize the impossibility result of [Kleinberg 2003] by considering sets of weaker properties. In particular, we use an alternative notion to solve the central issue of the consistency property. (The latter means that modifying the graph in a way consistent with a partition should not have counterintuitive effects). Our results clearly show that community partition methods should not be expected to perfectly satisfy all ideally desired properties. 

We then proceed to show that this incompatibility no longer holds when slightly relaxed versions of the properties are considered, and we provide in fact examples of simple quality functions satisfying these relaxed properties. An experimental study of these quality functions shows a behavior comparable to established methods in some situations, but more debatable results in others. This suggests that defining a notion of good partition in communities probably requires imposing additional properties. 
\end{abstract}


\pacs{89.75.Fb, 89.75.Hc, 05.10.-a, 89.65.-s, 87.23.Ge, 89.20Hh,  87.53.Wz}

\maketitle

\section{Introduction}

In the rich literature about defining communities in graphs, two major and complementary approaches are the proposal of particular criteria defining communities and the identification of axioms that should reasonably be satisfied by such criteria, in particular by value functions that the optimal partition should maximize.
When confronting the two approaches, the popular clustering criteria all fail to satisfy at least one of the reasonable axioms. For instance, the popular modularity criterion~\cite{Newman2004} is neither local (\ie the optimal partitioning of a subset of nodes does depend on the whole graph) nor consistently improving (\ie an optimal community partitioning for a given graph can loose its optimality when strengthening intra-community links and weakening inter-community links).

In fact Kleinberg~\cite{Kleinberg2002} has proved that it is impossible for a function $F$ associating a community partition to any weighted graph\footnote{This was in a framework where nodes are separated by a \quotes{distance}, as opposed to nodes being connected by edges with a certain weight as considered here. The two frameworks are however equivalent provided all weights are positive and self-loops can be discarded; it suffices then indeed to take the inverse of the weights as distances.} to simultaneously satisfy the three following properties, although getting two of them separately is easy:
\begin{itemize}
\item Richness: given any partition of the nodes $\Sigma$, there exists at least one graph for which $F$ returns $\Sigma$ as the unique\footnote{Uniqueness is necessary to avoid trivial multi-valued functions associating for example all partitions to all graphs.} associated partition;
\item Scale-invariance: given any graph, multiplying all the edge weight by a constant does not change the partition returned by $F$;
\item Consistency:  if $F$ returns a partition $\Sigma$ for a graph $G$, then it also returns $\Sigma$ for any graph $G'$ obtained by increasing the weight of intra-communities edges and decreasing the weight of inter-communities edges with respect to $\Sigma$.
\end{itemize}
The last property means that modifying a graph in a way \quotes{consistent} with its associated community partition, by increasing links inside communities and decreasing links between communities, 
should lead to a graph with the same community partition. Despite its natural formulation, this property is actually debatable \cite{VanLaarhoven:2014}, and we will see in Section \ref{sec:properties} that it has counterintuitive consequences.

Other recent approaches have run into similar barriers, e.g.~the hierarchy-based axioms in \cite{arxiv13017724} would not admit scale invariance; the various subgraph-based alternatives in \cite{Radicchi2004} and \cite{Hu2008} do not return a single partition per graph, and in fact for any graph they return among others the single community encompassing the entire node set, hence failing richness. More examples will be covered later.\\

In this paper, we begin by presenting two generalizations of the result of \cite{Kleinberg2002} in a more flexible framework where a value function evaluates the quality of each community partition for a graph (as opposed to considering a function that directly associates a community partition to each graph). The \quotes{best} partition associated to a graph can then be obtained by optimizing this value function over all possible partitions, but \quotes{sufficiently good} partitions can also be computed when an exact solution to this generally hard optimization problem is out of reach.
This approach is followed in \cite{VanLaarhoven:2014}, and is consistent with many well-known community partition methods, as those listed in Table \ref{table:costfunction} in Section \ref{sec:properties}. 

Our first main result (Section \ref{sec:imposs_results}) shows the impossibility of satisfying richness, scale-invariance and (together with some continuity assumption about the value function) \emph{a much weaker form of consistency than in \cite{Kleinberg2002}, which we argue appears more natural}. In our second impossibility result, the consistency requirement is removed altogether and replaced by some locality axiom, forbidding modifications in one part of a graph to affect the (relative) quality of communities in another part of the graph. 

We further add some insight to these impossibility results. The problem can be pinpointed to the request, for richness, to make both the unique ``all-encompassing'' partition optimal for some graph, and the ``all-singletons'' partition optimal for some other graph. If we drop the all-singletons partition from the richness requirement, then value functions satisfying all the axioms can be constructed, as we show by providing actual examples (Section \ref{sec:valfuncs}). It turns out that in previous papers which try to circumvent Kleinberg's impossibility result, this distinction is made implicitly. 
For instance in \cite{nips3491}, the authors propose to circumvent Kleinberg's impossibility by imposing similar axioms on a value function for partitions instead of directly on the ``clustering function'' $F$. 
Our variant of the impossibility theorem precisely shows that this does not resolve the impossibility. However, the value function proposed in~\cite{nips3491} is undefined for the ``all-singletons'' partition, and we identify that this is the key relaxation.
In \cite{nips4101}, the authors discuss and compare alternatives by incidentally restricting themselves to $k$-clustering, and \cite{rezab} considers the special case of bi-partitions. In the light of our result, fixing $k$ is just a way to satisfy the sufficient relaxation, i.e.~excluding either the all-singletons or the all-encompassing partition (and many others). In all those cases, this simple implicit relaxation of excluding the all-singleton partition (and possibly others) is in fact sufficient and the key step to allow satisfying Kleinberg's set of axioms. In this sense, the present paper establishes a clear and possibly practically relevant way to circumvent the impossibility.
Another point is that Kleinberg's impossibility refers to graphs without self-loops. Some previous investigations, like e.g.~\cite{VanLaarhoven:2014}, let self-loops play an instrumental role. In that case there are more options for constructing optimal graphs towards richness. We can in fact state that our example value functions, proposed in Section \ref{sec:valfuncs}, satisfy all the axioms of Kleinberg if we allow for self-loops.

We stress that this paper implies no value judgment about clustering criteria which fail to satisfy some of the proposed axioms (including the assumption about no self-loops). To the contrary, the impossibility results clarify that it is hopeless to look for value function-based criteria satisfying all the axioms. It is hence perfectly unavoidable to select a subset of axioms, and one reasonable selection criterion is to get the most useful results. One defensible viewpoint is to exclude some partitions as a priori irrelevant, e.g.~partitions containing almost only isolated nodes. Along this line, we analyze, in Section~\ref{sec:experiments}, the partitions obtained using one of the proposed value functions which satisfies all the axioms except strict richness for partitions involving singletons. Although our proposed value functions turn out to have strong similarities with the sum of community strengths and especially with the modular density, introduced respectively in \cite{medus2009alternative,Li2008}, they satisfy stronger axiomatic properties. Moreover, we use a more general way of normalizing the contribution of each community, allowing for some additional flexibility. The conclusion of our experimental investigation is that only particular tuning of our value function parameters leads to results that are compatible with expectations on benchmark problems. Hence, satisfying the ``standard'' axioms considered here is not a guarantee for more relevant results.\\

Towards the future, our results hence not only clarify that the set of historically proposed axioms cannot be kept in its most general form; they also highlight the need to add \emph{compatible} axioms that would isolate a most useful set of value functions, since our own examples satisfy a minimally relaxed set of axioms, yet they still leave a design freedom among which far from all choices behave as intuitively desired.

\section{Properties of quality functions}\label{sec:properties}

As mentioned in the introduction, there are different ways of specifying what a good partition in communities is \cite{fortunato2010community,schaeffer2007graph}. One can for example directly specify the properties that the (best) partition should have, or the algorithm to obtain it, as in the framework of \cite{Newman2004}.
We follow here the popular and more flexible option \cite{VanLaarhoven:2014} of defining a quality measure of a given partition into communities for any given weighted graph. The best partition into communities is then the one maximizing the quality function. This optimization problem is often computationally challenging \cite{Brandes2008a}, but the use of a quality function defined for all partitions allows using heuristic methods to compute relatively good partitions.

More formally, we consider weighted symmetric graphs $G(V,W)$ without self-loops, where $V=\{1,\dots,N\}$ is the set of nodes, $N$ being the number of nodes, and $W$ the set of weights $W_{ij}= W_{ji}\geq 0$ (with $W_{ii} = 0$). When $W_{ij}>0$, we say that $i$ and $j$ are connected or that $(i,j)=(j,i)$ is an edge, and it will sometimes be convenient to refer to the set $E$ of edges, which is entirely determined by $W$.
A partition $\sigma$ of $V$ into communities corresponds to an assignment of each  node $i\in V$ to a community label $\sigma_i\in \{1,\dots,n_c\}$, where $n_c$ depends on the partition $\sigma$. The partition induces communities $c_1,\dots,c_{n_c}$ defined by $c_k=\{i|\sigma_i=k\}$ and whose cardinality, i.e.~the number of nodes within the community, is denoted $\vert c_k \vert$.
We sometimes use the Kronecker delta $\delta\left(\sigma_i, \sigma_j\right)=1$ if $\sigma_i=\sigma_j$, and  0 otherwise, to express whether $i$ and $j$ belong to the same community. We denote by $s^{int}_i = \sum_{j \in V} W_{ij}~\delta\left(\sigma_i, \sigma_j\right)$ (resp. $s^{ext}_i = \sum_{j \in V} W_{ij}~\left(1-\delta\left(\sigma_i, \sigma_j\right)\right)$) the internal (resp. external) strength of node $i$, i.e.~the sum of edge weight connecting node $i$ to nodes in the same (resp. in any other) community, $s_i = \sum_{j\in V} W_{ij} = s^{ext}_i+s^{int}_i$ the total strength of node $i$, and $m=\sum_{i \in V} s_i$ which equals twice the total weight of the graph.

We then consider a \emph{value (or quality) function} $f(G,\sigma)$, that represents the quality of the community partition $\sigma$ on the graph $G$.
A classical example is Newman's Modularity \cite{Newman2004}, defined for a weighted undirected graph as $Q(G, \sigma) = \frac{1}{m}\sum\limits_{i,j\in V} \left(W_{ij}-\frac{s_i\,s_j}{m}\right)\delta\left(\sigma_i, \sigma_j\right)$, which measures the difference between the actual fraction of edges falling inside the communities and the expected fraction of such edges under the configuration null model with respect to the partition.
Since the graph is entirely determined by the matrix of weights $W$, we will sometimes abuse notation and refer to $f(W,\sigma)$.

Following the approach of \cite{VanLaarhoven:2014}, we now list some properties that are considered desirable for value functions. These properties are either taken or adapted from \cite{VanLaarhoven:2014}.
The first property is relatively natural towards ensuring robustness of conclusions with respect to the data, and towards providing favorable settings for community-finding algorithms.
\begin{proper}[Continuity]\label{proper:continuity}
For any community partition $\sigma$, the value function $f$ is continuous with respect to the weights $W$.
\end{proper}

Thus property~\ref{proper:continuity} excludes quality functions that would heavily rely on the presence or absence of an edge without considering its weight. Strictly speaking, the impossibility proofs later in the paper only require continuity ``at potential optimal $W,\sigma$ combinations''. As this is difficult to guarantee a priori and questionable for practical purposes, we here require full continuity.

The second property requires that only the \emph{ratio of weights} on different edges is relevant, as far as clustering decisions are concerned.
\begin{proper}[Scale invariance]\label{proper:scale}
For any graph $G(V,W)$, community partitions $\sigma,\sigma'$ of $V$ and $\alpha >0$, the following implication holds:
\begin{equation}\label{eq:scale}
f(W,\sigma) \geq f(W,\sigma') \Rightarrow f(\alpha W,\sigma) \geq f(\alpha W, \sigma'),
\end{equation}
or equivalently\footnote{In both statements, if we have an equality on the left hand side then $\sigma$ and $\sigma'$ can be swapped, so we also need equality on the right hand side. From this, the second statement readily implies the first. Furthermore, by redefining $W'=\alpha W$ and $\alpha'=1/\alpha$, we can reverse the first statement so equality on the right hand side also must imply equality on the left hand side, i.e.~inequality on the left can only be associated to inequality on the right. Thus the first statement implies the second one.} $f(W,\sigma) > f(W,\sigma') \Rightarrow f(\alpha W,\sigma) > f(\alpha W, \sigma').$
\end{proper}
Scale-invariance induces that the communities should not depend on an exogeneous threshold value for individual edge weights. Value functions like the adaptive scale modularity~\cite{VanLaarhoven:2014}, or the constant Potts model~\cite{Traag2011} are therefore not scale invariant.

The next property represents the fact that the community partition of one part of the graph should not be affected by the structure of other parts of the graph. It is notably not satisfied by the modularity, which suffers from the well-known resolution limit \cite{Good2009, Traag2011}, or in general when the clustering is influenced by the average weight in the graph. Different notions of locality can be proposed and we use the following.

\begin{proper}[$k$-locality]\label{proper:locality}
Given $k \in \{ 0,1 \}$, consider two graphs $G_1(V,W^{(1)})$ and $G_2(V,W^{(2)})$ whose restriction to a subset of nodes $V_0$ and its neighbors at distance $k$ is identical, that is: $W^{(1)}_{ij}=W^{(2)}_{ij}$ for all $i,j\in V_0$ and, if $k=1$, for all $i,j$ for which either $i$ or $j$ belongs to $V_0$. Consider then a community partition $\sigma$ for which $C_0=\left\{i\left| i\in V_0\right.\right\}$ is a community, and another community partition $\sigma'$ exactly identical to $\sigma$ except that the community $C_0$ is split in two communities $C_{01}$ and $C_{02}$. Then there holds
\begin{equation}\label{eq:local}
f(W^{(1)},\sigma) \geq f(W^{(1)},\sigma') \Leftrightarrow  f(W^{(2)},\sigma) \geq f(W^{(2)},\sigma')
\end{equation}
\end{proper}
Zero-locality thus means that the decision of splitting a community $C_0$ into two communities only depends on the weights of edges incident to nodes within $C_0$. 
The weaker notion of 1-locality would allow this decision to also depend on the edges incident to one node of $C_0$ and one node outside $C_0$. A similar case appears in the definition of ``locality'' in \cite{VanLaarhoven:2014} for instance. In the definition of \cite{VanLaarhoven:2014} however, $G_1$ and $G_2$ are allowed to have different node sets $V^{(1)}$, $V^{(2)}$ and $C_0$ might be split in an arbitrary way in both $\sigma$ and $\sigma'$. One could also request $\{ W^{(1)}_{ij}: j \in V^{(1)} \} = \{ W^{(2)}_{ij} : j \in V^{(2)}\}$ for each $i \in C_0$, without requiring that the endpoints $j$ of each edge match in $G_1$ and $G_2$.
Such locality notions would impose condition \eqref{eq:local} on a \emph{larger} class of situations with respect to our definition of Property \ref{proper:locality}. Thus our definition is a weaker property, easier to satisfy, hence providing a stronger impossibility result.
When proposing cost functions that do satisfy locality in Section \ref{sec:valfuncs}, we will show that they actually satisfy stronger notions of locality.

The next property excludes value functions for which some ``relevant'' partitions would never be optimal, independently of the graph. At this stage, we abstractly define the set $\Sigma$ of relevant partitions, which would typically depend on the application.

\begin{proper}[Richness, with respect to a set of partitions $\Sigma$]\label{proper:richness_set}  
For any partition $\sigma\in \Sigma$ of a set $V$ of nodes, there exists a graph $G(V,W)$, for which $\sigma$ is a strictly optimal community partition:
$
f(G,\sigma) > f(G,\sigma') \hspace{.4cm} \forall \sigma' \in \Sigma \setminus\{\sigma\} \, .
$
\end{proper}
Again, a stronger property could be stated by comparing to all partitions $\sigma'$ instead of only those in $\Sigma$. This makes no real difference for our results since our impossibility results in section~\ref{sec:imposs_results} consider $\Sigma$ to be the set of all partitions, and the proposed value functions in section \ref{sec:valfuncs} would also satisfy this stronger version of Property \ref{proper:richness_set}.

The properties introduced so far do not imply that nodes inside a community should be more connected to each other than to those outside of the community, which corresponds to the general intuitive idea of community partition. 
In order to formalize this idea, we introduce the notion of consistent improvement as in \cite{VanLaarhoven:2014}. Consider a graph $G(V,W)$ and a community partition $\sigma$. We say that $G'(V,W')$ is \emph{a consistent improvement of $G(V,W)$ with respect to $\sigma$} if $W'_{ij}\geq W_{ij}$ for all $i,j$ for which $\sigma_i=\sigma_j$; and $W'_{ij}\leq W_{ij}$ for all $i,j$ for which $\sigma_i \neq \sigma_j$. Compared to $G$, the graph $G'$ thus has links that are stronger inside the communities defined by $\sigma$ and weaker between these communities.

Even with this notion, it turns out to be non-trivial to formalize the fact that the dependence of partitions on weights should be consistent with our intuitive idea of communities.

A natural formulation would be that an optimal partition should remain optimal for all consistent improvements of the graph with respect to this partition. This ``absolute consistency'' requirement is the one used in Kleinberg's impossibility result \cite{Kleinberg2002}. But this condition actually has counterintuitive consequences, as illustrated on Figure \ref{fig:illustration_consistent_improvement}. Consider indeed a graph $G$ with a clique of four nodes $\{1,2,3,4\}$ all connected one to each other by edges of similar weights. Intuitively, one would want to consider this clique as a community in the optimal partition $\sigma$. Now any graph obtained by strongly increasing the weights of the edges $(1,2)$ and $(3,4)$ is a consistent improvement of $G$ with respect to this partition. If these weights are sufficiently increased however, we argue that partitioning those four nodes into two communities $\{1,2\}$ and $\{3,4\}$ would be more natural. This would however not be allowed by ``absolute consistency'' (see Figure \ref{fig:illustration_consistent_improvement} top), which requires the community $\{1,2,3,4\}$ to remain optimal under this consistent improvement. The impossibility proof by Kleinberg in fact relies precisely on this fact: if a partition belongs to the richness set, then none of its sub-partitions can ever be optimal.

There is however a natural workaround when using value functions: a consistent improvement should increase the value of a graph partition, but nothing forbids that the quality of another graph partition, consistent with the same improvement, increases even more and supplants the initial partition. On the example of Figure \ref{fig:illustration_consistent_improvement}, since the graph modification proposed on the top right is a consistent improvement for both the partition using $\{1,2,3,4\}$ and the one using two separate communities $\{1,2\}$ and $\{3,4\}$, this would impose no condition on their ordering. Therefore, in this paper, we only require that the ordering between partition qualities is preserved when a graph modification makes a bad partition less consistent and a good partition more consistent, as illustrated on the bottom right of Figure \ref{fig:illustration_consistent_improvement}. This weaker requirement is similar to the relative monotonicity in \cite{VanLaarhoven:2014} and we hence call it ``relative consistency''.

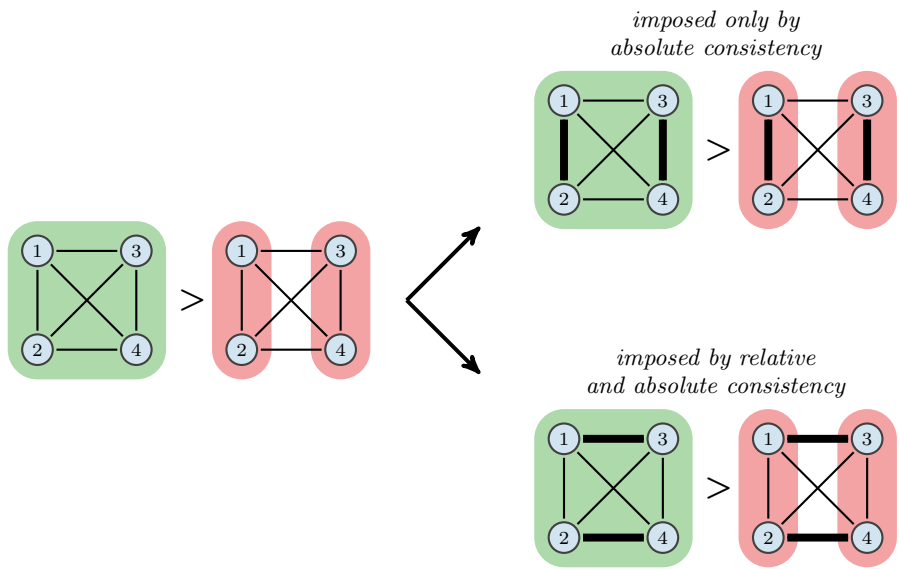
\begin{figure}
\centering
\input{figs/improvment.tex}
\caption{Illustration of the difference between (top) \quotes{absolute} consistency requirement \cite{Kleinberg2002} and (bottom) \quotes{relative} consistency requirement (Property \ref{proper:relative_consistent_improvement}). Symbol ``$>$'' denotes that the left partition is better than the right one.}\label{fig:illustration_consistent_improvement}
\end{figure}

\begin{proper}[relative consistent improvement]\label{proper:relative_consistent_improvement}
Consider two graphs $G(V,W)$, $G'(V,W')$ and two partitions $\sigma, \sigma'$ of $V$. Suppose that

(i) $f(G,\sigma) > f(G,\sigma')$,

(ii) $G'$ is a consistent improvement of $G$ with respect to $\sigma$,

(iii) $G$ is a consistent improvement of $G'$ with respect to $\sigma'$;

\noindent then $f(G',\sigma) > f(G',\sigma')$.
\end{proper}

In the following sections, we establish strict results about the impossibility to define quality functions which satisfy all the above properties. Table \ref{table:costfunction} lists how those properties are satisfied by some quality functions proposed in the literature. It is notable that none of them satisfies all properties, and we prove in Section \ref{sec:imposs_results} that this is indeed impossible. Note that some proposals for community partitioning involve more complex constructions than optimizing a value function. A simple variation is e.g.~introducing constraints on admissible partitions in the spirit of \cite{Radicchi2004}, which could be reformulated as violating continuity and restricting richness. Other examples can be more difficult to (minimally) fit in the value function framework, and for the sake of simplicity we choose to not include any of them in the table below.\vspace{3mm}

\def\asp{$y$}
\def\asn{$n$}
\setlength{\tabcolsep}{10pt}
{\small
\begin{sidewaystable}
\begin{center}
\begin{tabular}{rl ccccc}
 & $f(W,\sigma)$ & \myturn{Continuity} & \myturn{Scale invariance} & \myturn{$k$-locality} & \myturn{Full Richness} & \myturn{Rel. consistency}\\[10pt] \hline \vspace{10pt}
Modularity \cite{Newman2004} & $\frac{1}{m}\textstyle{\sum_{i,j\in V}} \left(W_{ij}-\frac{s_i\,s_j}{m}\right)\delta\left(\sigma_i, \sigma_j\right)$ &
\asp & \asp & \asn & \asn\footnote{$\Sigma$ must exclude all partitions that contain an isolated node.} & \asn \\[12pt]
Adaptative scale modularity \cite{VanLaarhoven:2014} & $\sum\limits_{k =1}^{n_c} \left[\frac{\sum_{\sigma_i=\sigma_j=k} W_{ij}}{{M+\gamma v_k}}-\left(\frac{v_k}{M+\gamma v_k}\right)^2\right]\;,\;\; v_k = \sum\limits_{\sigma_i=k \atop j\in V} W_{ij}$ 
& \asp & $\sim$\footnote{Yes for $M=0$}
& $k=1$ & $\sim$\footnote{Yes for $\gamma=0$ at least. No for $M=0$, nor for $\gamma \geq 1$ (note that the authors reach another conclusion \emph{with} self-loops.)} & $\sim$\footnote{Yes for $\gamma \geq 2$, as claimed by authors}\\[12pt]
Constrained q-states Potts model \cite{reichardt2004} & $\alpha\sum\limits_{\substack{i,j\in V\\ \sigma_i=\sigma_j}} \left[W_{ij}\right] - \beta \sum_{k=1}^{n_c}\left[\left| c_k \right| (\left| c_k \right|-1)\right]$
& \asp & \asn & $k=0$ & \asp & \asp \\[12pt]
Constant Potts model \cite{Traag2011} & $\sum\limits_{\substack{i,j\in V\\ \sigma_i=\sigma_j}} \left[W_{ij}-\gamma\right]$ &
\asp & \asn & $k=0$ & $\sim$\footnote{Yes for $\gamma\geq0$} & \asp \\[12pt]
Modularity density \cite{Li2008} & $\sum\limits_{k=1}^{n_c}\frac{1}{\left|c_k\right|}\sum\limits_{i\in c_k}s_i^{int}-s_i^{ext}$ & 
\asp & \asp & $k=1$ & \asn\footnote{\label{note:Sig2s}$\Sigma$ must exclude at least all partitions with more than one singleton community.} & \asp\\[12pt]
Sum of community strengths \cite{medus2009alternative}\footnote{Those authors add a positivity constraint; we here consider it without constraints.} & $\sum_{k=1}^{n_c} \frac{\sum_{i \in \sigma_k} \, s_i^{int} - s_i^{ext}}{\sum_{i \in \sigma_k} \, s_i^{int} + s_i^{ext}}$ & 
\asp & \asp & $k=1$ & \asn \textsuperscript{\ref{note:Sig2s}}
 & \asp\\[12pt]
Label Propagation \cite{tibely2008} & $\sum\limits_{\substack{i,j\in V\\ \sigma_i=\sigma_j}} \left[W_{ij}
\right]$ &
\asp & \asp & $k=0$ & \asn & \asp \\[12pt]
Map Equation \cite{Rosvall2008} & optimal coding of random walk &
y\footnote{for connected graph} & y & n & \asn\textsuperscript{\ref{note:Sig2s}}
& ? \\[12pt]
\end{tabular}
\end{center}
\label{table:costfunction}
\caption{Properties satisfied by quality functions.}
\end{sidewaystable}
}

\noindent \textbf{Remark:} To conclude this section, we would like to insist on two points about the freedom allowed by the introduced properties.\vspace{2mm}

\noindent \textbf{R1.} All our properties (with the exception of continuity) essentially state that the \emph{ordering} of partitions according to increasing values of $f$ must be invariant under some graph transformations, but say nothing about the actual values of $f$ when the graph changes. This does not preclude that e.g.~referring to property \ref{proper:locality} (locality) the change in value from $f(W,\sigma)$ to $f(W,\sigma')$ depends on weight values outside $C_0$ and its neighbors, as long as this dependence does not affect the ordering $f(\cdot,\sigma) > f(\cdot,\sigma')$.
This ordering-based approach allows more flexibility in the expression of the accepted value functions. One could of course consider analog but stronger properties about the \emph{values} of $f$. In particular, one can verify that the following properties are sufficient but not necessary for satisfying their weaker order-based counterpart:

[value-scale invariance] For any $W,\sigma$ and $\alpha>0$, we have $f(\alpha W,\sigma) = a(\alpha)\, f(W,\sigma)$ for some monotonically increasing function $a(\alpha) >0$.

[value-$k$-locality] Consider two graphs $G_1(V,W^{(1)})$ and $G_2(V,W^{(2)})$ whose restrictions to a subset of nodes are identical as in Property \ref{proper:locality}; and partitions $\sigma, \sigma'$ which are identical except possibly on $C_0$ and for both of which $C_0$ is a union of communities. Then 
$$f(G_1,\sigma)-f(G_1,\sigma') = f(G_2,\sigma) - f(G_2,\sigma') \, .$$

[value-consistent improvement] Let $G(V,W')$ be a consistent improvement of $G(V,W)$ with respect to a community partition $\sigma$, then $f(G',\sigma)\geq f(G,\sigma)$.\vspace{2mm}

\noindent \textbf{R2.} We do not impose invariance under permutation of the nodes as in \cite{VanLaarhoven:2014}. This means that the value functions could depend on the particular labels of the nodes, which \emph{does allow} using prior knowledge on the community partition.
For instance, to reinforce the value of having two particular nodes $i,j$ in the same community, we can define $f_p(W,\sigma) = f(W',\sigma)$ where $W'_{1,2}=p W_{1,2}$ for some $p\gg 1$ and $W'_{i,j}=W_{i,j}$ for all other $i,j$. If $f$ satisfies some of the above properties, then the same properties are satisfied by $f_p(W,\sigma)$.

In the sequel, our impossibility results in Section \ref{sec:imposs_results} allow all this freedom, while the value functions which we introduce as feasibility examples in Section \ref{sec:valfuncs} do satisfy the stronger value-based properties of R1 and are independent of node labels. In this sense, the distinction of the ``extra freedom'' in R1 and R2 allows no extra margin in the impossibility and possibility results proposed in the present paper.

\section{Impossibility results}\label{sec:imposs_results}

We now provide two extensions to Kleinberg's result \cite{Kleinberg2002} about incompatible axioms for community partition based on optimizing a cost function. Our first result replaces the absolute notion of consistency used in \cite{Kleinberg2002} by our relative one, Property \ref{proper:relative_consistent_improvement} (together with the continuity assumption of Property \ref{proper:continuity}, related to the quality-function based framework we use). Our second result shows that even when dropping the axiom of consistent improvement, it is impossible to satisfy all the other axioms if 0-locality is required.

In Section \ref{sec:valfuncs}, we will further identify richness as a central culprit for these impossibilities. Indeed, we will provide explicit value functions which show that the axioms can all be satisfied if richness is required with respect to all except a few particular partitions. Furthermore, we will see that these value functions satisfy all the axioms, including full richness, if the latter is allowed to exploit graphs with self-loops. The impossibility result \emph{\`a la} Kleinberg indeed applies to the setting without self-loops.\\

The next lemma, stating that graphs for which a partition is strictly optimal can always be assumed to have positive weights $W_{ij}>0$ on all pairs of nodes, will be useful in the sequel. The proof is available in Appendix \ref{app:proof_positive_weights}.

\begin{lem}\label{lem:positive_weights}
Let $f$ be a value function satisfying Property \ref{proper:continuity} on continuity with respect to $W$, and let $G(V,W)$ be a graph for which a partition $\sigma$ is strictly optimal: $f(G,\sigma) > f(G,\sigma')$ for all $\sigma' \neq \sigma$. Then there exists a $\delta>0$ such that $\sigma$ is also strictly optimal for any graph $G(V,W')$ with $|W'_{ij}-W_{ij}| < \delta$ for all $i,j$.

In particular, there exists an open set of graphs $G^+(V,W^+)$ for which $\sigma$ is also strictly optimal and for which $W^+_{ij}>0$ for all $i\neq j$. 
\end{lem}

Continuity is required only through this essential implication in the proofs of two main results presented in this Section.

\subsection{Generalizing Kleinberg's impossibility result}

\begin{thm}\label{thm:imposs_cont_scale_rich_relimpr}
There is no value function satisfying properties \ref{proper:continuity} (continuity), \ref{proper:scale} (scale invariance), \ref{proper:richness_set} (richness) with respect to all partitions, and \ref{proper:relative_consistent_improvement} (relative consistent improvement).
This holds true even if the value function is only required to be defined for a specific number of nodes $N>1$.
\end{thm}
More particularly, we prove the impossibility if the richness set $\Sigma$ must contain both the partition in one community to which all nodes belong, and the partition in $N$ singleton communities.

\begin{proof}
We consider any set $V$ of $N>1$ nodes and any value function satisfying the four properties mentioned in the statement of the result, and derive a contradiction.

Let $\sigma^g$ be a partition in one community to which all nodes belong, and $\sigma^s$ a partition in $N$ singleton communities. Let then $G^g(V,W^g)$ and $G^s(V,W^s)$ be graphs for which the respective partitions $\sigma^g$ and $\sigma^s$ are optimal; these graphs exist from the richness Property \ref{proper:richness_set}. In particular $f(W^s,\sigma^s)>f(W^s,\sigma^g)$ and $f(W^g,\sigma^g) > f(W^g,\sigma^s)$. Moreover, we can assume $W^s_{ij},W^g_{ij}>0$ for all $i,j$ thanks to Property \ref{proper:continuity} (continuity) which implies Lemma \ref{lem:positive_weights}. There exists thus a $\rho>1$ such that $\rho W^s_{ij}\geq W^g_{ij}$ for all $i\neq j$. 
By scale invariance, $f(W^s,\sigma^s)>f(W^s,\sigma^g)$ implies then that $f(\rho W^s,\sigma^s)>f(\rho W^s,\sigma^g)$. 

Since $\rho W^s_{ij}\geq W^g_{ij}$ for all $i\neq j$, observe that $\rho W^s$ is a consistent improvement of $W^g$ for $\sigma^g$, since it corresponds to increasing the weights of edges which are all trivially within the unique community of $\sigma^g$. Similarly, $W^g$ is a consistent improvement of $\rho W^s$ for the partition $\sigma^s$ into $N$ singleton communities, as it corresponds to decreasing the weights of edges between these communities. It follows then from 
$f(\rho W^s,\sigma^s)>f(\rho W^s,\sigma^g)$ and Property \ref{proper:relative_consistent_improvement} that $f(W^g,\sigma^s)> f(W^g,\sigma^g)$, in contradiction with the optimality of $\sigma^g$ for $W^g$ \ie $f(W^g,\sigma^g) > f(W^g,\sigma^s)$.
\end{proof}

Note that Theorem \ref{thm:imposs_cont_scale_rich_relimpr} does not mention any locality requirement. Before trying to find a value function satisfying locality, we must anyways first relax some of the other axioms. Possibilities in this direction are outlined in Section \ref{sec:valfuncs}. Yet before this we present an impossibility result with another set of properties.

\subsection{Impossibility due to locality instead of consistent improvement}

In the light of the impossibility result of the previous section, it is relevant to relax some assumptions or replace them by weaker or alternative versions. One candidate property could be the consistent improvement. We have indeed seen that 
absolute consistency can have problematic consequences, and it could be that our version of relative consistency is still too restrictive. 
Besides, it is not satisfied by some famous community partition methods such as the modularity maximization (see Table \ref{table:costfunction}).

This section however highlights a new impossibility, independent of any requirement related to consistent improvement, and that would therefore need to be taken into account for any property replacing consistent improvement. 
We remark that as consistent improvement is the only property of Section \ref{sec:properties} that imposes the partition to correspond to our intuitive notion of communities, this impossibility result is actually relevant for any arbitrary notion of value-based partition.

\begin{thm}\label{thm:imposs_cont_sloc_scale_rich}
There is no value function satisfying properties \ref{proper:continuity} (continuity), \ref{proper:locality} for $k=0$ ($0$-locality), \ref{proper:scale} (scale invariance) and \ref{proper:richness_set} (richness)  with respect to all partitions. This holds true even if the value function is only required to be defined for a specific number of nodes $N>1$.
\end{thm}
More particularly, we prove the impossibility if the richness set $\Sigma$ must contain a partition where two nodes form one community and a partition where these same two nodes form two singleton communities.

\begin{proof}
We will again show that the properties mentioned in the statement of the theorem imply a contradiction, for any set $V$ of $N>1$ nodes.

Consider two partitions $\sigma^a,\sigma^b$ that are identical except that nodes 1 and 2 form one community in $\sigma^a$ and form two separate singleton communities in $\sigma^b$. By the richness property, there exists two graphs $G^a(V,W^a)$ and $G^b(V,W^b)$ such that $\sigma^a$ is the unique optimal partition for $G^a$ and $\sigma^b$ the unique optimal partition for $G^b$. In particular, $f(W^a,\sigma^a)> f(W^a,\sigma^b)$ and $f(W^b,\sigma^b) > f(W^b,\sigma^a)$. 

It follows from the continuity property and thus Lemma \ref{lem:positive_weights} that the weights $W^a_{12}$ and $W^b_{12}$ of the edge connecting $1$ and $2$ in $G^a$ and $G^b$ respectively can both be assumed positive, so that $\rho = W^b_{12}/W^a_{12}>0$ is well defined. By the scale-invariance property, $f(W^a,\sigma^a)> f(W^a,\sigma^b)$ implies then $f(\rho W^a,\sigma^a)> f(\rho W^a,\sigma^b)$.

By definition of $\rho$, the restriction to $\{1,2\}$ of the graphs defined by $\rho W^a$ and $W^b$ are identical. Besides, remember that $\sigma^b$ is obtained from $\sigma^a$ by only splitting the community $\{1,2\}$ in two singleton communities. Then because of the $0$-locality property, the inequality $f(\rho W^a,\sigma^a)> f(\rho W^a,\sigma^b)$ implies $f(W^b,\sigma^a)> f(W^b,\sigma^b)$, in contradiction with the strict optimality of $\sigma^b$ for $G^b(V,W^b)$.
\end{proof}

\section{New value functions}\label{sec:valfuncs}
We have seen that none of the main value functions available in the literature satisfy all the desirable properties of Section \ref{sec:properties}, and that it is actually impossible to simultaneously satisfy their strict versions. We now show that slightly weaker versions of these properties are not mutually incompatible, by providing examples of value functions satisfying them. At this point we do not claim that these value functions lead to practically relevant community partitions; we only design them to show the possibility of simultaneously satisfying the properties of Section \ref{sec:properties}, when richness excludes a few specific partition types.\\

We first note that, when trying to follow a common approach in which the quality of a partition is computed by summing the individual contributions of all edges or of all edges inside communities, it appears necessary to introduce non-trivial scalings in order to simultaneously satisfy scale-invariance, locality and significant richness properties.
For example, in the general framework introduced by Reichardt and Bornholdt \cite{Reichardt2006}, the quality function is expressed as $f\left(W, \sigma\right) = \sum_{i,j} (W_{ij} - N_{ij})\, \delta(\sigma_i,\sigma_j)$ where $N_{ij}$ is the expected value of $W_{ij}$ in a given null model. This encourages grouping nodes in a community if they are joined by an edge whose weight exceeds the threshold $N_{ij}$. 
Now if this threshold depends on exogenous parameters (e.g.~a constant), then the value function $f$ will typically not be scale-invariant. 
On the other hand if the $N_{ij}$ are determined endogenously involving \emph{global} graph properties (like say the average node strength, as \eg in the modularity \cite{Newman2004}), then $f$ is typically non-local: a modification in any part of the graph can change the value of $N_{ij}$ and therefore affect the decision of having $i$ and $j$ in a same community or for instance as isolated nodes. This only leaves the possibility to compute $N_{ij}$ from the weights of edges incident on $i$ and $j$. A simple example of such quality function would be modularity density as defined in \cite{Li2008}:
\begin{equation}
f_{md} := \sum_{k}\frac{1}{\abs{c_k}} \sum_{i\in c_k} \prt{s_i^{int}-s_i^{ext}} \; .
\end{equation}
This is a Reichardt-Bornholdt-like version of the adaptive scale modularity from \cite{VanLaarhoven:2014} with $M=0$, satisfying scale-invariance and $1$-locality. It is not difficult to prove that if $\sigma$ is a partition with two isolated nodes, then $f_{md}$ can always be increased by grouping these nodes.

In the following, we show that by cleverly using community-size-dependent scalings in variations of $f_{md}$, one can obtain a class of quality functions with stronger locality or larger richness set than $f_{md}$ and than all the proposed quality functions we are aware of.

We refer the reader to Section \ref{sec:experiments} for an experimental investigation of the partitioning behavior with those value functions.

\newcommand{\fa}{f_A}

\subsection{Sum of ``average'' internal strengths}\label{sec:fa}

Our first new value function is a re-scaled version of the sum over all communities of the internal strength of their nodes:
\begin{equation}\label{eq:def_fa}
\fa(G,\sigma) = \sum_{k=1\dots n_c}\frac{1}{R(\abs{c_k})}\sum_{i\in c_k}s_i^{int}=\sum_{k=1\dots n_c}\frac{1}{R(\abs{c_k})}\sum_{i\in c_k}\sum_{j\in c_k} W_{ij},
\end{equation} 
where $R$ is any nonnegative normalization function satisfying the following conditions: (i) $R$ is strictly increasing, and (ii) $R(a)/a$ is nonincreasing. The class of possible $R$ includes all functions $R(a) = a^\beta$ for $\beta \in (0,1]$. For $\beta=1$ we are just summing the \emph{average} internal strengths of the communities.

\begin{thm}\label{thm:fa}
The value function $\fa$ defined in \eqref{eq:def_fa} satisfies Properties  \ref{proper:continuity} (continuity), \ref{proper:scale} (scale-invariance), \ref{proper:locality} for $k=0$ hence also $k=1$, Property \ref{proper:richness_set} (richness) for the set $\Sigma$ of partitions with at most one singleton community, and Property \ref{proper:relative_consistent_improvement} (relative consistent improvement).
\end{thm}
More precisely, $\fa$ satisfies the stronger properties of value-scale invariance, value-consistent improvement and value-$0$-locality mentioned in Remark R1. The locality property may in fact be extended like in \cite{VanLaarhoven:2014} to graphs $G^{(1)}$ and $G^{(2)}$ with different node sets and to $\sigma$, $\sigma'$ both containing arbitrary partitions of $C_0$.
\begin{proof}
For a fixed $\sigma$, the value of $\fa(W,\sigma)$ is a linear combination with positive coefficients of the weights $W_{ij}$ of edges inside the communities. It is therefore clearly continuous with respect to $W$ and scale invariant, as $\fa(\alpha W,\sigma) = \alpha \fa(W,\sigma)$. 

Regarding consistent improvement, consider weights $W'$ for which $W_{ij}'\geq W_{ij}$ whenever $\sigma_i=\sigma_j$  and $W_{ij}'\leq W_{ij}$ otherwise. Then there holds
$$
\fa(W',\sigma) = \sum_{k=1\dots n_c}\frac{1}{R(\abs{c_k})}\sum_{i\in c_k}\sum_{j\in c_k}W'_{ij}\geq \sum_{k=1\dots n_c}\frac{1}{R(\abs{c_k})}\sum_{i\in c_k}\sum_{j\in c_k}W_{ij} = \fa(W,\sigma).
$$
This proves value-consistent improvement as defined in Remark R1, which is a sufficient condition for Property \ref{proper:relative_consistent_improvement}.

Regarding locality, it is straightforward to see that the following value-locality property from Remark R1 holds and is sufficient to imply $0$-locality:  Given two graphs $G_1(V,W^{(1)})$ and $G_2(V,W^{(2)})$ with $W^{(1)}_{ij} = W^{(2)}_{ij}$ for all $i,j \in C_0 \subset V$, take any partition $\sigma$ for which $C_0$ is a union of communities. Consider $\sigma'$ exactly equal to $\sigma$ except that $C_0$ might be partitioned differently into communities. Then 
$$\fa(G_1,\sigma)-\fa(G_1,\sigma') = \fa(G_2,\sigma) - \fa(G_2,\sigma') \, .$$

There remains to prove that $\fa$ is rich for the set $\Sigma$ of partitions containing no more than one singleton community. For any such partition $\sigma^*$, we define a graph consisting of disjoint cliques corresponding exactly to the communities: $W_{ij}=1$ if $\sigma^*_i=\sigma^*_j$ (and $i\neq j$) and $W_{ij}=0$ else. The strict optimality of $\sigma^*$ for this graph is shown in Appendix \ref{app:richness_fa}. By Lemma \ref{lem:positive_weights} this automatically implies that the partition $\Sigma^*$ is strictly optimal for an open set of values of $W$.
\end{proof}

In terms of satisfying the properties of Section \ref{sec:properties}, $f_A$ is at least as good as the modular density $f_{md}$; in particular, as $f_A$ only comprises a positive term for each edge present inside a community, it satisfies the stronger property of 0-locality in conjunction with all the other axioms. This might appear surprising, considering e.g.~Radicchi's proposal \cite{Radicchi2004} that a community is expected to have stronger inside links than external links (the latter does not appear in $f_A$ and seems to almost exclude $0$-locality). Now, the academic properties satisfied by $f_A$ do not guarantee its practical relevance, which we further investigate in Section \ref{sec:experiments}. The existence of $f_A$ probably hints at an insufficiency of the listed properties towards guaranteeing relevant quality functions for community detection.

\newcommand{\fb}{f_B}

\subsection{Penalizing external strength}

From an axiomatic viewpoint, the shortcoming of $\fa$ introduced in Section \ref{sec:fa} is that a partition with several singleton communities is never strictly optimal. It could however be relevant in practice to have several singleton communities. 
The intuitive reason for the restricted richness of $\fa$ is the following: suppose that we want nodes 1 and 2 to form two singleton communities in a partition $\sigma$. Their internal strength would by definition be 0, so the contribution of $\{1,2\} \subset V$ to $\fa(W,\sigma)$ would be 0. On the other hand, if we join them in a two-nodes community to get the partition $\sigma'$, then their internal strength would be $W_{12}$, and the contribution of $\{1,2\} \subset V$ to $\fa(W,\sigma')$ would be $2W_{12}$. Hence we get $\fa(W,\sigma') \geq\fa(W,\sigma)$ for any $W$, such that $\sigma$ cannot be strictly optimal.

To palliate this issue, we add in this section an incentive to keep certain weakly connected nodes in separate singleton communities. More precisely, we will penalize every edge connecting a node $i$ to other communities, with a penalization weight that depends on the size of the community to which $i$ belongs and that happens to be 0 if the community only contains $i$.

Formally,
for some fixed parameter $\alpha >0$, we define the value function
\begin{align}\label{eq:def_fb}
\fb(W,\sigma) &= \sum_{k}\frac{1}{R(\abs{c_k})} \sum_{i\in c_k} \prt{s_i^{int}-\alpha(\abs{c_k}-1)s_i^{ext}}\\
& = \sum_{k}\frac{1}{R(\abs{c_k})} \sum_{i\in c_k} \prt{\sum_{j\in c_k}W_{ij}-\alpha(\abs{c_k}-1)\sum_{j\not\in c_k}W_{ij}},
\end{align}
where we remind that the normalization function is any positive increasing function for which $R(a)/a$ is nonincreasing.

The next theorem shows that the particular penalization weight $\alpha(\abs{c_k}-1)$ allows $\fb$ to be richer over a larger set $\Sigma$ than $\fa$, and than the modular density $f_{md}$ which is similar to $\fb$ but just with uniform penalization weight 1. With respect to $\fa$ this richness is however achieved at a cost, since $\fb$ only satisfies $1$-locality (as $f_{md}$ does), while $\fa$ satisfies $0$-locality. By Theorem \ref{thm:imposs_cont_scale_rich_relimpr}, irrespective of locality concerns, a further extension of the richness set $\Sigma$ is not possible unless other properties are further relaxed.
 
\begin{thm}\label{thm:fb}
The value function $\fb$ defined in \eqref{eq:def_fb} satisfies Properties \ref{proper:continuity} (continuity), \ref{proper:scale} (scale-invariance), \ref{proper:locality} (locality) with $k=1$, Property \ref{proper:richness_set} (richness) with respect to the set $\Sigma$ of partitions containing at least one community with more than one node, and Property \ref{proper:relative_consistent_improvement} (relative consistent improvement).
\end{thm}
More precisely, $\fa$ satisfies the stronger properties of value-scale invariance, value-consistent improvement and value-$1$-locality mentioned in Remark R1. The locality property may in fact be extended like in \cite{VanLaarhoven:2014} to graphs $G^{(1)}$ and $G^{(2)}$ with different node sets and to $\sigma$, $\sigma'$ both containing arbitrary partitions of $C_0$.

\begin{proof}
Once the partition $\sigma$ is fixed, $f(W,\sigma)$ is a linear function of the weights $W_{ij}$, it is hence trivially continuous in $W$ and scale-invariant as $\fb(\alpha W,\sigma) = \alpha \fb(W,\sigma)$.

Regarding consistent improvement, consider weights $W'$ for which $W_{ij}'\geq W_{ij}$ for all $i,j$ in the same communities and $W_{ij}'\leq W_{ij}$ for other pairs $i,j$. Then for all $i$ we have
\begin{eqnarray*}
\sum_{i \in C_k} \sum_{j \in C_k} W'_{ij} & \geq & \sum_{i \in C_k} \sum_{j \in C_k} W_{ij} \;\; \text{ and} \\
\sum_{i \in C_k} \sum_{j \not\in C_k} W'_{ij} & \geq & \sum_{i \in C_k} \sum_{j \not\in C_k} W_{ij} \;\; .
\end{eqnarray*}
Referring to the definition \eqref{eq:def_fb} this straightforwardly proves value-consistent improvement as defined in Remark R1, which is a sufficient condition for Property \ref{proper:relative_consistent_improvement}.

Regarding locality, it is straightforward to see that the following value-locality property from Remark R1 holds and is sufficient to imply $1$-locality:  Given two graphs $G_1(V,W^{(1)})$ and $G_2(V,W^{(2)})$ with $W^{(1)}_{ij} = W^{(2)}_{ij}$ for all $i,j \in C_0 \subset V$ and for all $i,j$ for which $i$ or $j$ belongs to $C_0$, take any partition $\sigma$ for which $C_0$ is a union of communities. Consider $\sigma'$ exactly equal to $\sigma$ except that $C_0$ might be partitioned differently into communities. Then 
$$\fb(G_1,\sigma)-\fb(G_1,\sigma') = \fb(G_2,\sigma) - \fb(G_2,\sigma') \, .$$

There remains to prove the richness property of $\fb$. We do this in two steps. For a given partition $\sigma$ of the node set $V$, let $V^* \subset V$ the nodes which do not belong to a singleton partition in $\sigma$ and $\sigma^*$ the partition of $V^*$ corresponding to $\sigma$. First, we prove that the same construction as for $\fa$, applied to the weights $W^*$ between nodes $V^*$, makes the partition $\sigma^*$ of $V^*$ optimal for $\fb$, i.e.~: for any $i \in V^*$, let $W^*_{ij}=1$ if $\sigma^*_i = \sigma^*_j$ and $W_{ij}=\delta$ otherwise, with small $\delta>0$ selected by continuity. Second, we add (if necessary) the nodes of $V\setminus V^*$ to this construction by taking $W_{i,j} = \varepsilon$ for all $i \in V\setminus V^*$ and we show that for $\varepsilon >0$ sufficiently small this $W$ makes $\sigma$ strictly optimal. Details are given in Appendix \ref{app:richness_fb}.
\end{proof}

\noindent \textbf{Remark:} Before moving to further considerations, we must mention that \emph{both $f_A$ and $f_B$ would satisfy all the axioms, including richness for all possible partitions of the graph nodes, if the graph was allowed to contain self-loops.} Linearity, scale-invariance and locality indeed remain trivially true if self-loops are added into $W$. Consistent improvement also still holds with the same proof. Regarding richness, we know that particular graphs, in which all self-loops have zero weight, already allow to make a large set of partitions strictly optimal. To make a partition with several isolated nodes optimal, it suffices to construct the optimal weights for the partition without those nodes, and then complete the graph by adding those nodes with each a strong self-loop and a very weak connection to any other node in the graph.

\newcommand{\fc}{f_C}

\subsection{Relaxing consistent improvement}

In Theorem \ref{thm:imposs_cont_sloc_scale_rich} we have shown that even if consistent improvement is dropped,  
then it is still not possible to satisfy all the other criteria including $0$-locality and richness with respect to all partitions. If richness is relaxed, then $\fa$ shows how also consistent improvement can be included. As a complement, the following value function shows that it is possible to be rich with respect to all partitions and satisfy $1$-locality instead of $0$-locality: 
$$\fc(W,\sigma) = \sum_k \, \frac{1}{\vert c_k \vert} \, \left( \left( \sum_{i \in c_k} \, s_i^{in} \right) - \alpha \vert c_k \vert (\vert c_k \vert - 1) \min_{j \in c_k,\, l \in V} W_{lj} \, \right) \, .$$

One can show that \emph{$\fc$ satisfies properties \ref{proper:continuity} (continuity), \ref{proper:locality} with $k=1$($1$-locality), \ref{proper:scale} (scale invariance) and \ref{proper:richness_set} (richness) with respect to all partitions provided $\alpha>1$, while it then (unavoidably, from Theorem \ref{thm:imposs_cont_scale_rich_relimpr}) fails to satisfy Property \ref{proper:relative_consistent_improvement} (relative consistent improvement).}

We do not provide a formal proof here because, before claiming any usefulness for $\fc$, it should be shown to satisfy some accepted criterion which would distinguish community partitions from arbitrary partitions; in the absence of such property, $\fc$ is left with the status of curiosity.

Accordingly, the main idea in the construction of $\fc$, namely enlarging the penalty term to the minimum over \emph{all} edges connected to the community, both internal and external, is motivated by technical richness arguments and not by intuitive community characteristics. Essentially, we observe that $\fb(W,\sigma^s) = \fc(W,\sigma^s) = 0$ for all $W$ for $\sigma^s$ the all-singletons-partition. However, thanks to the modified penalty in $\fc$, with $W_{ij}=1$ for all $i,j \in V$ we have $\fc(\sigma,W) = \sum_{k} (\abs{c_k}-1) (1-\alpha) < 0$ for all $\sigma \neq \sigma_s$, hence the strict optimality of the all-singletons partition in that case.

\subsection*{Invariance to symmetries}

To conclude this Section, let us mention that the value functions $\fa$, $\fb$ and $\fc$ are all independent of node labels, see Remark R2. They hence trivially satisfy permutation invariance as defined in \cite{VanLaarhoven:2014}, i.e.: if $G$ features some isometries then defining $\sigma'$ by applying one of the associated isometries to an initial partition $\sigma$ implies $f(G,\sigma') = f(G,\sigma)$.\\

The following Section experimentally investigates how useful the proposed value function $\fa$ and $\fb$ might be in practical clustering applications.

\section{Experiments}\label{sec:experiments}

\subsection*{Benchmark and methods}

We have investigated the behavior of our value functions by optimizing them for typical graphs generated using the popular LFR benchmark model~\cite{Lancichinetti2008Benchmark}. This planted partitions method first defines ``prior communities'' on a set of nodes, and then generates a graph with edges preferentially inside these prior communities. It further ensures power-law degree and community size distributions consistent with realistic networks. We have considered graphs with $N=1000$ nodes, an average degree $\overline{k}=25$ and community size in $\left[10,100\right]$. Two additional parameters allow tuning the significance of the prior communities and therefore the difficulty of extracting them: $\mu_w$ the expected proportion of the strength of each node connecting it to nodes outside its prior community, and $\mu_t$ the expected proportion of edges (without taking their weight into account) connecting a node to nodes outside its prior community. 

We optimize our values functions $f_A$ and $f_B$ with $R(x) = x^\beta$, for different values of $\beta \in [0,1]$. For this purpose, we use the local optimization algorithm introduced in~\cite{browet2013community}. The latter requires self-loops at intermediate steps, which our value functions accommodate naturally, although our theoretical analysis is made without self-loops.

We evaluate the partitions extracted by optimizing our quality functions in several ways. First, we compute their normalized mutual information (NMI)~\cite{Danon2005} to the ``ground truth'', which we define as being the planted partition drawn by the LFR benchmark model to generate the graph. Recall that this ``ground truth'' is only meaningful when the method wires, on average, more edges inside the communities, i.e. $\mu_t \ll 1$, and assigns significantly larger weight to edges inside communities, i.e. for $\mu_w \leq \mu_t$; strictly speaking it is never 100\% ``true'' except in trivial cases. In fact we are precisely seeking to better \emph{define} communities in (almost all) non-trivial cases, and we are optimizing cost functions locally. So, although the quantitative results do most often agree with intuitive decisions about ``visually reasonable or unreasonable'' partitions, this comparison to ground truth must be taken with the usual grain of salt. Therefore, in complement, we have computed two indicators of stability for the partitions, namely the number of communities found $n_c$ and the sum-of-squares of the sizes of those communities $\sum_k\left|c_k\right|^2$.

\subsection*{Results}

\input{figs/plot.tex}

Figure~\ref{fig:results} presents the results for $f_A$ as a function of $\beta$ and compares it to those obtained with the well-established modularity  criterion\cite{Newman2004}.

For easy to detect communities ($\mu_t=0.3$ and $\mu_w=0.1$, first column in Figure~\ref{fig:results}), sub-partitions of the planted communities are almost always found except for very small values of $\beta$. The planted communities are moreover exactly recovered for a wide range of values of $\beta$ smaller than 0.5. Selecting any value of $\beta$ should thus be expected to yield good results in such cases.

Reasonable results are also observed for planted communities of ``medium'' difficulty ($\mu_t=0.6$ and $\mu_w=0.3$, second column of Figure~\ref{fig:results}). The extracted communities do not correspond exactly to the planted one, but the results outperform those obtained with the modularity criterion (\cite{Newman2004}, see table \ref{table:costfunction}) for almost all values of $\beta$. This could be explained by the fact that modularity is known to hit a ``resolution limit'' which prevents it from correctly extracting communities in this situation. Our cost function $f_A$ apparently does not suffer from this problem.

The limitations of $f_A$ appear more clearly for graphs with harder-to-find communities ($\mu_t=0.8$ and $\mu_w=0.6$, third column of Figure~\ref{fig:results} ). 
Only a very narrow range of $\beta$ values give reasonable and more or less stable results -- i.e.~an a priori selected value of $\beta$ would almost surely fail. This is complementary to the modularity, which precisely handles this situation well.

In conclusion, the function $f_A$ which was proposed on purely academic grounds appears to give reasonable partitions in many cases, but of course not always. An issue is the choice of $\beta$, which can strongly affect the quality of the partitions produced. Acceptable results appear not too difficult to obtain by scanning different values and analyzing the evolution of certain stability measures, i.e. selecting for example values of $\beta$ for which these measures appear ``reasonable'' and stable with respect to small variations of $\beta$. 
Such strategies would however not be compatible with our axiomatic approach. Indeed, when $\beta$ is not fixed a priori, the property of locality would almost surely be lost, and consistent improvement would not be a priori guaranteed either (those are the two properties where different graphs must be compared). There might be ``soft'' tuning strategies which preserve the axiomatic properties, e.g.~letting $\beta$ depend on the number of nodes or replacing $R(|c_k|)$ by some slightly more complicated dependence function. 
Such explorations however go beyond the main purpose of the present paper.

Similar observations apply to $f_B$, see Figure \ref{fig:results_B}. Reasonable results are obtained for appropriate tuning of \emph{both} $\alpha$ and $\beta$, but the ``good'' tuning values depend on the graph. 
Since such tuning departs from our axiomatic approach, we do not include more extensive experimental results about $f_B$ in this paper.

\input{figs/plot_B.tex}

\section{Conclusion}

We have considered an axiomatic approach towards defining communities as the optimum of a value function. We show that without relying on self-loops, it is not possible to satisfy a complete set of standard properties.
Our main message is that this impossibility remains even when replacing Kleinberg's very strong consistent improvement requirement by a weaker and seemingly more natural form of consistency. 
We further show, by explicit construction, that by slightly restricting the set of a priori expected communities --- e.g.~excluding the case of all isolated nodes --- it becomes possible to satisfy all the axioms, and reach in some situations a performance comparable to modularity. Furthermore, our constructions do satisfy the complete set of axioms when self-loops can be used to satisfy richness. These two points clarify precisely how some previous papers were able to circumvent Kleinberg's impossibility result.

This points towards several options for future research. First comes the necessity to select a subset of the standard axioms; excluding a set of a priori uninteresting partitions has been identified as one economical way to do this, but there might be others. Second and maybe more importantly, while our experimental investigations show good performance for certain parameter values, they also return extremely poor partitions for other parameter values, even though the corresponding functions still satisfy all the axioms. This demonstrates that our set of axioms is by no means sufficient to single out useful value functions, and new properties should therefore be defined. A third direction would be to extend or circumvent our impossibility result in a framework departing from value functions.

Finally, the experimental investigations with our simple value functions provide results that might be of independent interest for community partitioning. Indeed, despite their simplicity, our value functions outperform modularity in several cases and for well-chosen parameter values. This suggests an adaptive tuning procedure for the value function parameters. Note that the axioms would no longer necessarily be satisfied when introducing such adaptive tuning.

\acknowledgements{This work is supported by the  DYSCO network (Dynamical Systems, Control, and Optimization), funded by the Interuniversity Attraction Poles Programme, initiated by the Belgian Federal Science Policy. 

The authors want to thank Vincent Traag, Vincent Blondel and Renaud Lambiotte for occasional but very helpful discussions on this topic, in addition to drawing our attention to it and providing useful references. This paper presents research results of the Belgian Network DYSCO (Dynamical Systems, Control, and Optimization), funded by the Inter-university Attraction Poles program, initiated by the Belgian State, Science Policy Office. QUANTIC is a joint project team involving INRIA Paris (host institution), Mines ParisTech, Ecole Normale Superieure Paris, Universit\'e Pierre \& Marie Curie Paris and CNRS.
}


\input{boundary_communities_BHS.bbl}
\appendix

\section{Proof of Lemma \ref{lem:positive_weights}} \label{app:proof_positive_weights}

Remember that there are only finitely many possible partitions. 
Let $\sigma^{(2)}$ be (one of) the second best partition(s) for the graph $G(V,W)$, that is, $f(G,\sigma^{(2)})\geq f(G,\sigma')$ for every $\sigma'\neq \sigma$, and let $\epsilon =\frac{1}{2}\prt{f(G,\sigma)- f(G,\sigma^{(2)})}$. By continuity, for every $\sigma'$ we can pick a $\delta'>0$ such that $\norm{W-W'}<\delta'$ implies 
$\abs{\prt{f(G(V,W'),\sigma')- f(G(V,W),\sigma')}}<\epsilon$. Let $\delta$ be the minimum of these $\delta'$ over all partitions, which is positive and well defined since there are finitely many of them. It follows that $\norm{W-W'}<\delta'$ implies $\abs{\prt{f(G(V,W'),\sigma')- f(G(V,W),\sigma')}}<\epsilon$ for all $\sigma'$. Using the definition of $\epsilon$, we have then for every partition $\sigma'$
$$
f(W',\sigma') < f(W,\sigma')+\epsilon \leq f(W,\sigma^{(2)})+ \epsilon = f(W,\sigma) - \epsilon < f(W',\sigma),
$$
which establishes the result.

\section{Richness of $\fa$ in Theorem \ref{thm:fa}}
\label{app:richness_fa}

For $\sigma^*$ a given partition, we consider the graph $G(V,W^*)$ where $W^*_{ij} = 1$ if $\sigma^*_i = \sigma^*_j$ and $W^*_{ij}=0$ otherwise, i.e. the graph where there is one clique for each community of $\sigma^*$. We will prove that $\sigma^*$ is the strictly optimal partition of $G(V,W^*)$ according to $\fa$. The proof uses two claims.\\

\emph{Claim 1: Consider a community $c_0$ of some arbitrary partition $\sigma$, and $\sigma'$ a partition identical to $\sigma$ except that $c_0$ is split into two communities $c_{01}$ and $c_{02}$. If $W^*_{ij} = 0$ for all $i \in c_{01}$ and $j\in c_{02}$ (i.e. no clique contains a node from \emph{both} communities)
then $\fa(\sigma',W^*) \geq \fa(\sigma,W^*)$, with equality holding if and only if $\sum_{i \in c_0} \sum_{j \in c_0} W_{ij} = 0$.}

\begin{proof} (of Claim 1):
Define $\bar s_k^{(int)} = \frac{1}{R(\abs{c_k})}\sum_{i\in c_k} \sum_{j\in c_k} W^*_{ij}$ the normalized internal strength of a community $c_k$. Then we have by definition
$$\fa(\sigma,W^*) - \bar s_0^{(int)} =  \fa(\sigma',W^*) - \bar s_{01}^{(int)} - \bar s_{02}^{(int)} \, .$$
If $W^*_{ij} = 0$ for all $i \in c_{01}$ and $j\in c_{02}$ then we have
$\sum_{j \in c_{0a}} W_{ij} = \sum_{j \in c_0} W_{ij}$ for all $i$ in $c_{0a}$ and $a \in \{ 0,1\}$. Therefore, $R(\abs{c_0}) \bar s_0^{(int)} = R(\abs{c_{01}}) \bar s_{01}^{(int)} + R(\abs{c_{02}}) \bar s_{02}^{(int)}$. From there we have
\begin{equation}\label{eq:delta_average_intern_proof_fa}
\bar s_{0}^{(int)} 
= \frac{R(\abs{c_{01}})}{R(\abs{c_{0}})}\bar s_{01}^{(int)}  +\frac{R(\abs{c_{02}})}{R(\abs{c_{0}})} \bar s_{02}^{(int)} \leq \bar  s_{01}^{(int)} + \bar s_{02}^{(int)},
\end{equation}
where we have used 
$\abs{c_0}  > \max\{\abs{c_{01}}, \abs{c_{02}} \}$ and the fact that $R$ is increasing, so that $\max\{ \frac{R(\abs{c_{01}})}{R(\abs{c_{0}})}, \frac{R(\abs{c_{02}})}{R(\abs{c_{0}})} \} < 1$. It follows thus also that the  inequality in \eqref{eq:delta_average_intern_proof_fa} is strict unless $\bar s_{01}^{(int)} = \bar s_{02}^{(int)} = 0$, that is, $W_{ij} = 0$ for all $i,j \in c_0$.\\
\end{proof}

\emph{Claim 2: Consider a partition $\sigma$ where $W^*_{ij} = 1$ for all $i,j \in c_{0}$ and define a partition $\sigma'$ identical to $\sigma$ except that $c_0$ is split into two communities $c_{01}$ and $c_{02}$. Then $\fa(\sigma,W^*) > \fa(\sigma',W^*)$.}

\begin{proof} (of Claim 2):
With the same notation as in the proof of Claim 1, we still have
$$\fa(\sigma,W^*) - \bar s_0^{(int)} =  \fa(\sigma',W^*) - \bar s_{01}^{(int)} - \bar s_{02}^{(int)} \, .$$
Since $W^*_{ij} = 1$ for all $i,j \in c_{01} \cup c_{02} = c_0$ we readily get $\bar s_{01}^{(int)}= \frac{\abs{c_{01}}}{R(\abs{c_{01}})} (\abs{c_{01}} -1)$, $\bar s_{02}^{(int)}=\frac{\abs{c_{02}}}{R(\abs{c_{02}})}(\abs{c_{02}} -1)$, and 
$$\bar s_0^{(int)} = \frac{\abs{c_{0}}}{R(\abs{c_{0}})}(\abs{c_0}-1) = \frac{\abs{c_{0}}}{R(\abs{c_{0}})}( \abs{c_{01}}+\abs{c_{02}} -1).$$
Remembering that $R(a)/a$ is nonincreasing and that $\abs{c_0}= \abs{c_{01}}+\abs{c_{02}}$, we obtain
$$
\bar s_{01}^{(int)} + \bar s_{02}^{(int)}\leq \frac{\abs{c_{0}}}{R(\abs{c_{0}})}(\abs{c_{01}} -1) + \frac{\abs{c_{0}}}{R(\abs{c_{0}})}(\abs{c_{02}} -1) < \frac{\abs{c_{0}}}{R(\abs{c_{0}})}( \abs{c_{01}}+\abs{c_{02}} -1)=\bar s_0^{(int)},
$$
so that  $\fa(\sigma,W^*) > \fa(\sigma',W^*)$.
\end{proof}

We can now prove the strict optimality of $\sigma^*$. Consider a partition $\sigma^{(1)}$ and denote its communities by $c^{(1)}_{k}$. Define $\sigma^{(2)}$ the partition where each $c^{(1)}_k$ has been further partitioned into the largest possible communities $c^{(2)}_{k1}$, $c^{(2)}_{k2}$,... $c^{(2)}_{kn_k}$ for which $W^*_{ij}=1$ whenever $\sigma^{(2)}_i = \sigma^{(2)}_j$. Note that as we require the communities to remain as large as possible, we also have $W^*_{ij}=0$ if $\sigma^{(1)}_i = \sigma^{(1)}_j$ but $\sigma^{(2)}_i \neq \sigma^{(2)}_j$. Therefore the transition from $\sigma^{(1)}$ to $\sigma^{(2)}$ can be obtained by repeated application of a transition from $\sigma$ to $\sigma'$ as set out in Claim 1. Thus $\fa(\sigma^{(2)},W^*) \geq \fa(\sigma^{(1)},W^*)$.

From $\sigma^{(2)}$, we can obtain the original partition $\sigma^*$ by repeated application of a transition from $\sigma'$ to $\sigma$ (note the reversed order!) as set out in Claim 2. We hence have $\fa(\sigma^*,W^*) \geq \fa(\sigma^{(2)},W^*)$.

Hence overall $\fa(\sigma^*,W^*) \geq \fa(\sigma^{(2)},W^*) \geq \fa(\sigma^{(1)},W^*)$ for any $\sigma^{(1)}$. Let us now examine in which cases we can have equality. To have $\fa(\sigma^*,W^*) = \fa(\sigma^{(2)},W^*)$ we need $\sigma^{(2)} = \sigma^*$. To have $\fa(\sigma^{(2)},W^*) = \fa(\sigma^{(1)},W^*)$ we need for each $c^{(1)}_{k} \in \sigma^{(1)}$:
\begin{itemize}
\item[{-}] either $n_k=1$ i.e.~$c^{(1)}_{k}$ has not been further partitioned in the process of obtaining $\sigma^{(2)}$, which also means that $W^*_{ij}=1$ for all $i,j \in c^{(1)}_{k}$;
\item[{-}] or $\sum_{i \in c^{(1)}_{k}} \sum_{j \in c^{(1)}_{k}} W_{ij} = 0$, which also implies that the $c^{(2)}_{ki}$ are all singleton communities for all $i$.
\end{itemize}
The second situation is excluded if $\sigma^{(2)} = \sigma^*$ and $\sigma^*$ has no more than one singleton community. Hence we must be in the first situation for every $c^{(1)}_{k} \in \sigma^{(1)}$, which means that $\sigma^{(2)}$ is obtained from $\sigma^{(1)}$ by applying no modification, and then $\sigma^*$ is obtained from $\sigma^{(2)}$ by applying no modification; i.e.~that $\sigma^{(1)} = \sigma^*$. For any $\sigma^{(1)}\neq \sigma^*$ we hence have $\fa(\sigma^*,W^*) > \fa(\sigma^{(1)},W^*)$, proving the strict optimality of $\sigma^*$.

\section{Richness of $\fb$ in Theorem \ref{thm:fb}}
\label{app:richness_fb}

To establish the richness of $\fb$ with respect to all partitions containing at least one community with more than one node, we will use the following Lemma.

\begin{lem}\label{lem:fb+1}
Let $\sigma_*$ be a strictly optimal partition for $\fb$ for a graph $G(V,W)$ on $N$ nodes where $\sum_{j} W_{ij} > 0$ for all $i \in V$, i.e. every node is incident to at least one edge with positive weight. Build the graph $G^{+1}(V^{+1},W^{+1})$ by adding a node $N+1$ to $G$ and one single edge of weight $\epsilon$ connecting node $N+1$ to some node of $V$, and let $\sigma_*^{+1}$ be the partition of $V^{+1}$ obtained by adding the singleton community $\{N+1\}$ to $\sigma_*$.

Then, there exists an $\epsilon^*>0$ such that if $\epsilon < \epsilon^*$, then $\sigma_*^{+1}$ is a strictly optimal partition for $G^{+1}$. 
As a consequence, $\sigma_*^{+1}$ and $G^{+1}$ also satisfy the hypothesis of the present Lemma (and can thus play the role of $\sigma_*$ and $G$ in a further iteration of it). 
\end{lem}

\begin{proof} (of Lemma \ref{lem:fb+1})
Any partition $\sigma^{+1}$ of $G^{+1}$ can be obtained from some particular partition $\sigma$ of $G$ by either adding the node $N+1$ to $\sigma$ as a singleton community, or merging node $N+1$ into an existing community of $\sigma$. We start by examining these two cases assuming $\epsilon = 0$, i.e.~$W_{(N+1)j} = 0$ for all $j \in V$, while $s_i^{ext}+s_i^{int} = \sum_{j \in V} W_{ij} > 0$ for all $i \in V$.

First consider a partition $\sigma^{+1}$ of $G^{+1}$ which results from the addition of node $N+1$ to some community $c_k$ of $\sigma$. Since $W^{+1}$ just contains additional zero entries with respect to $W$ and $s_i^{ext}+s_i^{int} > 0$ for all $i \in V$, it is straightforward to observe that
\begin{align}
\fb(G^{+1},\sigma^{+1}) - \fb(G,\sigma) &= \prt{\frac{1}{R(\abs{c_k}+1)}-\frac{1}{R(\abs{c_k})}} \prt{\sum_{i\in c_k}s_i^{int}}  \nonumber \\ &+ \prt{\frac{\abs{c_k}}{R(\abs{c_k}+1)}-\frac{\abs{c_k}-1}{R(\abs{c_k})}} \prt{-\alpha \sum_{i\in c_k}s_i^{ext}}.
\label{eq:fb+1_comp}
\end{align}
(With a slight abuse of notation, in $G^{+1}$ we still denote $c_k$ the set of nodes that constituted the community $c_k$ in partition $\sigma$ of the graph $G$.) In the first term of \eqref{eq:fb+1_comp} we have $\frac{1}{R(\abs{c_k}+1)} < \frac{1}{R(\abs{c_k})}$ because $R$ is increasing. For the second term, observe that since $R(a)$ is increasing and $R(a)/a$ is nonincreasing, the following inequality holds for any $a>0$:
$$\frac{a-1}{R(a)} = \frac{a}{R(a)} -\frac{1}{R(a)}< \frac{a+1}{R(a+1)} -\frac{1}{R(a+1)} = \frac{a}{R(a+1)}.$$
In particular, $\frac{\abs{c_k}}{R(\abs{c_k}+1)}>\frac{\abs{c_k}-1}{R(\abs{c_k})}$. Thus for $\sigma \neq \sigma_*$ we have $\fb(G^{+1},\sigma^{+1}) \leq \fb(G,\sigma) < \fb(G,\sigma_*) \, $.  For $\sigma=\sigma_*$, since we have excluded $\sum_{i\in c_k}s_i^{int} = \sum_{i\in c_k}s_i^{ext} = 0$, \eqref{eq:fb+1_comp} is strictly negative and thus in all cases we have $\fb(G^{+1},\sigma^{+1}) < \fb(G,\sigma_*)$.

Next consider the partition $\sigma^{+1}$ of $G^{+1}$ consisting of $\sigma$ plus the singleton community $\{ N+1 \}$, which would thus be $\sigma_*^{+1}$ if $\sigma=\sigma_*$. The definition of $\fb$ readily implies that $\fb(G^{+1},\sigma^{+1}) = \fb(G,\sigma)$. In particular, we $\fb(G^{+1},\sigma^{+1}) = \fb(G,\sigma) < \fb(G,\sigma_*) \, $ and thus again $\fb(G^{+1},\sigma^{+1}) < \fb(G,\sigma_*)$, whenever $\sigma \neq \sigma_*$. The only remaining case, for $\sigma = \sigma_*$, gives $\fb(G^{+1},\sigma_*^{+1}) = \fb(G,\sigma_*)$ and must thus be the optimal one. This establishes the result for $\epsilon = 0$.

The case of sufficiently small $\epsilon>0$ follows then directly from Lemma \ref{lem:positive_weights}.  
This ensures that node $N+1$ is also incident to an edge with positive weight, and thus $G^{+1}$ and $\sigma_*^{+1}$ also satisfy the assumptions of the Lemma.
\end{proof}

To prove that a target partition $\sigma$ of $V$ can be made optimal under $\fb$ for some well-chosen weights $W$, we build as for $\fa$ the graph with $W_{ij}=1$ if $\sigma_i=\sigma_j$, otherwise $W_{ij}=0$.
Define $G^*(V^*,W^*)$ and $\sigma_*$ respectively the graph and partition where all the nodes that form singleton communities in $\sigma$ have been dropped. This graph is non-empty because we have assumed $\sigma$ contains at least one community with more than one node.

With respect to $\fa$, the value function $\fb$ just adds a penalty for links between different communities, so that $\fb(W',\sigma') \leq \fa(W',\sigma')$ for any $W'$ and $\sigma'$. For our particular construction of $W^*$ associated to $\sigma^*$, we have $s_i^{ext} = 0$ for all $i$, hence $\fb(W^*,\sigma^*) = \fa(W^*,\sigma^*)$. Note that $\sigma^*$ contains no singleton community so we can apply the proof of Appendix \ref{app:richness_fa}, showing that $\sigma^*$ is strictly optimal for $W^*$ under $\fa$, i.e.~$\fa(W^*,\sigma^*) > \fa(W^*,\sigma')$ for all $\sigma' \neq \sigma^*$. We then get
$$\fb(W^*,\sigma^*) = \fa(W^*,\sigma^*) > \fa(W^*,\sigma') \geq \fb(W^*,\sigma') \,$$
for all $\sigma' \neq \sigma^*$, proving strict optimality of partition $\sigma^*$ for the graph weights $W^*$ also under $\fb$. By Lemma \ref{lem:positive_weights}, partition $\sigma^*$ will still be optimal under $\fb$ for a slightly modified $W^{*\epsilon}$ where $W^{*\epsilon}_{ij}>0$ for all $i,j \in V^*$.

The graph $G^{*\epsilon}(V^*,W^{*\epsilon})$ obtained in this way satisfies the assumptions of Lemma \ref{lem:fb+1} with optimal partition $\sigma^*$. 
The existence of a graph for which the target partition $\sigma$ is strictly optimal follows then from a repeated application of that Lemma, adding singleton communities one by one until $\sigma$ is reached.

\end{document}

%% file: declarations_for_revtex.tex
\usepackage{amssymb,amsmath}
\usepackage{amsfonts}
\usepackage[english]{babel}
\usepackage[latin1]{inputenc}
\usepackage{fancyhdr}
\usepackage{yfonts}
\usepackage{mathrsfs}
\usepackage{amsthm}

\usepackage{algorithm}
\floatname{algorithm}{Alg.}
\usepackage{algorithmic}
\usepackage{psfrag}

\providecommand{\com[2]}{\begin{tt}[#1: #2]\end{tt}}
\usepackage{color}

\newtheorem{thm}{Theorem}

\newtheorem{lem}{Lemma}

\newtheorem{proper}{Property}

\providecommand{\prt}[1]{\left( #1 \right)}

\providecommand{\abs[1]}{\left|#1\right|}
\providecommand{\norm[1]}{\abs{\abs{#1}}}%

 \providecommand{\quotes[1]}{``#1"}

\newcommand{\ie}{i.e.~}
\newcommand{\eg}{e.g.~}

%% file: figs/improvment.tex
 \def\fillcolorfirst{Paired-10-4!40}
 \def\fillcolorsecond{Paired-10-6!40}
\newcommand{\addgraphexample}[4]{
\node[node](1) at #1{$1$};
\node[node](3)[right = \nodeDistance of 1]{$3$};
\node[node](4)[below = \nodeDistance of 3]{$4$};
\node[node](2)[below =  \nodeDistance of 1]{$2$};
\draw[edge](1) to (3);
\draw[edge](1) to (2);
\draw[edge](1) to (4);
\draw[edge](2) to (3);
\draw[edge](2) to (4);
\draw[edge](3) to (4);
#2

\coordinate (mid3-4) at ($(3) ! .5 ! (4)$) ;
\coordinate (mid1-2) at ($(1) ! .5 ! (2)$) ;
\coordinate (mid) at ($(mid1-2) ! .5 ! (mid3-4)$) ;

\begin{pgfonlayer}{background}
\filldraw[fill=\fillcolorfirst, draw=none,thick,dashed] 
	($(1)+(0,\deltaFill)$) -- 
	($(3)+(0,\deltaFill)$) --  
	($(3)+(0, \deltaFill)$) arc (90:0: \deltaFill) --  
	($(4)+(\deltaFill,0)$) -- 
	($(4)+(\deltaFill,0)$) arc (0:-90: \deltaFill) --  
	($(2)+(0,-\deltaFill)$) -- 
	($(2)+(0,-\deltaFill)$)  arc (-90:-180: \deltaFill) --	
	($(1)+(-\deltaFill,0)$) -- 
	($(1)+(-\deltaFill,0)$)   arc (-180:-270: \deltaFill) --cycle ;
\end{pgfonlayer}

\node[font=\large](ll)[right = .5\nodeDistance of mid3-4]{$\boldsymbol{>}$};

\node[node](1p) [right = 1.1\nodeDistance of 3]{$1$};
\node[node](3p)[right = \nodeDistance of 1p]{$3$};
\node[node](4p)[below = \nodeDistance of 3p]{$4$};
\node[node](2p)[below =  \nodeDistance of 1p]{$2$};
\draw[edge](1p) to (3p);
\draw[edge](1p) to (2p);
\draw[edge](1p) to (4p);
\draw[edge](2p) to (3p);
\draw[edge](2p) to (4p);
\draw[edge](3p) to (4p);
#3

\begin{pgfonlayer}{background}
\filldraw[fill=\fillcolorsecond, draw=none,thick,dashed] 
	($(1p)+(\deltaFill,0)$) -- 
	($(2p)+(\deltaFill,0)$) --  
	($(2p)+(\deltaFill,0)$) arc (0:-180: \deltaFill) --  
	($(1p)+(-\deltaFill,0)$) -- 
	($(1p)+(-\deltaFill,0)$) arc (-180:-360: \deltaFill) --  cycle;

\filldraw[fill=\fillcolorsecond, draw=none,thick,dashed]
	($(3p)+(\deltaFill,0)$) -- 
	($(4p)+(\deltaFill,0)$) --  
	($(4p)+(\deltaFill,0)$) arc (0:-180: \deltaFill) --  
	($(3p)+(-\deltaFill,0)$) -- 
	($(3p)+(-\deltaFill,0)$) arc (-180:-360: \deltaFill) --  cycle;

\end{pgfonlayer}

\coordinate (mid3p-4p) at ($(3p) ! .5 ! (4p)$) ;
\coordinate (mid1p-2p) at ($(1p) ! .5 ! (2p)$) ;
\coordinate (midp) at ($(mid1p-2p) ! .5 ! (mid3p-4p)$) ;

\coordinate (title) at ($(3) ! .5 ! (1p)$) ;

#4
}

\input{figs/tikzlibrarycolorbrewer.code.tex}
\newlength{\nodeDistance}
\setlength{\nodeDistance}{25pt}
\newlength{\deltaFill}
\setlength{\deltaFill}{.45\nodeDistance}
\begin{tikzpicture}
\tikzstyle{node}=[circle,thick,draw=black!75,fill=Paired-10-2!20,minimum size=1mm, font=\scriptsize, inner sep=2pt]
\tikzstyle{edge}=[-, thick, >=stealth', shorten >=1pt, shorten <=1pt]
\tikzstyle{largeEdge}=[-, line width=3pt, >=stealth', shorten >=1pt, shorten <=1pt]

\addgraphexample{(0,0)}{}{}{}
\coordinate (arrow1) at ($ (mid3p-4p) + (\nodeDistance,0)$);
\coordinate (arrow2) at ($(arrow1) + (1,1)$);
\coordinate (arrow3)at ($(arrow1) + (1,-1)$);
\draw[edge, ->, line width=1.5pt, shorten <=0pt](arrow1) to (arrow2);
\draw[edge, ->, line width=1.5pt, shorten <=0pt](arrow1) to (arrow3);

\addgraphexample{(7,2)}
							 {\draw[largeEdge](1) to (2);\draw[largeEdge](3) to (4);}
							 {\draw[largeEdge](1p) to (2p);\draw[largeEdge](3p) to (4p);}
							 {\node[align=center, text width=5cm, font=\itshape] (t1) at ($(title)+(0,\nodeDistance)$) {imposed only by\\ absolute consistency};}
\addgraphexample{(7,-2.5)}
							 {\draw[largeEdge](1) to (3);\draw[largeEdge](2) to (4);}
							 {\draw[largeEdge](1p) to (3p);\draw[largeEdge](2p) to (4p);}
							 {\node[align=center, text width=7cm, font=\itshape] (t1) at ($(title)+(0,\nodeDistance)$) {imposed by relative\\ and absolute consistency};}

\end{tikzpicture}

%% file: figs/plot.tex
\input{figs/tikzlibrarycolorbrewer.code.tex}
\def\minipagewidth{.3\textwidth}
\def\colorfirstplot{Paired-6-4}
\def\colorsecondplot{Paired-6-6}

\newcommand\addnmiplot[4]{
\begin{tikzpicture}
\begin{axis}[
		enlarge x limits=false,
		xmin=0,
		xmax=1,
		xtick= {0, 0.5, 1},
		ymin=0.6,
		ymax=1.05,
		xlabel=$\beta$,
		scale only axis, 
		xlabel near ticks,
		ylabel near ticks,
		width=100pt, 
	#2		
	]
\addplot+[draw=Paired-6-2,
				mark=none,
				line width=1pt,
				style=solid]				
					table[x=beta,y=NMI] {#1};
#3;
			
#4			 ;

\addplot+[draw=Paired-6-2,
				mark=none,
				line width=1pt,
				style=dashed]				
					table[x=beta,y=NMIMOD] {#1};

\end{axis}	

\begin{axis}[
		yticklabel style = {color=\colorsecondplot},
		enlarge x limits=false,
		 axis y line*=right,
		xtick= {0, 0.5, 1},
		ytick={1},
		yticklabel={\phantom{400}},
		xmin=0,
		xmax=1,
		ymin=0,
		ymax=1,
		scale only axis, 
		width=100pt, 
	]

%
%

\end{axis}		
					
\end{tikzpicture}
}

\def\fakelegend{
\addplot+[draw=Paired-6-2,
				mark=none,
				line width=1pt,
				style=dashed]				
					coordinates {(-1,0)};		
\addlegendentry{$Q$ based partition}	
\addplot+[draw=\colorfirstplot,
				mark=none,
				line width=1pt,
				style=solid]				
					coordinates {(-1,0)};
\addplot+[draw=\colorsecondplot,
				mark=none,
				line width=1pt,
				style=solid]				
					coordinates {(-1,0)};		

\addplot+[draw=\colorfirstplot,
				mark=none,
				line width=1pt,
				style=dashed]				
					coordinates {(-1,0)};
\addplot+[draw=\colorsecondplot,
				mark=none,
				line width=1pt,
				style=dashed]				
					coordinates {(-1,0)};		

}

\newcommand\addstabilityplot[4]{
\begin{tikzpicture}
\begin{semilogyaxis}[
		yticklabel style = {color=\colorfirstplot},
		ytick style = {color=\colorfirstplot},
		enlarge x limits=false,
		 axis y line*=left,
		 xticklabel=\empty,
		xmin=0,
		xmax=1,
		ymin=1000,
		ymax=100000,
		scale only axis, 
		ylabel near ticks,
		width=100pt, 
	#2		
	]
	
\addlegendimage{
 legend image code/.code={%
 			\draw[line width=1pt, draw=\colorfirstplot, fill=\colorfirstplot] (0cm,0cm) -- (.6cm,0cm);
 			\node[draw=none](t) at (.7cm,0cm){/};
			\draw[line width=1pt, draw=\colorsecondplot, fill=\colorsecondplot] (0.8cm,0cm) -- (1.4cm,0cm);
        }
}
\addlegendentry{$f_A$ based partition}

\addlegendimage{
 legend image code/.code={%
 			\draw[dashed, line width=1pt, draw=\colorfirstplot, fill=\colorfirstplot] (0cm,0cm) -- (.6cm,0cm);
 			\node[draw=none](t) at (.7cm,0cm){/};
			\draw[dashed, line width=1pt, draw=\colorsecondplot, fill=\colorsecondplot] (0.8cm,0cm) -- (1.4cm,0cm);
        }
}
\addlegendentry{planted partition}		

\addplot+[draw=\colorfirstplot,
				mark=none,
				line width=1pt,
				style=solid]				
					table[x=beta,y=NC2] {#1};

\addplot+[draw=\colorfirstplot,
				mark=none,
				line width=1pt,
				style=dashed]				
					table[x=beta,y=NC2T] {#1};

#4;

\end{semilogyaxis}

\begin{axis}[
		yticklabel style = {color=\colorsecondplot},
		ylabel style = {color=\colorsecondplot},
		enlarge x limits=false,
		 axis y line*=right,
		xtick= {0, 0.5, 1},
		xmin=0,
		xmax=1,
		ymin=0,
		ymax=400,
		scale only axis, 
		xlabel near ticks,
		ylabel near ticks,
		width=100pt, 
#3
	]

\addplot+[draw=\colorsecondplot,
				mark=none,
				line width=1pt,
				style=solid]				
					table[x=beta,y=C] {#1};

\addplot+[draw=\colorsecondplot,
				mark=none,
				line width=1pt,
				style=dashed]				
					table[x=beta,y=CT] {#1};

\end{axis}

\end{tikzpicture}
}

\begin{figure}
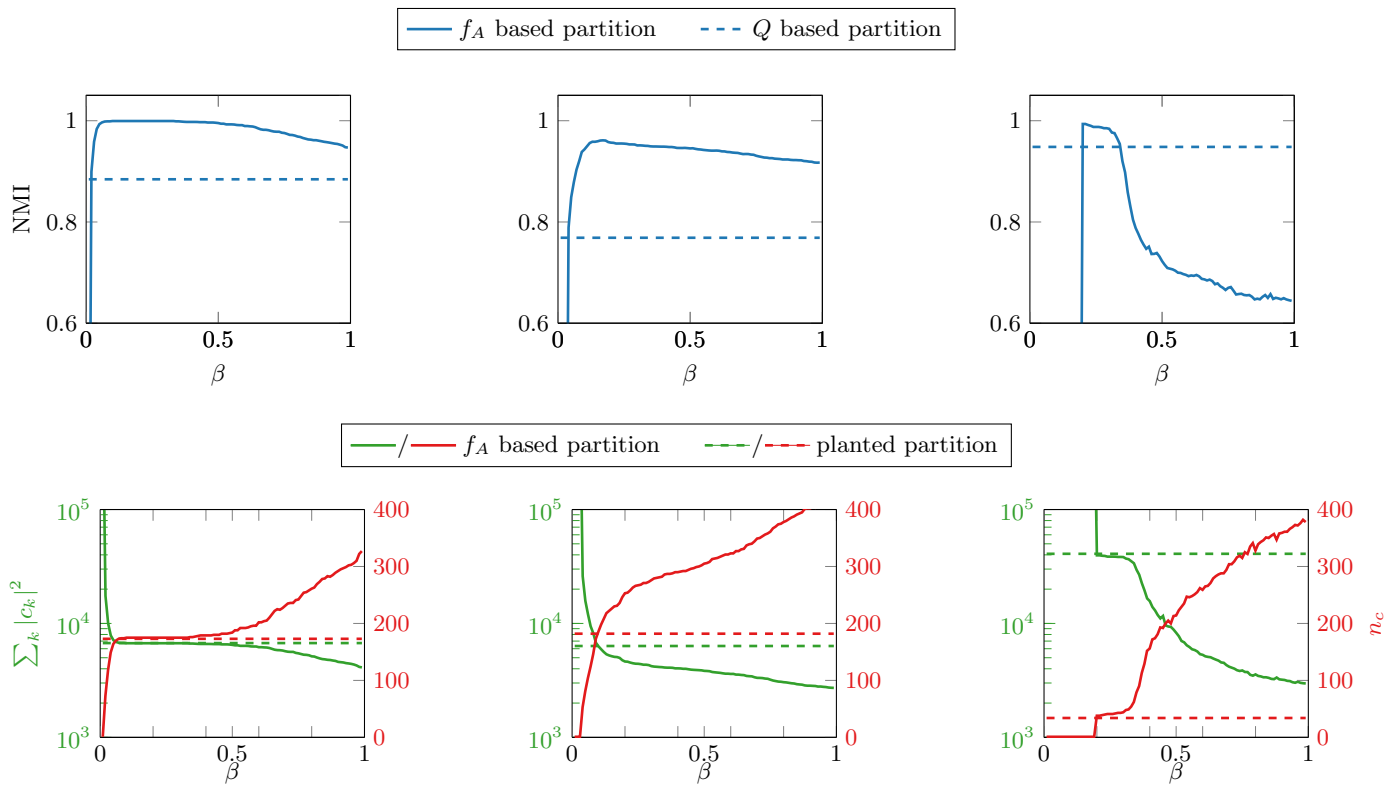

\centering
\pgfplotslegendfromname{legend_multi_plot}
\vskip10pt
\begin{minipage}[b]{\minipagewidth}
  \centering
  \addnmiplot {figs/G1.dat}{ylabel={\small {NMI}},xlabel =$\beta$, legend columns=4, legend style={/tikz/every even column/.append style={column sep=0.5cm}}, legend to name={legend_multi_plot} , }{\addlegendentry{$f_A$ based partition}}{\fakelegend}
  \end{minipage}
\hfill  
\begin{minipage}[b]{\minipagewidth}
  \centering
    \addnmiplot {figs/G2.dat}{ylabel=\phantom{\small {NMI}},xlabel =$\beta$}{}{}
\end{minipage}
  \hfill
  \begin{minipage}[b]{\minipagewidth}
  \centering
      \addnmiplot {figs/G3.dat}{ylabel=\phantom{\small {NMI}},xlabel =$\beta$}{}{}
\end{minipage}

\pgfplotslegendfromname{legend_second_plot}
\vskip8pt
\begin{minipage}[b]{\minipagewidth}
  \centering
  \addstabilityplot{figs/G1.dat}{ylabel={\color{\colorfirstplot}$\sum_k\left|c_k\right|^2$},xlabel =$\beta$, legend columns=4, legend style={/tikz/every even column/.append style={column sep=0.5cm}}, legend to name={legend_second_plot} }{}{}
\end{minipage}
\hfill
\begin{minipage}[b]{\minipagewidth}
  \centering
  \addstabilityplot{figs/G2.dat}{ylabel=\phantom{$\sum_k\left|c_k\right|^2$},xlabel =$\beta$, legend columns=4, legend style={/tikz/every even column/.append style={column sep=0.5cm}}, legend to name={legend_second_plot_I} }{}{}
\end{minipage}
\hfill
\begin{minipage}[b]{\minipagewidth}
  \centering
  \addstabilityplot{figs/G3.dat}{ylabel=\phantom{$\sum_k\left|c_k\right|^2$},xlabel =$\beta$, legend columns=4, legend style={/tikz/every even column/.append style={column sep=0.5cm}}, legend to name={legend_second_plot_II} }{ylabel={\color{\colorsecondplot}$n_c$}}{}
\end{minipage}

\caption{(color online) Quality and stability of the partitions returned by $f_A$ as a function of $\beta$ (continuous curve). The plots on the first row show the normalized mutual information (NMI) of the partitions extracted using $f_A$ with respect to the ``ground truth'' partition planted by the LFR benchmark model (solid line), compared to that obtained with the modularity (dashed line). The second row shows the number of communities (red, increasing with $\beta$) and their sum-of-square size (green, decreasing in $\beta$). The value of those criteria for the planted partition is also provided (dashed line). The difficulty of extracting communities increases from left to right (different $\mu_t$ and $\mu_w$ values in the LFR graph generator, see main text). For easy (left column) and medium (middle column) difficulty graphs, $f_A$ consistently outperforms modularity. For difficult graphs (right column), it is necessary to tune $\beta$ in a precise range in order to obtain, with our function $f_A$, results comparable to the modularity. This good range of $\beta$ values depends on the graph (not shown), which shows a limitation of our value function.}

\label{fig:results}
\end{figure}

%% file: figs/tikzlibrarycolorbrewer.code.tex
\definecolor{YlGn-3-1}{RGB}{247,252,185}
\definecolor{YlGn-3-C}{RGB}{247,252,185}
\definecolor{YlGn-3-2}{RGB}{173,221,142}
\definecolor{YlGn-3-F}{RGB}{173,221,142}
\definecolor{YlGn-3-3}{RGB}{49,163,84}
\definecolor{YlGn-3-I}{RGB}{49,163,84}
\definecolor{YlGn-4-1}{RGB}{255,255,204}
\definecolor{YlGn-4-B}{RGB}{255,255,204}
\definecolor{YlGn-4-2}{RGB}{194,230,153}
\definecolor{YlGn-4-E}{RGB}{194,230,153}
\definecolor{YlGn-4-3}{RGB}{120,198,121}
\definecolor{YlGn-4-G}{RGB}{120,198,121}
\definecolor{YlGn-4-4}{RGB}{35,132,67}
\definecolor{YlGn-4-J}{RGB}{35,132,67}
\definecolor{YlGn-5-1}{RGB}{255,255,204}
\definecolor{YlGn-5-B}{RGB}{255,255,204}
\definecolor{YlGn-5-2}{RGB}{194,230,153}
\definecolor{YlGn-5-E}{RGB}{194,230,153}
\definecolor{YlGn-5-3}{RGB}{120,198,121}
\definecolor{YlGn-5-G}{RGB}{120,198,121}
\definecolor{YlGn-5-4}{RGB}{49,163,84}
\definecolor{YlGn-5-I}{RGB}{49,163,84}
\definecolor{YlGn-5-5}{RGB}{0,104,55}
\definecolor{YlGn-5-K}{RGB}{0,104,55}
\definecolor{YlGn-6-1}{RGB}{255,255,204}
\definecolor{YlGn-6-B}{RGB}{255,255,204}
\definecolor{YlGn-6-2}{RGB}{217,240,163}
\definecolor{YlGn-6-D}{RGB}{217,240,163}
\definecolor{YlGn-6-3}{RGB}{173,221,142}
\definecolor{YlGn-6-F}{RGB}{173,221,142}
\definecolor{YlGn-6-4}{RGB}{120,198,121}
\definecolor{YlGn-6-G}{RGB}{120,198,121}
\definecolor{YlGn-6-5}{RGB}{49,163,84}
\definecolor{YlGn-6-I}{RGB}{49,163,84}
\definecolor{YlGn-6-6}{RGB}{0,104,55}
\definecolor{YlGn-6-K}{RGB}{0,104,55}
\definecolor{YlGn-7-1}{RGB}{255,255,204}
\definecolor{YlGn-7-B}{RGB}{255,255,204}
\definecolor{YlGn-7-2}{RGB}{217,240,163}
\definecolor{YlGn-7-D}{RGB}{217,240,163}
\definecolor{YlGn-7-3}{RGB}{173,221,142}
\definecolor{YlGn-7-F}{RGB}{173,221,142}
\definecolor{YlGn-7-4}{RGB}{120,198,121}
\definecolor{YlGn-7-G}{RGB}{120,198,121}
\definecolor{YlGn-7-5}{RGB}{65,171,93}
\definecolor{YlGn-7-H}{RGB}{65,171,93}
\definecolor{YlGn-7-6}{RGB}{35,132,67}
\definecolor{YlGn-7-J}{RGB}{35,132,67}
\definecolor{YlGn-7-7}{RGB}{0,90,50}
\definecolor{YlGn-7-L}{RGB}{0,90,50}
\definecolor{YlGn-8-1}{RGB}{255,255,229}
\definecolor{YlGn-8-A}{RGB}{255,255,229}
\definecolor{YlGn-8-2}{RGB}{247,252,185}
\definecolor{YlGn-8-C}{RGB}{247,252,185}
\definecolor{YlGn-8-3}{RGB}{217,240,163}
\definecolor{YlGn-8-D}{RGB}{217,240,163}
\definecolor{YlGn-8-4}{RGB}{173,221,142}
\definecolor{YlGn-8-F}{RGB}{173,221,142}
\definecolor{YlGn-8-5}{RGB}{120,198,121}
\definecolor{YlGn-8-G}{RGB}{120,198,121}
\definecolor{YlGn-8-6}{RGB}{65,171,93}
\definecolor{YlGn-8-H}{RGB}{65,171,93}
\definecolor{YlGn-8-7}{RGB}{35,132,67}
\definecolor{YlGn-8-J}{RGB}{35,132,67}
\definecolor{YlGn-8-8}{RGB}{0,90,50}
\definecolor{YlGn-8-L}{RGB}{0,90,50}
\definecolor{YlGn-9-1}{RGB}{255,255,229}
\definecolor{YlGn-9-A}{RGB}{255,255,229}
\definecolor{YlGn-9-2}{RGB}{247,252,185}
\definecolor{YlGn-9-C}{RGB}{247,252,185}
\definecolor{YlGn-9-3}{RGB}{217,240,163}
\definecolor{YlGn-9-D}{RGB}{217,240,163}
\definecolor{YlGn-9-4}{RGB}{173,221,142}
\definecolor{YlGn-9-F}{RGB}{173,221,142}
\definecolor{YlGn-9-5}{RGB}{120,198,121}
\definecolor{YlGn-9-G}{RGB}{120,198,121}
\definecolor{YlGn-9-6}{RGB}{65,171,93}
\definecolor{YlGn-9-H}{RGB}{65,171,93}
\definecolor{YlGn-9-7}{RGB}{35,132,67}
\definecolor{YlGn-9-J}{RGB}{35,132,67}
\definecolor{YlGn-9-8}{RGB}{0,104,55}
\definecolor{YlGn-9-K}{RGB}{0,104,55}
\definecolor{YlGn-9-9}{RGB}{0,69,41}
\definecolor{YlGn-9-M}{RGB}{0,69,41}
\definecolor{YlGnBu-3-1}{RGB}{237,248,177}
\definecolor{YlGnBu-3-C}{RGB}{237,248,177}
\definecolor{YlGnBu-3-2}{RGB}{127,205,187}
\definecolor{YlGnBu-3-F}{RGB}{127,205,187}
\definecolor{YlGnBu-3-3}{RGB}{44,127,184}
\definecolor{YlGnBu-3-I}{RGB}{44,127,184}
\definecolor{YlGnBu-4-1}{RGB}{255,255,204}
\definecolor{YlGnBu-4-B}{RGB}{255,255,204}
\definecolor{YlGnBu-4-2}{RGB}{161,218,180}
\definecolor{YlGnBu-4-E}{RGB}{161,218,180}
\definecolor{YlGnBu-4-3}{RGB}{65,182,196}
\definecolor{YlGnBu-4-G}{RGB}{65,182,196}
\definecolor{YlGnBu-4-4}{RGB}{34,94,168}
\definecolor{YlGnBu-4-J}{RGB}{34,94,168}
\definecolor{YlGnBu-5-1}{RGB}{255,255,204}
\definecolor{YlGnBu-5-B}{RGB}{255,255,204}
\definecolor{YlGnBu-5-2}{RGB}{161,218,180}
\definecolor{YlGnBu-5-E}{RGB}{161,218,180}
\definecolor{YlGnBu-5-3}{RGB}{65,182,196}
\definecolor{YlGnBu-5-G}{RGB}{65,182,196}
\definecolor{YlGnBu-5-4}{RGB}{44,127,184}
\definecolor{YlGnBu-5-I}{RGB}{44,127,184}
\definecolor{YlGnBu-5-5}{RGB}{37,52,148}
\definecolor{YlGnBu-5-K}{RGB}{37,52,148}
\definecolor{YlGnBu-6-1}{RGB}{255,255,204}
\definecolor{YlGnBu-6-B}{RGB}{255,255,204}
\definecolor{YlGnBu-6-2}{RGB}{199,233,180}
\definecolor{YlGnBu-6-D}{RGB}{199,233,180}
\definecolor{YlGnBu-6-3}{RGB}{127,205,187}
\definecolor{YlGnBu-6-F}{RGB}{127,205,187}
\definecolor{YlGnBu-6-4}{RGB}{65,182,196}
\definecolor{YlGnBu-6-G}{RGB}{65,182,196}
\definecolor{YlGnBu-6-5}{RGB}{44,127,184}
\definecolor{YlGnBu-6-I}{RGB}{44,127,184}
\definecolor{YlGnBu-6-6}{RGB}{37,52,148}
\definecolor{YlGnBu-6-K}{RGB}{37,52,148}
\definecolor{YlGnBu-7-1}{RGB}{255,255,204}
\definecolor{YlGnBu-7-B}{RGB}{255,255,204}
\definecolor{YlGnBu-7-2}{RGB}{199,233,180}
\definecolor{YlGnBu-7-D}{RGB}{199,233,180}
\definecolor{YlGnBu-7-3}{RGB}{127,205,187}
\definecolor{YlGnBu-7-F}{RGB}{127,205,187}
\definecolor{YlGnBu-7-4}{RGB}{65,182,196}
\definecolor{YlGnBu-7-G}{RGB}{65,182,196}
\definecolor{YlGnBu-7-5}{RGB}{29,145,192}
\definecolor{YlGnBu-7-H}{RGB}{29,145,192}
\definecolor{YlGnBu-7-6}{RGB}{34,94,168}
\definecolor{YlGnBu-7-J}{RGB}{34,94,168}
\definecolor{YlGnBu-7-7}{RGB}{12,44,132}
\definecolor{YlGnBu-7-L}{RGB}{12,44,132}
\definecolor{YlGnBu-8-1}{RGB}{255,255,217}
\definecolor{YlGnBu-8-A}{RGB}{255,255,217}
\definecolor{YlGnBu-8-2}{RGB}{237,248,177}
\definecolor{YlGnBu-8-C}{RGB}{237,248,177}
\definecolor{YlGnBu-8-3}{RGB}{199,233,180}
\definecolor{YlGnBu-8-D}{RGB}{199,233,180}
\definecolor{YlGnBu-8-4}{RGB}{127,205,187}
\definecolor{YlGnBu-8-F}{RGB}{127,205,187}
\definecolor{YlGnBu-8-5}{RGB}{65,182,196}
\definecolor{YlGnBu-8-G}{RGB}{65,182,196}
\definecolor{YlGnBu-8-6}{RGB}{29,145,192}
\definecolor{YlGnBu-8-H}{RGB}{29,145,192}
\definecolor{YlGnBu-8-7}{RGB}{34,94,168}
\definecolor{YlGnBu-8-J}{RGB}{34,94,168}
\definecolor{YlGnBu-8-8}{RGB}{12,44,132}
\definecolor{YlGnBu-8-L}{RGB}{12,44,132}
\definecolor{YlGnBu-9-1}{RGB}{255,255,217}
\definecolor{YlGnBu-9-A}{RGB}{255,255,217}
\definecolor{YlGnBu-9-2}{RGB}{237,248,177}
\definecolor{YlGnBu-9-C}{RGB}{237,248,177}
\definecolor{YlGnBu-9-3}{RGB}{199,233,180}
\definecolor{YlGnBu-9-D}{RGB}{199,233,180}
\definecolor{YlGnBu-9-4}{RGB}{127,205,187}
\definecolor{YlGnBu-9-F}{RGB}{127,205,187}
\definecolor{YlGnBu-9-5}{RGB}{65,182,196}
\definecolor{YlGnBu-9-G}{RGB}{65,182,196}
\definecolor{YlGnBu-9-6}{RGB}{29,145,192}
\definecolor{YlGnBu-9-H}{RGB}{29,145,192}
\definecolor{YlGnBu-9-7}{RGB}{34,94,168}
\definecolor{YlGnBu-9-J}{RGB}{34,94,168}
\definecolor{YlGnBu-9-8}{RGB}{37,52,148}
\definecolor{YlGnBu-9-K}{RGB}{37,52,148}
\definecolor{YlGnBu-9-9}{RGB}{8,29,88}
\definecolor{YlGnBu-9-M}{RGB}{8,29,88}
\definecolor{GnBu-3-1}{RGB}{224,243,219}
\definecolor{GnBu-3-C}{RGB}{224,243,219}
\definecolor{GnBu-3-2}{RGB}{168,221,181}
\definecolor{GnBu-3-F}{RGB}{168,221,181}
\definecolor{GnBu-3-3}{RGB}{67,162,202}
\definecolor{GnBu-3-I}{RGB}{67,162,202}
\definecolor{GnBu-4-1}{RGB}{240,249,232}
\definecolor{GnBu-4-B}{RGB}{240,249,232}
\definecolor{GnBu-4-2}{RGB}{186,228,188}
\definecolor{GnBu-4-E}{RGB}{186,228,188}
\definecolor{GnBu-4-3}{RGB}{123,204,196}
\definecolor{GnBu-4-G}{RGB}{123,204,196}
\definecolor{GnBu-4-4}{RGB}{43,140,190}
\definecolor{GnBu-4-J}{RGB}{43,140,190}
\definecolor{GnBu-5-1}{RGB}{240,249,232}
\definecolor{GnBu-5-B}{RGB}{240,249,232}
\definecolor{GnBu-5-2}{RGB}{186,228,188}
\definecolor{GnBu-5-E}{RGB}{186,228,188}
\definecolor{GnBu-5-3}{RGB}{123,204,196}
\definecolor{GnBu-5-G}{RGB}{123,204,196}
\definecolor{GnBu-5-4}{RGB}{67,162,202}
\definecolor{GnBu-5-I}{RGB}{67,162,202}
\definecolor{GnBu-5-5}{RGB}{8,104,172}
\definecolor{GnBu-5-K}{RGB}{8,104,172}
\definecolor{GnBu-6-1}{RGB}{240,249,232}
\definecolor{GnBu-6-B}{RGB}{240,249,232}
\definecolor{GnBu-6-2}{RGB}{204,235,197}
\definecolor{GnBu-6-D}{RGB}{204,235,197}
\definecolor{GnBu-6-3}{RGB}{168,221,181}
\definecolor{GnBu-6-F}{RGB}{168,221,181}
\definecolor{GnBu-6-4}{RGB}{123,204,196}
\definecolor{GnBu-6-G}{RGB}{123,204,196}
\definecolor{GnBu-6-5}{RGB}{67,162,202}
\definecolor{GnBu-6-I}{RGB}{67,162,202}
\definecolor{GnBu-6-6}{RGB}{8,104,172}
\definecolor{GnBu-6-K}{RGB}{8,104,172}
\definecolor{GnBu-7-1}{RGB}{240,249,232}
\definecolor{GnBu-7-B}{RGB}{240,249,232}
\definecolor{GnBu-7-2}{RGB}{204,235,197}
\definecolor{GnBu-7-D}{RGB}{204,235,197}
\definecolor{GnBu-7-3}{RGB}{168,221,181}
\definecolor{GnBu-7-F}{RGB}{168,221,181}
\definecolor{GnBu-7-4}{RGB}{123,204,196}
\definecolor{GnBu-7-G}{RGB}{123,204,196}
\definecolor{GnBu-7-5}{RGB}{78,179,211}
\definecolor{GnBu-7-H}{RGB}{78,179,211}
\definecolor{GnBu-7-6}{RGB}{43,140,190}
\definecolor{GnBu-7-J}{RGB}{43,140,190}
\definecolor{GnBu-7-7}{RGB}{8,88,158}
\definecolor{GnBu-7-L}{RGB}{8,88,158}
\definecolor{GnBu-8-1}{RGB}{247,252,240}
\definecolor{GnBu-8-A}{RGB}{247,252,240}
\definecolor{GnBu-8-2}{RGB}{224,243,219}
\definecolor{GnBu-8-C}{RGB}{224,243,219}
\definecolor{GnBu-8-3}{RGB}{204,235,197}
\definecolor{GnBu-8-D}{RGB}{204,235,197}
\definecolor{GnBu-8-4}{RGB}{168,221,181}
\definecolor{GnBu-8-F}{RGB}{168,221,181}
\definecolor{GnBu-8-5}{RGB}{123,204,196}
\definecolor{GnBu-8-G}{RGB}{123,204,196}
\definecolor{GnBu-8-6}{RGB}{78,179,211}
\definecolor{GnBu-8-H}{RGB}{78,179,211}
\definecolor{GnBu-8-7}{RGB}{43,140,190}
\definecolor{GnBu-8-J}{RGB}{43,140,190}
\definecolor{GnBu-8-8}{RGB}{8,88,158}
\definecolor{GnBu-8-L}{RGB}{8,88,158}
\definecolor{GnBu-9-1}{RGB}{247,252,240}
\definecolor{GnBu-9-A}{RGB}{247,252,240}
\definecolor{GnBu-9-2}{RGB}{224,243,219}
\definecolor{GnBu-9-C}{RGB}{224,243,219}
\definecolor{GnBu-9-3}{RGB}{204,235,197}
\definecolor{GnBu-9-D}{RGB}{204,235,197}
\definecolor{GnBu-9-4}{RGB}{168,221,181}
\definecolor{GnBu-9-F}{RGB}{168,221,181}
\definecolor{GnBu-9-5}{RGB}{123,204,196}
\definecolor{GnBu-9-G}{RGB}{123,204,196}
\definecolor{GnBu-9-6}{RGB}{78,179,211}
\definecolor{GnBu-9-H}{RGB}{78,179,211}
\definecolor{GnBu-9-7}{RGB}{43,140,190}
\definecolor{GnBu-9-J}{RGB}{43,140,190}
\definecolor{GnBu-9-8}{RGB}{8,104,172}
\definecolor{GnBu-9-K}{RGB}{8,104,172}
\definecolor{GnBu-9-9}{RGB}{8,64,129}
\definecolor{GnBu-9-M}{RGB}{8,64,129}
\definecolor{BuGn-3-1}{RGB}{229,245,249}
\definecolor{BuGn-3-C}{RGB}{229,245,249}
\definecolor{BuGn-3-2}{RGB}{153,216,201}
\definecolor{BuGn-3-F}{RGB}{153,216,201}
\definecolor{BuGn-3-3}{RGB}{44,162,95}
\definecolor{BuGn-3-I}{RGB}{44,162,95}
\definecolor{BuGn-4-1}{RGB}{237,248,251}
\definecolor{BuGn-4-B}{RGB}{237,248,251}
\definecolor{BuGn-4-2}{RGB}{178,226,226}
\definecolor{BuGn-4-E}{RGB}{178,226,226}
\definecolor{BuGn-4-3}{RGB}{102,194,164}
\definecolor{BuGn-4-G}{RGB}{102,194,164}
\definecolor{BuGn-4-4}{RGB}{35,139,69}
\definecolor{BuGn-4-J}{RGB}{35,139,69}
\definecolor{BuGn-5-1}{RGB}{237,248,251}
\definecolor{BuGn-5-B}{RGB}{237,248,251}
\definecolor{BuGn-5-2}{RGB}{178,226,226}
\definecolor{BuGn-5-E}{RGB}{178,226,226}
\definecolor{BuGn-5-3}{RGB}{102,194,164}
\definecolor{BuGn-5-G}{RGB}{102,194,164}
\definecolor{BuGn-5-4}{RGB}{44,162,95}
\definecolor{BuGn-5-I}{RGB}{44,162,95}
\definecolor{BuGn-5-5}{RGB}{0,109,44}
\definecolor{BuGn-5-K}{RGB}{0,109,44}
\definecolor{BuGn-6-1}{RGB}{237,248,251}
\definecolor{BuGn-6-B}{RGB}{237,248,251}
\definecolor{BuGn-6-2}{RGB}{204,236,230}
\definecolor{BuGn-6-D}{RGB}{204,236,230}
\definecolor{BuGn-6-3}{RGB}{153,216,201}
\definecolor{BuGn-6-F}{RGB}{153,216,201}
\definecolor{BuGn-6-4}{RGB}{102,194,164}
\definecolor{BuGn-6-G}{RGB}{102,194,164}
\definecolor{BuGn-6-5}{RGB}{44,162,95}
\definecolor{BuGn-6-I}{RGB}{44,162,95}
\definecolor{BuGn-6-6}{RGB}{0,109,44}
\definecolor{BuGn-6-K}{RGB}{0,109,44}
\definecolor{BuGn-7-1}{RGB}{237,248,251}
\definecolor{BuGn-7-B}{RGB}{237,248,251}
\definecolor{BuGn-7-2}{RGB}{204,236,230}
\definecolor{BuGn-7-D}{RGB}{204,236,230}
\definecolor{BuGn-7-3}{RGB}{153,216,201}
\definecolor{BuGn-7-F}{RGB}{153,216,201}
\definecolor{BuGn-7-4}{RGB}{102,194,164}
\definecolor{BuGn-7-G}{RGB}{102,194,164}
\definecolor{BuGn-7-5}{RGB}{65,174,118}
\definecolor{BuGn-7-H}{RGB}{65,174,118}
\definecolor{BuGn-7-6}{RGB}{35,139,69}
\definecolor{BuGn-7-J}{RGB}{35,139,69}
\definecolor{BuGn-7-7}{RGB}{0,88,36}
\definecolor{BuGn-7-L}{RGB}{0,88,36}
\definecolor{BuGn-8-1}{RGB}{247,252,253}
\definecolor{BuGn-8-A}{RGB}{247,252,253}
\definecolor{BuGn-8-2}{RGB}{229,245,249}
\definecolor{BuGn-8-C}{RGB}{229,245,249}
\definecolor{BuGn-8-3}{RGB}{204,236,230}
\definecolor{BuGn-8-D}{RGB}{204,236,230}
\definecolor{BuGn-8-4}{RGB}{153,216,201}
\definecolor{BuGn-8-F}{RGB}{153,216,201}
\definecolor{BuGn-8-5}{RGB}{102,194,164}
\definecolor{BuGn-8-G}{RGB}{102,194,164}
\definecolor{BuGn-8-6}{RGB}{65,174,118}
\definecolor{BuGn-8-H}{RGB}{65,174,118}
\definecolor{BuGn-8-7}{RGB}{35,139,69}
\definecolor{BuGn-8-J}{RGB}{35,139,69}
\definecolor{BuGn-8-8}{RGB}{0,88,36}
\definecolor{BuGn-8-L}{RGB}{0,88,36}
\definecolor{BuGn-9-1}{RGB}{247,252,253}
\definecolor{BuGn-9-A}{RGB}{247,252,253}
\definecolor{BuGn-9-2}{RGB}{229,245,249}
\definecolor{BuGn-9-C}{RGB}{229,245,249}
\definecolor{BuGn-9-3}{RGB}{204,236,230}
\definecolor{BuGn-9-D}{RGB}{204,236,230}
\definecolor{BuGn-9-4}{RGB}{153,216,201}
\definecolor{BuGn-9-F}{RGB}{153,216,201}
\definecolor{BuGn-9-5}{RGB}{102,194,164}
\definecolor{BuGn-9-G}{RGB}{102,194,164}
\definecolor{BuGn-9-6}{RGB}{65,174,118}
\definecolor{BuGn-9-H}{RGB}{65,174,118}
\definecolor{BuGn-9-7}{RGB}{35,139,69}
\definecolor{BuGn-9-J}{RGB}{35,139,69}
\definecolor{BuGn-9-8}{RGB}{0,109,44}
\definecolor{BuGn-9-K}{RGB}{0,109,44}
\definecolor{BuGn-9-9}{RGB}{0,68,27}
\definecolor{BuGn-9-M}{RGB}{0,68,27}
\definecolor{PuBuGn-3-1}{RGB}{236,226,240}
\definecolor{PuBuGn-3-C}{RGB}{236,226,240}
\definecolor{PuBuGn-3-2}{RGB}{166,189,219}
\definecolor{PuBuGn-3-F}{RGB}{166,189,219}
\definecolor{PuBuGn-3-3}{RGB}{28,144,153}
\definecolor{PuBuGn-3-I}{RGB}{28,144,153}
\definecolor{PuBuGn-4-1}{RGB}{246,239,247}
\definecolor{PuBuGn-4-B}{RGB}{246,239,247}
\definecolor{PuBuGn-4-2}{RGB}{189,201,225}
\definecolor{PuBuGn-4-E}{RGB}{189,201,225}
\definecolor{PuBuGn-4-3}{RGB}{103,169,207}
\definecolor{PuBuGn-4-G}{RGB}{103,169,207}
\definecolor{PuBuGn-4-4}{RGB}{2,129,138}
\definecolor{PuBuGn-4-J}{RGB}{2,129,138}
\definecolor{PuBuGn-5-1}{RGB}{246,239,247}
\definecolor{PuBuGn-5-B}{RGB}{246,239,247}
\definecolor{PuBuGn-5-2}{RGB}{189,201,225}
\definecolor{PuBuGn-5-E}{RGB}{189,201,225}
\definecolor{PuBuGn-5-3}{RGB}{103,169,207}
\definecolor{PuBuGn-5-G}{RGB}{103,169,207}
\definecolor{PuBuGn-5-4}{RGB}{28,144,153}
\definecolor{PuBuGn-5-I}{RGB}{28,144,153}
\definecolor{PuBuGn-5-5}{RGB}{1,108,89}
\definecolor{PuBuGn-5-K}{RGB}{1,108,89}
\definecolor{PuBuGn-6-1}{RGB}{246,239,247}
\definecolor{PuBuGn-6-B}{RGB}{246,239,247}
\definecolor{PuBuGn-6-2}{RGB}{208,209,230}
\definecolor{PuBuGn-6-D}{RGB}{208,209,230}
\definecolor{PuBuGn-6-3}{RGB}{166,189,219}
\definecolor{PuBuGn-6-F}{RGB}{166,189,219}
\definecolor{PuBuGn-6-4}{RGB}{103,169,207}
\definecolor{PuBuGn-6-G}{RGB}{103,169,207}
\definecolor{PuBuGn-6-5}{RGB}{28,144,153}
\definecolor{PuBuGn-6-I}{RGB}{28,144,153}
\definecolor{PuBuGn-6-6}{RGB}{1,108,89}
\definecolor{PuBuGn-6-K}{RGB}{1,108,89}
\definecolor{PuBuGn-7-1}{RGB}{246,239,247}
\definecolor{PuBuGn-7-B}{RGB}{246,239,247}
\definecolor{PuBuGn-7-2}{RGB}{208,209,230}
\definecolor{PuBuGn-7-D}{RGB}{208,209,230}
\definecolor{PuBuGn-7-3}{RGB}{166,189,219}
\definecolor{PuBuGn-7-F}{RGB}{166,189,219}
\definecolor{PuBuGn-7-4}{RGB}{103,169,207}
\definecolor{PuBuGn-7-G}{RGB}{103,169,207}
\definecolor{PuBuGn-7-5}{RGB}{54,144,192}
\definecolor{PuBuGn-7-H}{RGB}{54,144,192}
\definecolor{PuBuGn-7-6}{RGB}{2,129,138}
\definecolor{PuBuGn-7-J}{RGB}{2,129,138}
\definecolor{PuBuGn-7-7}{RGB}{1,100,80}
\definecolor{PuBuGn-7-L}{RGB}{1,100,80}
\definecolor{PuBuGn-8-1}{RGB}{255,247,251}
\definecolor{PuBuGn-8-A}{RGB}{255,247,251}
\definecolor{PuBuGn-8-2}{RGB}{236,226,240}
\definecolor{PuBuGn-8-C}{RGB}{236,226,240}
\definecolor{PuBuGn-8-3}{RGB}{208,209,230}
\definecolor{PuBuGn-8-D}{RGB}{208,209,230}
\definecolor{PuBuGn-8-4}{RGB}{166,189,219}
\definecolor{PuBuGn-8-F}{RGB}{166,189,219}
\definecolor{PuBuGn-8-5}{RGB}{103,169,207}
\definecolor{PuBuGn-8-G}{RGB}{103,169,207}
\definecolor{PuBuGn-8-6}{RGB}{54,144,192}
\definecolor{PuBuGn-8-H}{RGB}{54,144,192}
\definecolor{PuBuGn-8-7}{RGB}{2,129,138}
\definecolor{PuBuGn-8-J}{RGB}{2,129,138}
\definecolor{PuBuGn-8-8}{RGB}{1,100,80}
\definecolor{PuBuGn-8-L}{RGB}{1,100,80}
\definecolor{PuBuGn-9-1}{RGB}{255,247,251}
\definecolor{PuBuGn-9-A}{RGB}{255,247,251}
\definecolor{PuBuGn-9-2}{RGB}{236,226,240}
\definecolor{PuBuGn-9-C}{RGB}{236,226,240}
\definecolor{PuBuGn-9-3}{RGB}{208,209,230}
\definecolor{PuBuGn-9-D}{RGB}{208,209,230}
\definecolor{PuBuGn-9-4}{RGB}{166,189,219}
\definecolor{PuBuGn-9-F}{RGB}{166,189,219}
\definecolor{PuBuGn-9-5}{RGB}{103,169,207}
\definecolor{PuBuGn-9-G}{RGB}{103,169,207}
\definecolor{PuBuGn-9-6}{RGB}{54,144,192}
\definecolor{PuBuGn-9-H}{RGB}{54,144,192}
\definecolor{PuBuGn-9-7}{RGB}{2,129,138}
\definecolor{PuBuGn-9-J}{RGB}{2,129,138}
\definecolor{PuBuGn-9-8}{RGB}{1,108,89}
\definecolor{PuBuGn-9-K}{RGB}{1,108,89}
\definecolor{PuBuGn-9-9}{RGB}{1,70,54}
\definecolor{PuBuGn-9-M}{RGB}{1,70,54}
\definecolor{PuBu-3-1}{RGB}{236,231,242}
\definecolor{PuBu-3-C}{RGB}{236,231,242}
\definecolor{PuBu-3-2}{RGB}{166,189,219}
\definecolor{PuBu-3-F}{RGB}{166,189,219}
\definecolor{PuBu-3-3}{RGB}{43,140,190}
\definecolor{PuBu-3-I}{RGB}{43,140,190}
\definecolor{PuBu-4-1}{RGB}{241,238,246}
\definecolor{PuBu-4-B}{RGB}{241,238,246}
\definecolor{PuBu-4-2}{RGB}{189,201,225}
\definecolor{PuBu-4-E}{RGB}{189,201,225}
\definecolor{PuBu-4-3}{RGB}{116,169,207}
\definecolor{PuBu-4-G}{RGB}{116,169,207}
\definecolor{PuBu-4-4}{RGB}{5,112,176}
\definecolor{PuBu-4-J}{RGB}{5,112,176}
\definecolor{PuBu-5-1}{RGB}{241,238,246}
\definecolor{PuBu-5-B}{RGB}{241,238,246}
\definecolor{PuBu-5-2}{RGB}{189,201,225}
\definecolor{PuBu-5-E}{RGB}{189,201,225}
\definecolor{PuBu-5-3}{RGB}{116,169,207}
\definecolor{PuBu-5-G}{RGB}{116,169,207}
\definecolor{PuBu-5-4}{RGB}{43,140,190}
\definecolor{PuBu-5-I}{RGB}{43,140,190}
\definecolor{PuBu-5-5}{RGB}{4,90,141}
\definecolor{PuBu-5-K}{RGB}{4,90,141}
\definecolor{PuBu-6-1}{RGB}{241,238,246}
\definecolor{PuBu-6-B}{RGB}{241,238,246}
\definecolor{PuBu-6-2}{RGB}{208,209,230}
\definecolor{PuBu-6-D}{RGB}{208,209,230}
\definecolor{PuBu-6-3}{RGB}{166,189,219}
\definecolor{PuBu-6-F}{RGB}{166,189,219}
\definecolor{PuBu-6-4}{RGB}{116,169,207}
\definecolor{PuBu-6-G}{RGB}{116,169,207}
\definecolor{PuBu-6-5}{RGB}{43,140,190}
\definecolor{PuBu-6-I}{RGB}{43,140,190}
\definecolor{PuBu-6-6}{RGB}{4,90,141}
\definecolor{PuBu-6-K}{RGB}{4,90,141}
\definecolor{PuBu-7-1}{RGB}{241,238,246}
\definecolor{PuBu-7-B}{RGB}{241,238,246}
\definecolor{PuBu-7-2}{RGB}{208,209,230}
\definecolor{PuBu-7-D}{RGB}{208,209,230}
\definecolor{PuBu-7-3}{RGB}{166,189,219}
\definecolor{PuBu-7-F}{RGB}{166,189,219}
\definecolor{PuBu-7-4}{RGB}{116,169,207}
\definecolor{PuBu-7-G}{RGB}{116,169,207}
\definecolor{PuBu-7-5}{RGB}{54,144,192}
\definecolor{PuBu-7-H}{RGB}{54,144,192}
\definecolor{PuBu-7-6}{RGB}{5,112,176}
\definecolor{PuBu-7-J}{RGB}{5,112,176}
\definecolor{PuBu-7-7}{RGB}{3,78,123}
\definecolor{PuBu-7-L}{RGB}{3,78,123}
\definecolor{PuBu-8-1}{RGB}{255,247,251}
\definecolor{PuBu-8-A}{RGB}{255,247,251}
\definecolor{PuBu-8-2}{RGB}{236,231,242}
\definecolor{PuBu-8-C}{RGB}{236,231,242}
\definecolor{PuBu-8-3}{RGB}{208,209,230}
\definecolor{PuBu-8-D}{RGB}{208,209,230}
\definecolor{PuBu-8-4}{RGB}{166,189,219}
\definecolor{PuBu-8-F}{RGB}{166,189,219}
\definecolor{PuBu-8-5}{RGB}{116,169,207}
\definecolor{PuBu-8-G}{RGB}{116,169,207}
\definecolor{PuBu-8-6}{RGB}{54,144,192}
\definecolor{PuBu-8-H}{RGB}{54,144,192}
\definecolor{PuBu-8-7}{RGB}{5,112,176}
\definecolor{PuBu-8-J}{RGB}{5,112,176}
\definecolor{PuBu-8-8}{RGB}{3,78,123}
\definecolor{PuBu-8-L}{RGB}{3,78,123}
\definecolor{PuBu-9-1}{RGB}{255,247,251}
\definecolor{PuBu-9-A}{RGB}{255,247,251}
\definecolor{PuBu-9-2}{RGB}{236,231,242}
\definecolor{PuBu-9-C}{RGB}{236,231,242}
\definecolor{PuBu-9-3}{RGB}{208,209,230}
\definecolor{PuBu-9-D}{RGB}{208,209,230}
\definecolor{PuBu-9-4}{RGB}{166,189,219}
\definecolor{PuBu-9-F}{RGB}{166,189,219}
\definecolor{PuBu-9-5}{RGB}{116,169,207}
\definecolor{PuBu-9-G}{RGB}{116,169,207}
\definecolor{PuBu-9-6}{RGB}{54,144,192}
\definecolor{PuBu-9-H}{RGB}{54,144,192}
\definecolor{PuBu-9-7}{RGB}{5,112,176}
\definecolor{PuBu-9-J}{RGB}{5,112,176}
\definecolor{PuBu-9-8}{RGB}{4,90,141}
\definecolor{PuBu-9-K}{RGB}{4,90,141}
\definecolor{PuBu-9-9}{RGB}{2,56,88}
\definecolor{PuBu-9-M}{RGB}{2,56,88}
\definecolor{BuPu-3-1}{RGB}{224,236,244}
\definecolor{BuPu-3-C}{RGB}{224,236,244}
\definecolor{BuPu-3-2}{RGB}{158,188,218}
\definecolor{BuPu-3-F}{RGB}{158,188,218}
\definecolor{BuPu-3-3}{RGB}{136,86,167}
\definecolor{BuPu-3-I}{RGB}{136,86,167}
\definecolor{BuPu-4-1}{RGB}{237,248,251}
\definecolor{BuPu-4-B}{RGB}{237,248,251}
\definecolor{BuPu-4-2}{RGB}{179,205,227}
\definecolor{BuPu-4-E}{RGB}{179,205,227}
\definecolor{BuPu-4-3}{RGB}{140,150,198}
\definecolor{BuPu-4-G}{RGB}{140,150,198}
\definecolor{BuPu-4-4}{RGB}{136,65,157}
\definecolor{BuPu-4-J}{RGB}{136,65,157}
\definecolor{BuPu-5-1}{RGB}{237,248,251}
\definecolor{BuPu-5-B}{RGB}{237,248,251}
\definecolor{BuPu-5-2}{RGB}{179,205,227}
\definecolor{BuPu-5-E}{RGB}{179,205,227}
\definecolor{BuPu-5-3}{RGB}{140,150,198}
\definecolor{BuPu-5-G}{RGB}{140,150,198}
\definecolor{BuPu-5-4}{RGB}{136,86,167}
\definecolor{BuPu-5-I}{RGB}{136,86,167}
\definecolor{BuPu-5-5}{RGB}{129,15,124}
\definecolor{BuPu-5-K}{RGB}{129,15,124}
\definecolor{BuPu-6-1}{RGB}{237,248,251}
\definecolor{BuPu-6-B}{RGB}{237,248,251}
\definecolor{BuPu-6-2}{RGB}{191,211,230}
\definecolor{BuPu-6-D}{RGB}{191,211,230}
\definecolor{BuPu-6-3}{RGB}{158,188,218}
\definecolor{BuPu-6-F}{RGB}{158,188,218}
\definecolor{BuPu-6-4}{RGB}{140,150,198}
\definecolor{BuPu-6-G}{RGB}{140,150,198}
\definecolor{BuPu-6-5}{RGB}{136,86,167}
\definecolor{BuPu-6-I}{RGB}{136,86,167}
\definecolor{BuPu-6-6}{RGB}{129,15,124}
\definecolor{BuPu-6-K}{RGB}{129,15,124}
\definecolor{BuPu-7-1}{RGB}{237,248,251}
\definecolor{BuPu-7-B}{RGB}{237,248,251}
\definecolor{BuPu-7-2}{RGB}{191,211,230}
\definecolor{BuPu-7-D}{RGB}{191,211,230}
\definecolor{BuPu-7-3}{RGB}{158,188,218}
\definecolor{BuPu-7-F}{RGB}{158,188,218}
\definecolor{BuPu-7-4}{RGB}{140,150,198}
\definecolor{BuPu-7-G}{RGB}{140,150,198}
\definecolor{BuPu-7-5}{RGB}{140,107,177}
\definecolor{BuPu-7-H}{RGB}{140,107,177}
\definecolor{BuPu-7-6}{RGB}{136,65,157}
\definecolor{BuPu-7-J}{RGB}{136,65,157}
\definecolor{BuPu-7-7}{RGB}{110,1,107}
\definecolor{BuPu-7-L}{RGB}{110,1,107}
\definecolor{BuPu-8-1}{RGB}{247,252,253}
\definecolor{BuPu-8-A}{RGB}{247,252,253}
\definecolor{BuPu-8-2}{RGB}{224,236,244}
\definecolor{BuPu-8-C}{RGB}{224,236,244}
\definecolor{BuPu-8-3}{RGB}{191,211,230}
\definecolor{BuPu-8-D}{RGB}{191,211,230}
\definecolor{BuPu-8-4}{RGB}{158,188,218}
\definecolor{BuPu-8-F}{RGB}{158,188,218}
\definecolor{BuPu-8-5}{RGB}{140,150,198}
\definecolor{BuPu-8-G}{RGB}{140,150,198}
\definecolor{BuPu-8-6}{RGB}{140,107,177}
\definecolor{BuPu-8-H}{RGB}{140,107,177}
\definecolor{BuPu-8-7}{RGB}{136,65,157}
\definecolor{BuPu-8-J}{RGB}{136,65,157}
\definecolor{BuPu-8-8}{RGB}{110,1,107}
\definecolor{BuPu-8-L}{RGB}{110,1,107}
\definecolor{BuPu-9-1}{RGB}{247,252,253}
\definecolor{BuPu-9-A}{RGB}{247,252,253}
\definecolor{BuPu-9-2}{RGB}{224,236,244}
\definecolor{BuPu-9-C}{RGB}{224,236,244}
\definecolor{BuPu-9-3}{RGB}{191,211,230}
\definecolor{BuPu-9-D}{RGB}{191,211,230}
\definecolor{BuPu-9-4}{RGB}{158,188,218}
\definecolor{BuPu-9-F}{RGB}{158,188,218}
\definecolor{BuPu-9-5}{RGB}{140,150,198}
\definecolor{BuPu-9-G}{RGB}{140,150,198}
\definecolor{BuPu-9-6}{RGB}{140,107,177}
\definecolor{BuPu-9-H}{RGB}{140,107,177}
\definecolor{BuPu-9-7}{RGB}{136,65,157}
\definecolor{BuPu-9-J}{RGB}{136,65,157}
\definecolor{BuPu-9-8}{RGB}{129,15,124}
\definecolor{BuPu-9-K}{RGB}{129,15,124}
\definecolor{BuPu-9-9}{RGB}{77,0,75}
\definecolor{BuPu-9-M}{RGB}{77,0,75}
\definecolor{RdPu-3-1}{RGB}{253,224,221}
\definecolor{RdPu-3-C}{RGB}{253,224,221}
\definecolor{RdPu-3-2}{RGB}{250,159,181}
\definecolor{RdPu-3-F}{RGB}{250,159,181}
\definecolor{RdPu-3-3}{RGB}{197,27,138}
\definecolor{RdPu-3-I}{RGB}{197,27,138}
\definecolor{RdPu-4-1}{RGB}{254,235,226}
\definecolor{RdPu-4-B}{RGB}{254,235,226}
\definecolor{RdPu-4-2}{RGB}{251,180,185}
\definecolor{RdPu-4-E}{RGB}{251,180,185}
\definecolor{RdPu-4-3}{RGB}{247,104,161}
\definecolor{RdPu-4-G}{RGB}{247,104,161}
\definecolor{RdPu-4-4}{RGB}{174,1,126}
\definecolor{RdPu-4-J}{RGB}{174,1,126}
\definecolor{RdPu-5-1}{RGB}{254,235,226}
\definecolor{RdPu-5-B}{RGB}{254,235,226}
\definecolor{RdPu-5-2}{RGB}{251,180,185}
\definecolor{RdPu-5-E}{RGB}{251,180,185}
\definecolor{RdPu-5-3}{RGB}{247,104,161}
\definecolor{RdPu-5-G}{RGB}{247,104,161}
\definecolor{RdPu-5-4}{RGB}{197,27,138}
\definecolor{RdPu-5-I}{RGB}{197,27,138}
\definecolor{RdPu-5-5}{RGB}{122,1,119}
\definecolor{RdPu-5-K}{RGB}{122,1,119}
\definecolor{RdPu-6-1}{RGB}{254,235,226}
\definecolor{RdPu-6-B}{RGB}{254,235,226}
\definecolor{RdPu-6-2}{RGB}{252,197,192}
\definecolor{RdPu-6-D}{RGB}{252,197,192}
\definecolor{RdPu-6-3}{RGB}{250,159,181}
\definecolor{RdPu-6-F}{RGB}{250,159,181}
\definecolor{RdPu-6-4}{RGB}{247,104,161}
\definecolor{RdPu-6-G}{RGB}{247,104,161}
\definecolor{RdPu-6-5}{RGB}{197,27,138}
\definecolor{RdPu-6-I}{RGB}{197,27,138}
\definecolor{RdPu-6-6}{RGB}{122,1,119}
\definecolor{RdPu-6-K}{RGB}{122,1,119}
\definecolor{RdPu-7-1}{RGB}{254,235,226}
\definecolor{RdPu-7-B}{RGB}{254,235,226}
\definecolor{RdPu-7-2}{RGB}{252,197,192}
\definecolor{RdPu-7-D}{RGB}{252,197,192}
\definecolor{RdPu-7-3}{RGB}{250,159,181}
\definecolor{RdPu-7-F}{RGB}{250,159,181}
\definecolor{RdPu-7-4}{RGB}{247,104,161}
\definecolor{RdPu-7-G}{RGB}{247,104,161}
\definecolor{RdPu-7-5}{RGB}{221,52,151}
\definecolor{RdPu-7-H}{RGB}{221,52,151}
\definecolor{RdPu-7-6}{RGB}{174,1,126}
\definecolor{RdPu-7-J}{RGB}{174,1,126}
\definecolor{RdPu-7-7}{RGB}{122,1,119}
\definecolor{RdPu-7-L}{RGB}{122,1,119}
\definecolor{RdPu-8-1}{RGB}{255,247,243}
\definecolor{RdPu-8-A}{RGB}{255,247,243}
\definecolor{RdPu-8-2}{RGB}{253,224,221}
\definecolor{RdPu-8-C}{RGB}{253,224,221}
\definecolor{RdPu-8-3}{RGB}{252,197,192}
\definecolor{RdPu-8-D}{RGB}{252,197,192}
\definecolor{RdPu-8-4}{RGB}{250,159,181}
\definecolor{RdPu-8-F}{RGB}{250,159,181}
\definecolor{RdPu-8-5}{RGB}{247,104,161}
\definecolor{RdPu-8-G}{RGB}{247,104,161}
\definecolor{RdPu-8-6}{RGB}{221,52,151}
\definecolor{RdPu-8-H}{RGB}{221,52,151}
\definecolor{RdPu-8-7}{RGB}{174,1,126}
\definecolor{RdPu-8-J}{RGB}{174,1,126}
\definecolor{RdPu-8-8}{RGB}{122,1,119}
\definecolor{RdPu-8-L}{RGB}{122,1,119}
\definecolor{RdPu-9-1}{RGB}{255,247,243}
\definecolor{RdPu-9-A}{RGB}{255,247,243}
\definecolor{RdPu-9-2}{RGB}{253,224,221}
\definecolor{RdPu-9-C}{RGB}{253,224,221}
\definecolor{RdPu-9-3}{RGB}{252,197,192}
\definecolor{RdPu-9-D}{RGB}{252,197,192}
\definecolor{RdPu-9-4}{RGB}{250,159,181}
\definecolor{RdPu-9-F}{RGB}{250,159,181}
\definecolor{RdPu-9-5}{RGB}{247,104,161}
\definecolor{RdPu-9-G}{RGB}{247,104,161}
\definecolor{RdPu-9-6}{RGB}{221,52,151}
\definecolor{RdPu-9-H}{RGB}{221,52,151}
\definecolor{RdPu-9-7}{RGB}{174,1,126}
\definecolor{RdPu-9-J}{RGB}{174,1,126}
\definecolor{RdPu-9-8}{RGB}{122,1,119}
\definecolor{RdPu-9-K}{RGB}{122,1,119}
\definecolor{RdPu-9-9}{RGB}{73,0,106}
\definecolor{RdPu-9-M}{RGB}{73,0,106}
\definecolor{PuRd-3-1}{RGB}{231,225,239}
\definecolor{PuRd-3-C}{RGB}{231,225,239}
\definecolor{PuRd-3-2}{RGB}{201,148,199}
\definecolor{PuRd-3-F}{RGB}{201,148,199}
\definecolor{PuRd-3-3}{RGB}{221,28,119}
\definecolor{PuRd-3-I}{RGB}{221,28,119}
\definecolor{PuRd-4-1}{RGB}{241,238,246}
\definecolor{PuRd-4-B}{RGB}{241,238,246}
\definecolor{PuRd-4-2}{RGB}{215,181,216}
\definecolor{PuRd-4-E}{RGB}{215,181,216}
\definecolor{PuRd-4-3}{RGB}{223,101,176}
\definecolor{PuRd-4-G}{RGB}{223,101,176}
\definecolor{PuRd-4-4}{RGB}{206,18,86}
\definecolor{PuRd-4-J}{RGB}{206,18,86}
\definecolor{PuRd-5-1}{RGB}{241,238,246}
\definecolor{PuRd-5-B}{RGB}{241,238,246}
\definecolor{PuRd-5-2}{RGB}{215,181,216}
\definecolor{PuRd-5-E}{RGB}{215,181,216}
\definecolor{PuRd-5-3}{RGB}{223,101,176}
\definecolor{PuRd-5-G}{RGB}{223,101,176}
\definecolor{PuRd-5-4}{RGB}{221,28,119}
\definecolor{PuRd-5-I}{RGB}{221,28,119}
\definecolor{PuRd-5-5}{RGB}{152,0,67}
\definecolor{PuRd-5-K}{RGB}{152,0,67}
\definecolor{PuRd-6-1}{RGB}{241,238,246}
\definecolor{PuRd-6-B}{RGB}{241,238,246}
\definecolor{PuRd-6-2}{RGB}{212,185,218}
\definecolor{PuRd-6-D}{RGB}{212,185,218}
\definecolor{PuRd-6-3}{RGB}{201,148,199}
\definecolor{PuRd-6-F}{RGB}{201,148,199}
\definecolor{PuRd-6-4}{RGB}{223,101,176}
\definecolor{PuRd-6-G}{RGB}{223,101,176}
\definecolor{PuRd-6-5}{RGB}{221,28,119}
\definecolor{PuRd-6-I}{RGB}{221,28,119}
\definecolor{PuRd-6-6}{RGB}{152,0,67}
\definecolor{PuRd-6-K}{RGB}{152,0,67}
\definecolor{PuRd-7-1}{RGB}{241,238,246}
\definecolor{PuRd-7-B}{RGB}{241,238,246}
\definecolor{PuRd-7-2}{RGB}{212,185,218}
\definecolor{PuRd-7-D}{RGB}{212,185,218}
\definecolor{PuRd-7-3}{RGB}{201,148,199}
\definecolor{PuRd-7-F}{RGB}{201,148,199}
\definecolor{PuRd-7-4}{RGB}{223,101,176}
\definecolor{PuRd-7-G}{RGB}{223,101,176}
\definecolor{PuRd-7-5}{RGB}{231,41,138}
\definecolor{PuRd-7-H}{RGB}{231,41,138}
\definecolor{PuRd-7-6}{RGB}{206,18,86}
\definecolor{PuRd-7-J}{RGB}{206,18,86}
\definecolor{PuRd-7-7}{RGB}{145,0,63}
\definecolor{PuRd-7-L}{RGB}{145,0,63}
\definecolor{PuRd-8-1}{RGB}{247,244,249}
\definecolor{PuRd-8-A}{RGB}{247,244,249}
\definecolor{PuRd-8-2}{RGB}{231,225,239}
\definecolor{PuRd-8-C}{RGB}{231,225,239}
\definecolor{PuRd-8-3}{RGB}{212,185,218}
\definecolor{PuRd-8-D}{RGB}{212,185,218}
\definecolor{PuRd-8-4}{RGB}{201,148,199}
\definecolor{PuRd-8-F}{RGB}{201,148,199}
\definecolor{PuRd-8-5}{RGB}{223,101,176}
\definecolor{PuRd-8-G}{RGB}{223,101,176}
\definecolor{PuRd-8-6}{RGB}{231,41,138}
\definecolor{PuRd-8-H}{RGB}{231,41,138}
\definecolor{PuRd-8-7}{RGB}{206,18,86}
\definecolor{PuRd-8-J}{RGB}{206,18,86}
\definecolor{PuRd-8-8}{RGB}{145,0,63}
\definecolor{PuRd-8-L}{RGB}{145,0,63}
\definecolor{PuRd-9-1}{RGB}{247,244,249}
\definecolor{PuRd-9-A}{RGB}{247,244,249}
\definecolor{PuRd-9-2}{RGB}{231,225,239}
\definecolor{PuRd-9-C}{RGB}{231,225,239}
\definecolor{PuRd-9-3}{RGB}{212,185,218}
\definecolor{PuRd-9-D}{RGB}{212,185,218}
\definecolor{PuRd-9-4}{RGB}{201,148,199}
\definecolor{PuRd-9-F}{RGB}{201,148,199}
\definecolor{PuRd-9-5}{RGB}{223,101,176}
\definecolor{PuRd-9-G}{RGB}{223,101,176}
\definecolor{PuRd-9-6}{RGB}{231,41,138}
\definecolor{PuRd-9-H}{RGB}{231,41,138}
\definecolor{PuRd-9-7}{RGB}{206,18,86}
\definecolor{PuRd-9-J}{RGB}{206,18,86}
\definecolor{PuRd-9-8}{RGB}{152,0,67}
\definecolor{PuRd-9-K}{RGB}{152,0,67}
\definecolor{PuRd-9-9}{RGB}{103,0,31}
\definecolor{PuRd-9-M}{RGB}{103,0,31}
\definecolor{OrRd-3-1}{RGB}{254,232,200}
\definecolor{OrRd-3-C}{RGB}{254,232,200}
\definecolor{OrRd-3-2}{RGB}{253,187,132}
\definecolor{OrRd-3-F}{RGB}{253,187,132}
\definecolor{OrRd-3-3}{RGB}{227,74,51}
\definecolor{OrRd-3-I}{RGB}{227,74,51}
\definecolor{OrRd-4-1}{RGB}{254,240,217}
\definecolor{OrRd-4-B}{RGB}{254,240,217}
\definecolor{OrRd-4-2}{RGB}{253,204,138}
\definecolor{OrRd-4-E}{RGB}{253,204,138}
\definecolor{OrRd-4-3}{RGB}{252,141,89}
\definecolor{OrRd-4-G}{RGB}{252,141,89}
\definecolor{OrRd-4-4}{RGB}{215,48,31}
\definecolor{OrRd-4-J}{RGB}{215,48,31}
\definecolor{OrRd-5-1}{RGB}{254,240,217}
\definecolor{OrRd-5-B}{RGB}{254,240,217}
\definecolor{OrRd-5-2}{RGB}{253,204,138}
\definecolor{OrRd-5-E}{RGB}{253,204,138}
\definecolor{OrRd-5-3}{RGB}{252,141,89}
\definecolor{OrRd-5-G}{RGB}{252,141,89}
\definecolor{OrRd-5-4}{RGB}{227,74,51}
\definecolor{OrRd-5-I}{RGB}{227,74,51}
\definecolor{OrRd-5-5}{RGB}{179,0,0}
\definecolor{OrRd-5-K}{RGB}{179,0,0}
\definecolor{OrRd-6-1}{RGB}{254,240,217}
\definecolor{OrRd-6-B}{RGB}{254,240,217}
\definecolor{OrRd-6-2}{RGB}{253,212,158}
\definecolor{OrRd-6-D}{RGB}{253,212,158}
\definecolor{OrRd-6-3}{RGB}{253,187,132}
\definecolor{OrRd-6-F}{RGB}{253,187,132}
\definecolor{OrRd-6-4}{RGB}{252,141,89}
\definecolor{OrRd-6-G}{RGB}{252,141,89}
\definecolor{OrRd-6-5}{RGB}{227,74,51}
\definecolor{OrRd-6-I}{RGB}{227,74,51}
\definecolor{OrRd-6-6}{RGB}{179,0,0}
\definecolor{OrRd-6-K}{RGB}{179,0,0}
\definecolor{OrRd-7-1}{RGB}{254,240,217}
\definecolor{OrRd-7-B}{RGB}{254,240,217}
\definecolor{OrRd-7-2}{RGB}{253,212,158}
\definecolor{OrRd-7-D}{RGB}{253,212,158}
\definecolor{OrRd-7-3}{RGB}{253,187,132}
\definecolor{OrRd-7-F}{RGB}{253,187,132}
\definecolor{OrRd-7-4}{RGB}{252,141,89}
\definecolor{OrRd-7-G}{RGB}{252,141,89}
\definecolor{OrRd-7-5}{RGB}{239,101,72}
\definecolor{OrRd-7-H}{RGB}{239,101,72}
\definecolor{OrRd-7-6}{RGB}{215,48,31}
\definecolor{OrRd-7-J}{RGB}{215,48,31}
\definecolor{OrRd-7-7}{RGB}{153,0,0}
\definecolor{OrRd-7-L}{RGB}{153,0,0}
\definecolor{OrRd-8-1}{RGB}{255,247,236}
\definecolor{OrRd-8-A}{RGB}{255,247,236}
\definecolor{OrRd-8-2}{RGB}{254,232,200}
\definecolor{OrRd-8-C}{RGB}{254,232,200}
\definecolor{OrRd-8-3}{RGB}{253,212,158}
\definecolor{OrRd-8-D}{RGB}{253,212,158}
\definecolor{OrRd-8-4}{RGB}{253,187,132}
\definecolor{OrRd-8-F}{RGB}{253,187,132}
\definecolor{OrRd-8-5}{RGB}{252,141,89}
\definecolor{OrRd-8-G}{RGB}{252,141,89}
\definecolor{OrRd-8-6}{RGB}{239,101,72}
\definecolor{OrRd-8-H}{RGB}{239,101,72}
\definecolor{OrRd-8-7}{RGB}{215,48,31}
\definecolor{OrRd-8-J}{RGB}{215,48,31}
\definecolor{OrRd-8-8}{RGB}{153,0,0}
\definecolor{OrRd-8-L}{RGB}{153,0,0}
\definecolor{OrRd-9-1}{RGB}{255,247,236}
\definecolor{OrRd-9-A}{RGB}{255,247,236}
\definecolor{OrRd-9-2}{RGB}{254,232,200}
\definecolor{OrRd-9-C}{RGB}{254,232,200}
\definecolor{OrRd-9-3}{RGB}{253,212,158}
\definecolor{OrRd-9-D}{RGB}{253,212,158}
\definecolor{OrRd-9-4}{RGB}{253,187,132}
\definecolor{OrRd-9-F}{RGB}{253,187,132}
\definecolor{OrRd-9-5}{RGB}{252,141,89}
\definecolor{OrRd-9-G}{RGB}{252,141,89}
\definecolor{OrRd-9-6}{RGB}{239,101,72}
\definecolor{OrRd-9-H}{RGB}{239,101,72}
\definecolor{OrRd-9-7}{RGB}{215,48,31}
\definecolor{OrRd-9-J}{RGB}{215,48,31}
\definecolor{OrRd-9-8}{RGB}{179,0,0}
\definecolor{OrRd-9-K}{RGB}{179,0,0}
\definecolor{OrRd-9-9}{RGB}{127,0,0}
\definecolor{OrRd-9-M}{RGB}{127,0,0}
\definecolor{YlOrRd-3-1}{RGB}{255,237,160}
\definecolor{YlOrRd-3-C}{RGB}{255,237,160}
\definecolor{YlOrRd-3-2}{RGB}{254,178,76}
\definecolor{YlOrRd-3-F}{RGB}{254,178,76}
\definecolor{YlOrRd-3-3}{RGB}{240,59,32}
\definecolor{YlOrRd-3-I}{RGB}{240,59,32}
\definecolor{YlOrRd-4-1}{RGB}{255,255,178}
\definecolor{YlOrRd-4-B}{RGB}{255,255,178}
\definecolor{YlOrRd-4-2}{RGB}{254,204,92}
\definecolor{YlOrRd-4-E}{RGB}{254,204,92}
\definecolor{YlOrRd-4-3}{RGB}{253,141,60}
\definecolor{YlOrRd-4-G}{RGB}{253,141,60}
\definecolor{YlOrRd-4-4}{RGB}{227,26,28}
\definecolor{YlOrRd-4-J}{RGB}{227,26,28}
\definecolor{YlOrRd-5-1}{RGB}{255,255,178}
\definecolor{YlOrRd-5-B}{RGB}{255,255,178}
\definecolor{YlOrRd-5-2}{RGB}{254,204,92}
\definecolor{YlOrRd-5-E}{RGB}{254,204,92}
\definecolor{YlOrRd-5-3}{RGB}{253,141,60}
\definecolor{YlOrRd-5-G}{RGB}{253,141,60}
\definecolor{YlOrRd-5-4}{RGB}{240,59,32}
\definecolor{YlOrRd-5-I}{RGB}{240,59,32}
\definecolor{YlOrRd-5-5}{RGB}{189,0,38}
\definecolor{YlOrRd-5-K}{RGB}{189,0,38}
\definecolor{YlOrRd-6-1}{RGB}{255,255,178}
\definecolor{YlOrRd-6-B}{RGB}{255,255,178}
\definecolor{YlOrRd-6-2}{RGB}{254,217,118}
\definecolor{YlOrRd-6-D}{RGB}{254,217,118}
\definecolor{YlOrRd-6-3}{RGB}{254,178,76}
\definecolor{YlOrRd-6-F}{RGB}{254,178,76}
\definecolor{YlOrRd-6-4}{RGB}{253,141,60}
\definecolor{YlOrRd-6-G}{RGB}{253,141,60}
\definecolor{YlOrRd-6-5}{RGB}{240,59,32}
\definecolor{YlOrRd-6-I}{RGB}{240,59,32}
\definecolor{YlOrRd-6-6}{RGB}{189,0,38}
\definecolor{YlOrRd-6-K}{RGB}{189,0,38}
\definecolor{YlOrRd-7-1}{RGB}{255,255,178}
\definecolor{YlOrRd-7-B}{RGB}{255,255,178}
\definecolor{YlOrRd-7-2}{RGB}{254,217,118}
\definecolor{YlOrRd-7-D}{RGB}{254,217,118}
\definecolor{YlOrRd-7-3}{RGB}{254,178,76}
\definecolor{YlOrRd-7-F}{RGB}{254,178,76}
\definecolor{YlOrRd-7-4}{RGB}{253,141,60}
\definecolor{YlOrRd-7-G}{RGB}{253,141,60}
\definecolor{YlOrRd-7-5}{RGB}{252,78,42}
\definecolor{YlOrRd-7-H}{RGB}{252,78,42}
\definecolor{YlOrRd-7-6}{RGB}{227,26,28}
\definecolor{YlOrRd-7-J}{RGB}{227,26,28}
\definecolor{YlOrRd-7-7}{RGB}{177,0,38}
\definecolor{YlOrRd-7-L}{RGB}{177,0,38}
\definecolor{YlOrRd-8-1}{RGB}{255,255,204}
\definecolor{YlOrRd-8-A}{RGB}{255,255,204}
\definecolor{YlOrRd-8-2}{RGB}{255,237,160}
\definecolor{YlOrRd-8-C}{RGB}{255,237,160}
\definecolor{YlOrRd-8-3}{RGB}{254,217,118}
\definecolor{YlOrRd-8-D}{RGB}{254,217,118}
\definecolor{YlOrRd-8-4}{RGB}{254,178,76}
\definecolor{YlOrRd-8-F}{RGB}{254,178,76}
\definecolor{YlOrRd-8-5}{RGB}{253,141,60}
\definecolor{YlOrRd-8-G}{RGB}{253,141,60}
\definecolor{YlOrRd-8-6}{RGB}{252,78,42}
\definecolor{YlOrRd-8-H}{RGB}{252,78,42}
\definecolor{YlOrRd-8-7}{RGB}{227,26,28}
\definecolor{YlOrRd-8-J}{RGB}{227,26,28}
\definecolor{YlOrRd-8-8}{RGB}{177,0,38}
\definecolor{YlOrRd-8-L}{RGB}{177,0,38}
\definecolor{YlOrRd-9-1}{RGB}{255,255,204}
\definecolor{YlOrRd-9-A}{RGB}{255,255,204}
\definecolor{YlOrRd-9-2}{RGB}{255,237,160}
\definecolor{YlOrRd-9-C}{RGB}{255,237,160}
\definecolor{YlOrRd-9-3}{RGB}{254,217,118}
\definecolor{YlOrRd-9-D}{RGB}{254,217,118}
\definecolor{YlOrRd-9-4}{RGB}{254,178,76}
\definecolor{YlOrRd-9-F}{RGB}{254,178,76}
\definecolor{YlOrRd-9-5}{RGB}{253,141,60}
\definecolor{YlOrRd-9-G}{RGB}{253,141,60}
\definecolor{YlOrRd-9-6}{RGB}{252,78,42}
\definecolor{YlOrRd-9-H}{RGB}{252,78,42}
\definecolor{YlOrRd-9-7}{RGB}{227,26,28}
\definecolor{YlOrRd-9-J}{RGB}{227,26,28}
\definecolor{YlOrRd-9-8}{RGB}{189,0,38}
\definecolor{YlOrRd-9-K}{RGB}{189,0,38}
\definecolor{YlOrRd-9-9}{RGB}{128,0,38}
\definecolor{YlOrRd-9-M}{RGB}{128,0,38}
\definecolor{YlOrBr-3-1}{RGB}{255,247,188}
\definecolor{YlOrBr-3-C}{RGB}{255,247,188}
\definecolor{YlOrBr-3-2}{RGB}{254,196,79}
\definecolor{YlOrBr-3-F}{RGB}{254,196,79}
\definecolor{YlOrBr-3-3}{RGB}{217,95,14}
\definecolor{YlOrBr-3-I}{RGB}{217,95,14}
\definecolor{YlOrBr-4-1}{RGB}{255,255,212}
\definecolor{YlOrBr-4-B}{RGB}{255,255,212}
\definecolor{YlOrBr-4-2}{RGB}{254,217,142}
\definecolor{YlOrBr-4-E}{RGB}{254,217,142}
\definecolor{YlOrBr-4-3}{RGB}{254,153,41}
\definecolor{YlOrBr-4-G}{RGB}{254,153,41}
\definecolor{YlOrBr-4-4}{RGB}{204,76,2}
\definecolor{YlOrBr-4-J}{RGB}{204,76,2}
\definecolor{YlOrBr-5-1}{RGB}{255,255,212}
\definecolor{YlOrBr-5-B}{RGB}{255,255,212}
\definecolor{YlOrBr-5-2}{RGB}{254,217,142}
\definecolor{YlOrBr-5-E}{RGB}{254,217,142}
\definecolor{YlOrBr-5-3}{RGB}{254,153,41}
\definecolor{YlOrBr-5-G}{RGB}{254,153,41}
\definecolor{YlOrBr-5-4}{RGB}{217,95,14}
\definecolor{YlOrBr-5-I}{RGB}{217,95,14}
\definecolor{YlOrBr-5-5}{RGB}{153,52,4}
\definecolor{YlOrBr-5-K}{RGB}{153,52,4}
\definecolor{YlOrBr-6-1}{RGB}{255,255,212}
\definecolor{YlOrBr-6-B}{RGB}{255,255,212}
\definecolor{YlOrBr-6-2}{RGB}{254,227,145}
\definecolor{YlOrBr-6-D}{RGB}{254,227,145}
\definecolor{YlOrBr-6-3}{RGB}{254,196,79}
\definecolor{YlOrBr-6-F}{RGB}{254,196,79}
\definecolor{YlOrBr-6-4}{RGB}{254,153,41}
\definecolor{YlOrBr-6-G}{RGB}{254,153,41}
\definecolor{YlOrBr-6-5}{RGB}{217,95,14}
\definecolor{YlOrBr-6-I}{RGB}{217,95,14}
\definecolor{YlOrBr-6-6}{RGB}{153,52,4}
\definecolor{YlOrBr-6-K}{RGB}{153,52,4}
\definecolor{YlOrBr-7-1}{RGB}{255,255,212}
\definecolor{YlOrBr-7-B}{RGB}{255,255,212}
\definecolor{YlOrBr-7-2}{RGB}{254,227,145}
\definecolor{YlOrBr-7-D}{RGB}{254,227,145}
\definecolor{YlOrBr-7-3}{RGB}{254,196,79}
\definecolor{YlOrBr-7-F}{RGB}{254,196,79}
\definecolor{YlOrBr-7-4}{RGB}{254,153,41}
\definecolor{YlOrBr-7-G}{RGB}{254,153,41}
\definecolor{YlOrBr-7-5}{RGB}{236,112,20}
\definecolor{YlOrBr-7-H}{RGB}{236,112,20}
\definecolor{YlOrBr-7-6}{RGB}{204,76,2}
\definecolor{YlOrBr-7-J}{RGB}{204,76,2}
\definecolor{YlOrBr-7-7}{RGB}{140,45,4}
\definecolor{YlOrBr-7-L}{RGB}{140,45,4}
\definecolor{YlOrBr-8-1}{RGB}{255,255,229}
\definecolor{YlOrBr-8-A}{RGB}{255,255,229}
\definecolor{YlOrBr-8-2}{RGB}{255,247,188}
\definecolor{YlOrBr-8-C}{RGB}{255,247,188}
\definecolor{YlOrBr-8-3}{RGB}{254,227,145}
\definecolor{YlOrBr-8-D}{RGB}{254,227,145}
\definecolor{YlOrBr-8-4}{RGB}{254,196,79}
\definecolor{YlOrBr-8-F}{RGB}{254,196,79}
\definecolor{YlOrBr-8-5}{RGB}{254,153,41}
\definecolor{YlOrBr-8-G}{RGB}{254,153,41}
\definecolor{YlOrBr-8-6}{RGB}{236,112,20}
\definecolor{YlOrBr-8-H}{RGB}{236,112,20}
\definecolor{YlOrBr-8-7}{RGB}{204,76,2}
\definecolor{YlOrBr-8-J}{RGB}{204,76,2}
\definecolor{YlOrBr-8-8}{RGB}{140,45,4}
\definecolor{YlOrBr-8-L}{RGB}{140,45,4}
\definecolor{YlOrBr-9-1}{RGB}{255,255,229}
\definecolor{YlOrBr-9-A}{RGB}{255,255,229}
\definecolor{YlOrBr-9-2}{RGB}{255,247,188}
\definecolor{YlOrBr-9-C}{RGB}{255,247,188}
\definecolor{YlOrBr-9-3}{RGB}{254,227,145}
\definecolor{YlOrBr-9-D}{RGB}{254,227,145}
\definecolor{YlOrBr-9-4}{RGB}{254,196,79}
\definecolor{YlOrBr-9-F}{RGB}{254,196,79}
\definecolor{YlOrBr-9-5}{RGB}{254,153,41}
\definecolor{YlOrBr-9-G}{RGB}{254,153,41}
\definecolor{YlOrBr-9-6}{RGB}{236,112,20}
\definecolor{YlOrBr-9-H}{RGB}{236,112,20}
\definecolor{YlOrBr-9-7}{RGB}{204,76,2}
\definecolor{YlOrBr-9-J}{RGB}{204,76,2}
\definecolor{YlOrBr-9-8}{RGB}{153,52,4}
\definecolor{YlOrBr-9-K}{RGB}{153,52,4}
\definecolor{YlOrBr-9-9}{RGB}{102,37,6}
\definecolor{YlOrBr-9-M}{RGB}{102,37,6}
\definecolor{Purples-3-1}{RGB}{239,237,245}
\definecolor{Purples-3-C}{RGB}{239,237,245}
\definecolor{Purples-3-2}{RGB}{188,189,220}
\definecolor{Purples-3-F}{RGB}{188,189,220}
\definecolor{Purples-3-3}{RGB}{117,107,177}
\definecolor{Purples-3-I}{RGB}{117,107,177}
\definecolor{Purples-4-1}{RGB}{242,240,247}
\definecolor{Purples-4-B}{RGB}{242,240,247}
\definecolor{Purples-4-2}{RGB}{203,201,226}
\definecolor{Purples-4-E}{RGB}{203,201,226}
\definecolor{Purples-4-3}{RGB}{158,154,200}
\definecolor{Purples-4-G}{RGB}{158,154,200}
\definecolor{Purples-4-4}{RGB}{106,81,163}
\definecolor{Purples-4-J}{RGB}{106,81,163}
\definecolor{Purples-5-1}{RGB}{242,240,247}
\definecolor{Purples-5-B}{RGB}{242,240,247}
\definecolor{Purples-5-2}{RGB}{203,201,226}
\definecolor{Purples-5-E}{RGB}{203,201,226}
\definecolor{Purples-5-3}{RGB}{158,154,200}
\definecolor{Purples-5-G}{RGB}{158,154,200}
\definecolor{Purples-5-4}{RGB}{117,107,177}
\definecolor{Purples-5-I}{RGB}{117,107,177}
\definecolor{Purples-5-5}{RGB}{84,39,143}
\definecolor{Purples-5-K}{RGB}{84,39,143}
\definecolor{Purples-6-1}{RGB}{242,240,247}
\definecolor{Purples-6-B}{RGB}{242,240,247}
\definecolor{Purples-6-2}{RGB}{218,218,235}
\definecolor{Purples-6-D}{RGB}{218,218,235}
\definecolor{Purples-6-3}{RGB}{188,189,220}
\definecolor{Purples-6-F}{RGB}{188,189,220}
\definecolor{Purples-6-4}{RGB}{158,154,200}
\definecolor{Purples-6-G}{RGB}{158,154,200}
\definecolor{Purples-6-5}{RGB}{117,107,177}
\definecolor{Purples-6-I}{RGB}{117,107,177}
\definecolor{Purples-6-6}{RGB}{84,39,143}
\definecolor{Purples-6-K}{RGB}{84,39,143}
\definecolor{Purples-7-1}{RGB}{242,240,247}
\definecolor{Purples-7-B}{RGB}{242,240,247}
\definecolor{Purples-7-2}{RGB}{218,218,235}
\definecolor{Purples-7-D}{RGB}{218,218,235}
\definecolor{Purples-7-3}{RGB}{188,189,220}
\definecolor{Purples-7-F}{RGB}{188,189,220}
\definecolor{Purples-7-4}{RGB}{158,154,200}
\definecolor{Purples-7-G}{RGB}{158,154,200}
\definecolor{Purples-7-5}{RGB}{128,125,186}
\definecolor{Purples-7-H}{RGB}{128,125,186}
\definecolor{Purples-7-6}{RGB}{106,81,163}
\definecolor{Purples-7-J}{RGB}{106,81,163}
\definecolor{Purples-7-7}{RGB}{74,20,134}
\definecolor{Purples-7-L}{RGB}{74,20,134}
\definecolor{Purples-8-1}{RGB}{252,251,253}
\definecolor{Purples-8-A}{RGB}{252,251,253}
\definecolor{Purples-8-2}{RGB}{239,237,245}
\definecolor{Purples-8-C}{RGB}{239,237,245}
\definecolor{Purples-8-3}{RGB}{218,218,235}
\definecolor{Purples-8-D}{RGB}{218,218,235}
\definecolor{Purples-8-4}{RGB}{188,189,220}
\definecolor{Purples-8-F}{RGB}{188,189,220}
\definecolor{Purples-8-5}{RGB}{158,154,200}
\definecolor{Purples-8-G}{RGB}{158,154,200}
\definecolor{Purples-8-6}{RGB}{128,125,186}
\definecolor{Purples-8-H}{RGB}{128,125,186}
\definecolor{Purples-8-7}{RGB}{106,81,163}
\definecolor{Purples-8-J}{RGB}{106,81,163}
\definecolor{Purples-8-8}{RGB}{74,20,134}
\definecolor{Purples-8-L}{RGB}{74,20,134}
\definecolor{Purples-9-1}{RGB}{252,251,253}
\definecolor{Purples-9-A}{RGB}{252,251,253}
\definecolor{Purples-9-2}{RGB}{239,237,245}
\definecolor{Purples-9-C}{RGB}{239,237,245}
\definecolor{Purples-9-3}{RGB}{218,218,235}
\definecolor{Purples-9-D}{RGB}{218,218,235}
\definecolor{Purples-9-4}{RGB}{188,189,220}
\definecolor{Purples-9-F}{RGB}{188,189,220}
\definecolor{Purples-9-5}{RGB}{158,154,200}
\definecolor{Purples-9-G}{RGB}{158,154,200}
\definecolor{Purples-9-6}{RGB}{128,125,186}
\definecolor{Purples-9-H}{RGB}{128,125,186}
\definecolor{Purples-9-7}{RGB}{106,81,163}
\definecolor{Purples-9-J}{RGB}{106,81,163}
\definecolor{Purples-9-8}{RGB}{84,39,143}
\definecolor{Purples-9-K}{RGB}{84,39,143}
\definecolor{Purples-9-9}{RGB}{63,0,125}
\definecolor{Purples-9-M}{RGB}{63,0,125}
\definecolor{Blues-3-1}{RGB}{222,235,247}
\definecolor{Blues-3-C}{RGB}{222,235,247}
\definecolor{Blues-3-2}{RGB}{158,202,225}
\definecolor{Blues-3-F}{RGB}{158,202,225}
\definecolor{Blues-3-3}{RGB}{49,130,189}
\definecolor{Blues-3-I}{RGB}{49,130,189}
\definecolor{Blues-4-1}{RGB}{239,243,255}
\definecolor{Blues-4-B}{RGB}{239,243,255}
\definecolor{Blues-4-2}{RGB}{189,215,231}
\definecolor{Blues-4-E}{RGB}{189,215,231}
\definecolor{Blues-4-3}{RGB}{107,174,214}
\definecolor{Blues-4-G}{RGB}{107,174,214}
\definecolor{Blues-4-4}{RGB}{33,113,181}
\definecolor{Blues-4-J}{RGB}{33,113,181}
\definecolor{Blues-5-1}{RGB}{239,243,255}
\definecolor{Blues-5-B}{RGB}{239,243,255}
\definecolor{Blues-5-2}{RGB}{189,215,231}
\definecolor{Blues-5-E}{RGB}{189,215,231}
\definecolor{Blues-5-3}{RGB}{107,174,214}
\definecolor{Blues-5-G}{RGB}{107,174,214}
\definecolor{Blues-5-4}{RGB}{49,130,189}
\definecolor{Blues-5-I}{RGB}{49,130,189}
\definecolor{Blues-5-5}{RGB}{8,81,156}
\definecolor{Blues-5-K}{RGB}{8,81,156}
\definecolor{Blues-6-1}{RGB}{239,243,255}
\definecolor{Blues-6-B}{RGB}{239,243,255}
\definecolor{Blues-6-2}{RGB}{198,219,239}
\definecolor{Blues-6-D}{RGB}{198,219,239}
\definecolor{Blues-6-3}{RGB}{158,202,225}
\definecolor{Blues-6-F}{RGB}{158,202,225}
\definecolor{Blues-6-4}{RGB}{107,174,214}
\definecolor{Blues-6-G}{RGB}{107,174,214}
\definecolor{Blues-6-5}{RGB}{49,130,189}
\definecolor{Blues-6-I}{RGB}{49,130,189}
\definecolor{Blues-6-6}{RGB}{8,81,156}
\definecolor{Blues-6-K}{RGB}{8,81,156}
\definecolor{Blues-7-1}{RGB}{239,243,255}
\definecolor{Blues-7-B}{RGB}{239,243,255}
\definecolor{Blues-7-2}{RGB}{198,219,239}
\definecolor{Blues-7-D}{RGB}{198,219,239}
\definecolor{Blues-7-3}{RGB}{158,202,225}
\definecolor{Blues-7-F}{RGB}{158,202,225}
\definecolor{Blues-7-4}{RGB}{107,174,214}
\definecolor{Blues-7-G}{RGB}{107,174,214}
\definecolor{Blues-7-5}{RGB}{66,146,198}
\definecolor{Blues-7-H}{RGB}{66,146,198}
\definecolor{Blues-7-6}{RGB}{33,113,181}
\definecolor{Blues-7-J}{RGB}{33,113,181}
\definecolor{Blues-7-7}{RGB}{8,69,148}
\definecolor{Blues-7-L}{RGB}{8,69,148}
\definecolor{Blues-8-1}{RGB}{247,251,255}
\definecolor{Blues-8-A}{RGB}{247,251,255}
\definecolor{Blues-8-2}{RGB}{222,235,247}
\definecolor{Blues-8-C}{RGB}{222,235,247}
\definecolor{Blues-8-3}{RGB}{198,219,239}
\definecolor{Blues-8-D}{RGB}{198,219,239}
\definecolor{Blues-8-4}{RGB}{158,202,225}
\definecolor{Blues-8-F}{RGB}{158,202,225}
\definecolor{Blues-8-5}{RGB}{107,174,214}
\definecolor{Blues-8-G}{RGB}{107,174,214}
\definecolor{Blues-8-6}{RGB}{66,146,198}
\definecolor{Blues-8-H}{RGB}{66,146,198}
\definecolor{Blues-8-7}{RGB}{33,113,181}
\definecolor{Blues-8-J}{RGB}{33,113,181}
\definecolor{Blues-8-8}{RGB}{8,69,148}
\definecolor{Blues-8-L}{RGB}{8,69,148}
\definecolor{Blues-9-1}{RGB}{247,251,255}
\definecolor{Blues-9-A}{RGB}{247,251,255}
\definecolor{Blues-9-2}{RGB}{222,235,247}
\definecolor{Blues-9-C}{RGB}{222,235,247}
\definecolor{Blues-9-3}{RGB}{198,219,239}
\definecolor{Blues-9-D}{RGB}{198,219,239}
\definecolor{Blues-9-4}{RGB}{158,202,225}
\definecolor{Blues-9-F}{RGB}{158,202,225}
\definecolor{Blues-9-5}{RGB}{107,174,214}
\definecolor{Blues-9-G}{RGB}{107,174,214}
\definecolor{Blues-9-6}{RGB}{66,146,198}
\definecolor{Blues-9-H}{RGB}{66,146,198}
\definecolor{Blues-9-7}{RGB}{33,113,181}
\definecolor{Blues-9-J}{RGB}{33,113,181}
\definecolor{Blues-9-8}{RGB}{8,81,156}
\definecolor{Blues-9-K}{RGB}{8,81,156}
\definecolor{Blues-9-9}{RGB}{8,48,107}
\definecolor{Blues-9-M}{RGB}{8,48,107}
\definecolor{Greens-3-1}{RGB}{229,245,224}
\definecolor{Greens-3-C}{RGB}{229,245,224}
\definecolor{Greens-3-2}{RGB}{161,217,155}
\definecolor{Greens-3-F}{RGB}{161,217,155}
\definecolor{Greens-3-3}{RGB}{49,163,84}
\definecolor{Greens-3-I}{RGB}{49,163,84}
\definecolor{Greens-4-1}{RGB}{237,248,233}
\definecolor{Greens-4-B}{RGB}{237,248,233}
\definecolor{Greens-4-2}{RGB}{186,228,179}
\definecolor{Greens-4-E}{RGB}{186,228,179}
\definecolor{Greens-4-3}{RGB}{116,196,118}
\definecolor{Greens-4-G}{RGB}{116,196,118}
\definecolor{Greens-4-4}{RGB}{35,139,69}
\definecolor{Greens-4-J}{RGB}{35,139,69}
\definecolor{Greens-5-1}{RGB}{237,248,233}
\definecolor{Greens-5-B}{RGB}{237,248,233}
\definecolor{Greens-5-2}{RGB}{186,228,179}
\definecolor{Greens-5-E}{RGB}{186,228,179}
\definecolor{Greens-5-3}{RGB}{116,196,118}
\definecolor{Greens-5-G}{RGB}{116,196,118}
\definecolor{Greens-5-4}{RGB}{49,163,84}
\definecolor{Greens-5-I}{RGB}{49,163,84}
\definecolor{Greens-5-5}{RGB}{0,109,44}
\definecolor{Greens-5-K}{RGB}{0,109,44}
\definecolor{Greens-6-1}{RGB}{237,248,233}
\definecolor{Greens-6-B}{RGB}{237,248,233}
\definecolor{Greens-6-2}{RGB}{199,233,192}
\definecolor{Greens-6-D}{RGB}{199,233,192}
\definecolor{Greens-6-3}{RGB}{161,217,155}
\definecolor{Greens-6-F}{RGB}{161,217,155}
\definecolor{Greens-6-4}{RGB}{116,196,118}
\definecolor{Greens-6-G}{RGB}{116,196,118}
\definecolor{Greens-6-5}{RGB}{49,163,84}
\definecolor{Greens-6-I}{RGB}{49,163,84}
\definecolor{Greens-6-6}{RGB}{0,109,44}
\definecolor{Greens-6-K}{RGB}{0,109,44}
\definecolor{Greens-7-1}{RGB}{237,248,233}
\definecolor{Greens-7-B}{RGB}{237,248,233}
\definecolor{Greens-7-2}{RGB}{199,233,192}
\definecolor{Greens-7-D}{RGB}{199,233,192}
\definecolor{Greens-7-3}{RGB}{161,217,155}
\definecolor{Greens-7-F}{RGB}{161,217,155}
\definecolor{Greens-7-4}{RGB}{116,196,118}
\definecolor{Greens-7-G}{RGB}{116,196,118}
\definecolor{Greens-7-5}{RGB}{65,171,93}
\definecolor{Greens-7-H}{RGB}{65,171,93}
\definecolor{Greens-7-6}{RGB}{35,139,69}
\definecolor{Greens-7-J}{RGB}{35,139,69}
\definecolor{Greens-7-7}{RGB}{0,90,50}
\definecolor{Greens-7-L}{RGB}{0,90,50}
\definecolor{Greens-8-1}{RGB}{247,252,245}
\definecolor{Greens-8-A}{RGB}{247,252,245}
\definecolor{Greens-8-2}{RGB}{229,245,224}
\definecolor{Greens-8-C}{RGB}{229,245,224}
\definecolor{Greens-8-3}{RGB}{199,233,192}
\definecolor{Greens-8-D}{RGB}{199,233,192}
\definecolor{Greens-8-4}{RGB}{161,217,155}
\definecolor{Greens-8-F}{RGB}{161,217,155}
\definecolor{Greens-8-5}{RGB}{116,196,118}
\definecolor{Greens-8-G}{RGB}{116,196,118}
\definecolor{Greens-8-6}{RGB}{65,171,93}
\definecolor{Greens-8-H}{RGB}{65,171,93}
\definecolor{Greens-8-7}{RGB}{35,139,69}
\definecolor{Greens-8-J}{RGB}{35,139,69}
\definecolor{Greens-8-8}{RGB}{0,90,50}
\definecolor{Greens-8-L}{RGB}{0,90,50}
\definecolor{Greens-9-1}{RGB}{247,252,245}
\definecolor{Greens-9-A}{RGB}{247,252,245}
\definecolor{Greens-9-2}{RGB}{229,245,224}
\definecolor{Greens-9-C}{RGB}{229,245,224}
\definecolor{Greens-9-3}{RGB}{199,233,192}
\definecolor{Greens-9-D}{RGB}{199,233,192}
\definecolor{Greens-9-4}{RGB}{161,217,155}
\definecolor{Greens-9-F}{RGB}{161,217,155}
\definecolor{Greens-9-5}{RGB}{116,196,118}
\definecolor{Greens-9-G}{RGB}{116,196,118}
\definecolor{Greens-9-6}{RGB}{65,171,93}
\definecolor{Greens-9-H}{RGB}{65,171,93}
\definecolor{Greens-9-7}{RGB}{35,139,69}
\definecolor{Greens-9-J}{RGB}{35,139,69}
\definecolor{Greens-9-8}{RGB}{0,109,44}
\definecolor{Greens-9-K}{RGB}{0,109,44}
\definecolor{Greens-9-9}{RGB}{0,68,27}
\definecolor{Greens-9-M}{RGB}{0,68,27}
\definecolor{Oranges-3-1}{RGB}{254,230,206}
\definecolor{Oranges-3-C}{RGB}{254,230,206}
\definecolor{Oranges-3-2}{RGB}{253,174,107}
\definecolor{Oranges-3-F}{RGB}{253,174,107}
\definecolor{Oranges-3-3}{RGB}{230,85,13}
\definecolor{Oranges-3-I}{RGB}{230,85,13}
\definecolor{Oranges-4-1}{RGB}{254,237,222}
\definecolor{Oranges-4-B}{RGB}{254,237,222}
\definecolor{Oranges-4-2}{RGB}{253,190,133}
\definecolor{Oranges-4-E}{RGB}{253,190,133}
\definecolor{Oranges-4-3}{RGB}{253,141,60}
\definecolor{Oranges-4-G}{RGB}{253,141,60}
\definecolor{Oranges-4-4}{RGB}{217,71,1}
\definecolor{Oranges-4-J}{RGB}{217,71,1}
\definecolor{Oranges-5-1}{RGB}{254,237,222}
\definecolor{Oranges-5-B}{RGB}{254,237,222}
\definecolor{Oranges-5-2}{RGB}{253,190,133}
\definecolor{Oranges-5-E}{RGB}{253,190,133}
\definecolor{Oranges-5-3}{RGB}{253,141,60}
\definecolor{Oranges-5-G}{RGB}{253,141,60}
\definecolor{Oranges-5-4}{RGB}{230,85,13}
\definecolor{Oranges-5-I}{RGB}{230,85,13}
\definecolor{Oranges-5-5}{RGB}{166,54,3}
\definecolor{Oranges-5-K}{RGB}{166,54,3}
\definecolor{Oranges-6-1}{RGB}{254,237,222}
\definecolor{Oranges-6-B}{RGB}{254,237,222}
\definecolor{Oranges-6-2}{RGB}{253,208,162}
\definecolor{Oranges-6-D}{RGB}{253,208,162}
\definecolor{Oranges-6-3}{RGB}{253,174,107}
\definecolor{Oranges-6-F}{RGB}{253,174,107}
\definecolor{Oranges-6-4}{RGB}{253,141,60}
\definecolor{Oranges-6-G}{RGB}{253,141,60}
\definecolor{Oranges-6-5}{RGB}{230,85,13}
\definecolor{Oranges-6-I}{RGB}{230,85,13}
\definecolor{Oranges-6-6}{RGB}{166,54,3}
\definecolor{Oranges-6-K}{RGB}{166,54,3}
\definecolor{Oranges-7-1}{RGB}{254,237,222}
\definecolor{Oranges-7-B}{RGB}{254,237,222}
\definecolor{Oranges-7-2}{RGB}{253,208,162}
\definecolor{Oranges-7-D}{RGB}{253,208,162}
\definecolor{Oranges-7-3}{RGB}{253,174,107}
\definecolor{Oranges-7-F}{RGB}{253,174,107}
\definecolor{Oranges-7-4}{RGB}{253,141,60}
\definecolor{Oranges-7-G}{RGB}{253,141,60}
\definecolor{Oranges-7-5}{RGB}{241,105,19}
\definecolor{Oranges-7-H}{RGB}{241,105,19}
\definecolor{Oranges-7-6}{RGB}{217,72,1}
\definecolor{Oranges-7-J}{RGB}{217,72,1}
\definecolor{Oranges-7-7}{RGB}{140,45,4}
\definecolor{Oranges-7-L}{RGB}{140,45,4}
\definecolor{Oranges-8-1}{RGB}{255,245,235}
\definecolor{Oranges-8-A}{RGB}{255,245,235}
\definecolor{Oranges-8-2}{RGB}{254,230,206}
\definecolor{Oranges-8-C}{RGB}{254,230,206}
\definecolor{Oranges-8-3}{RGB}{253,208,162}
\definecolor{Oranges-8-D}{RGB}{253,208,162}
\definecolor{Oranges-8-4}{RGB}{253,174,107}
\definecolor{Oranges-8-F}{RGB}{253,174,107}
\definecolor{Oranges-8-5}{RGB}{253,141,60}
\definecolor{Oranges-8-G}{RGB}{253,141,60}
\definecolor{Oranges-8-6}{RGB}{241,105,19}
\definecolor{Oranges-8-H}{RGB}{241,105,19}
\definecolor{Oranges-8-7}{RGB}{217,72,1}
\definecolor{Oranges-8-J}{RGB}{217,72,1}
\definecolor{Oranges-8-8}{RGB}{140,45,4}
\definecolor{Oranges-8-L}{RGB}{140,45,4}
\definecolor{Oranges-9-1}{RGB}{255,245,235}
\definecolor{Oranges-9-A}{RGB}{255,245,235}
\definecolor{Oranges-9-2}{RGB}{254,230,206}
\definecolor{Oranges-9-C}{RGB}{254,230,206}
\definecolor{Oranges-9-3}{RGB}{253,208,162}
\definecolor{Oranges-9-D}{RGB}{253,208,162}
\definecolor{Oranges-9-4}{RGB}{253,174,107}
\definecolor{Oranges-9-F}{RGB}{253,174,107}
\definecolor{Oranges-9-5}{RGB}{253,141,60}
\definecolor{Oranges-9-G}{RGB}{253,141,60}
\definecolor{Oranges-9-6}{RGB}{241,105,19}
\definecolor{Oranges-9-H}{RGB}{241,105,19}
\definecolor{Oranges-9-7}{RGB}{217,72,1}
\definecolor{Oranges-9-J}{RGB}{217,72,1}
\definecolor{Oranges-9-8}{RGB}{166,54,3}
\definecolor{Oranges-9-K}{RGB}{166,54,3}
\definecolor{Oranges-9-9}{RGB}{127,39,4}
\definecolor{Oranges-9-M}{RGB}{127,39,4}
\definecolor{Reds-3-1}{RGB}{254,224,210}
\definecolor{Reds-3-C}{RGB}{254,224,210}
\definecolor{Reds-3-2}{RGB}{252,146,114}
\definecolor{Reds-3-F}{RGB}{252,146,114}
\definecolor{Reds-3-3}{RGB}{222,45,38}
\definecolor{Reds-3-I}{RGB}{222,45,38}
\definecolor{Reds-4-1}{RGB}{254,229,217}
\definecolor{Reds-4-B}{RGB}{254,229,217}
\definecolor{Reds-4-2}{RGB}{252,174,145}
\definecolor{Reds-4-E}{RGB}{252,174,145}
\definecolor{Reds-4-3}{RGB}{251,106,74}
\definecolor{Reds-4-G}{RGB}{251,106,74}
\definecolor{Reds-4-4}{RGB}{203,24,29}
\definecolor{Reds-4-J}{RGB}{203,24,29}
\definecolor{Reds-5-1}{RGB}{254,229,217}
\definecolor{Reds-5-B}{RGB}{254,229,217}
\definecolor{Reds-5-2}{RGB}{252,174,145}
\definecolor{Reds-5-E}{RGB}{252,174,145}
\definecolor{Reds-5-3}{RGB}{251,106,74}
\definecolor{Reds-5-G}{RGB}{251,106,74}
\definecolor{Reds-5-4}{RGB}{222,45,38}
\definecolor{Reds-5-I}{RGB}{222,45,38}
\definecolor{Reds-5-5}{RGB}{165,15,21}
\definecolor{Reds-5-K}{RGB}{165,15,21}
\definecolor{Reds-6-1}{RGB}{254,229,217}
\definecolor{Reds-6-B}{RGB}{254,229,217}
\definecolor{Reds-6-2}{RGB}{252,187,161}
\definecolor{Reds-6-D}{RGB}{252,187,161}
\definecolor{Reds-6-3}{RGB}{252,146,114}
\definecolor{Reds-6-F}{RGB}{252,146,114}
\definecolor{Reds-6-4}{RGB}{251,106,74}
\definecolor{Reds-6-G}{RGB}{251,106,74}
\definecolor{Reds-6-5}{RGB}{222,45,38}
\definecolor{Reds-6-I}{RGB}{222,45,38}
\definecolor{Reds-6-6}{RGB}{165,15,21}
\definecolor{Reds-6-K}{RGB}{165,15,21}
\definecolor{Reds-7-1}{RGB}{254,229,217}
\definecolor{Reds-7-B}{RGB}{254,229,217}
\definecolor{Reds-7-2}{RGB}{252,187,161}
\definecolor{Reds-7-D}{RGB}{252,187,161}
\definecolor{Reds-7-3}{RGB}{252,146,114}
\definecolor{Reds-7-F}{RGB}{252,146,114}
\definecolor{Reds-7-4}{RGB}{251,106,74}
\definecolor{Reds-7-G}{RGB}{251,106,74}
\definecolor{Reds-7-5}{RGB}{239,59,44}
\definecolor{Reds-7-H}{RGB}{239,59,44}
\definecolor{Reds-7-6}{RGB}{203,24,29}
\definecolor{Reds-7-J}{RGB}{203,24,29}
\definecolor{Reds-7-7}{RGB}{153,0,13}
\definecolor{Reds-7-L}{RGB}{153,0,13}
\definecolor{Reds-8-1}{RGB}{255,245,240}
\definecolor{Reds-8-A}{RGB}{255,245,240}
\definecolor{Reds-8-2}{RGB}{254,224,210}
\definecolor{Reds-8-C}{RGB}{254,224,210}
\definecolor{Reds-8-3}{RGB}{252,187,161}
\definecolor{Reds-8-D}{RGB}{252,187,161}
\definecolor{Reds-8-4}{RGB}{252,146,114}
\definecolor{Reds-8-F}{RGB}{252,146,114}
\definecolor{Reds-8-5}{RGB}{251,106,74}
\definecolor{Reds-8-G}{RGB}{251,106,74}
\definecolor{Reds-8-6}{RGB}{239,59,44}
\definecolor{Reds-8-H}{RGB}{239,59,44}
\definecolor{Reds-8-7}{RGB}{203,24,29}
\definecolor{Reds-8-J}{RGB}{203,24,29}
\definecolor{Reds-8-8}{RGB}{153,0,13}
\definecolor{Reds-8-L}{RGB}{153,0,13}
\definecolor{Reds-9-1}{RGB}{255,245,240}
\definecolor{Reds-9-A}{RGB}{255,245,240}
\definecolor{Reds-9-2}{RGB}{254,224,210}
\definecolor{Reds-9-C}{RGB}{254,224,210}
\definecolor{Reds-9-3}{RGB}{252,187,161}
\definecolor{Reds-9-D}{RGB}{252,187,161}
\definecolor{Reds-9-4}{RGB}{252,146,114}
\definecolor{Reds-9-F}{RGB}{252,146,114}
\definecolor{Reds-9-5}{RGB}{251,106,74}
\definecolor{Reds-9-G}{RGB}{251,106,74}
\definecolor{Reds-9-6}{RGB}{239,59,44}
\definecolor{Reds-9-H}{RGB}{239,59,44}
\definecolor{Reds-9-7}{RGB}{203,24,29}
\definecolor{Reds-9-J}{RGB}{203,24,29}
\definecolor{Reds-9-8}{RGB}{165,15,21}
\definecolor{Reds-9-K}{RGB}{165,15,21}
\definecolor{Reds-9-9}{RGB}{103,0,13}
\definecolor{Reds-9-M}{RGB}{103,0,13}
\definecolor{Greys-3-1}{RGB}{240,240,240}
\definecolor{Greys-3-C}{RGB}{240,240,240}
\definecolor{Greys-3-2}{RGB}{189,189,189}
\definecolor{Greys-3-F}{RGB}{189,189,189}
\definecolor{Greys-3-3}{RGB}{99,99,99}
\definecolor{Greys-3-I}{RGB}{99,99,99}
\definecolor{Greys-4-1}{RGB}{247,247,247}
\definecolor{Greys-4-B}{RGB}{247,247,247}
\definecolor{Greys-4-2}{RGB}{204,204,204}
\definecolor{Greys-4-E}{RGB}{204,204,204}
\definecolor{Greys-4-3}{RGB}{150,150,150}
\definecolor{Greys-4-G}{RGB}{150,150,150}
\definecolor{Greys-4-4}{RGB}{82,82,82}
\definecolor{Greys-4-J}{RGB}{82,82,82}
\definecolor{Greys-5-1}{RGB}{247,247,247}
\definecolor{Greys-5-B}{RGB}{247,247,247}
\definecolor{Greys-5-2}{RGB}{204,204,204}
\definecolor{Greys-5-E}{RGB}{204,204,204}
\definecolor{Greys-5-3}{RGB}{150,150,150}
\definecolor{Greys-5-G}{RGB}{150,150,150}
\definecolor{Greys-5-4}{RGB}{99,99,99}
\definecolor{Greys-5-I}{RGB}{99,99,99}
\definecolor{Greys-5-5}{RGB}{37,37,37}
\definecolor{Greys-5-K}{RGB}{37,37,37}
\definecolor{Greys-6-1}{RGB}{247,247,247}
\definecolor{Greys-6-B}{RGB}{247,247,247}
\definecolor{Greys-6-2}{RGB}{217,217,217}
\definecolor{Greys-6-D}{RGB}{217,217,217}
\definecolor{Greys-6-3}{RGB}{189,189,189}
\definecolor{Greys-6-F}{RGB}{189,189,189}
\definecolor{Greys-6-4}{RGB}{150,150,150}
\definecolor{Greys-6-G}{RGB}{150,150,150}
\definecolor{Greys-6-5}{RGB}{99,99,99}
\definecolor{Greys-6-I}{RGB}{99,99,99}
\definecolor{Greys-6-6}{RGB}{37,37,37}
\definecolor{Greys-6-K}{RGB}{37,37,37}
\definecolor{Greys-7-1}{RGB}{247,247,247}
\definecolor{Greys-7-B}{RGB}{247,247,247}
\definecolor{Greys-7-2}{RGB}{217,217,217}
\definecolor{Greys-7-D}{RGB}{217,217,217}
\definecolor{Greys-7-3}{RGB}{189,189,189}
\definecolor{Greys-7-F}{RGB}{189,189,189}
\definecolor{Greys-7-4}{RGB}{150,150,150}
\definecolor{Greys-7-G}{RGB}{150,150,150}
\definecolor{Greys-7-5}{RGB}{115,115,115}
\definecolor{Greys-7-H}{RGB}{115,115,115}
\definecolor{Greys-7-6}{RGB}{82,82,82}
\definecolor{Greys-7-J}{RGB}{82,82,82}
\definecolor{Greys-7-7}{RGB}{37,37,37}
\definecolor{Greys-7-L}{RGB}{37,37,37}
\definecolor{Greys-8-1}{RGB}{255,255,255}
\definecolor{Greys-8-A}{RGB}{255,255,255}
\definecolor{Greys-8-2}{RGB}{240,240,240}
\definecolor{Greys-8-C}{RGB}{240,240,240}
\definecolor{Greys-8-3}{RGB}{217,217,217}
\definecolor{Greys-8-D}{RGB}{217,217,217}
\definecolor{Greys-8-4}{RGB}{189,189,189}
\definecolor{Greys-8-F}{RGB}{189,189,189}
\definecolor{Greys-8-5}{RGB}{150,150,150}
\definecolor{Greys-8-G}{RGB}{150,150,150}
\definecolor{Greys-8-6}{RGB}{115,115,115}
\definecolor{Greys-8-H}{RGB}{115,115,115}
\definecolor{Greys-8-7}{RGB}{82,82,82}
\definecolor{Greys-8-J}{RGB}{82,82,82}
\definecolor{Greys-8-8}{RGB}{37,37,37}
\definecolor{Greys-8-L}{RGB}{37,37,37}
\definecolor{Greys-9-1}{RGB}{255,255,255}
\definecolor{Greys-9-A}{RGB}{255,255,255}
\definecolor{Greys-9-2}{RGB}{240,240,240}
\definecolor{Greys-9-C}{RGB}{240,240,240}
\definecolor{Greys-9-3}{RGB}{217,217,217}
\definecolor{Greys-9-D}{RGB}{217,217,217}
\definecolor{Greys-9-4}{RGB}{189,189,189}
\definecolor{Greys-9-F}{RGB}{189,189,189}
\definecolor{Greys-9-5}{RGB}{150,150,150}
\definecolor{Greys-9-G}{RGB}{150,150,150}
\definecolor{Greys-9-6}{RGB}{115,115,115}
\definecolor{Greys-9-H}{RGB}{115,115,115}
\definecolor{Greys-9-7}{RGB}{82,82,82}
\definecolor{Greys-9-J}{RGB}{82,82,82}
\definecolor{Greys-9-8}{RGB}{37,37,37}
\definecolor{Greys-9-K}{RGB}{37,37,37}
\definecolor{Greys-9-9}{RGB}{0,0,0}
\definecolor{Greys-9-M}{RGB}{0,0,0}
\definecolor{PuOr-3-1}{RGB}{241,163,64}
\definecolor{PuOr-3-E}{RGB}{241,163,64}
\definecolor{PuOr-3-2}{RGB}{247,247,247}
\definecolor{PuOr-3-H}{RGB}{247,247,247}
\definecolor{PuOr-3-3}{RGB}{153,142,195}
\definecolor{PuOr-3-K}{RGB}{153,142,195}
\definecolor{PuOr-4-1}{RGB}{230,97,1}
\definecolor{PuOr-4-C}{RGB}{230,97,1}
\definecolor{PuOr-4-2}{RGB}{253,184,99}
\definecolor{PuOr-4-F}{RGB}{253,184,99}
\definecolor{PuOr-4-3}{RGB}{178,171,210}
\definecolor{PuOr-4-J}{RGB}{178,171,210}
\definecolor{PuOr-4-4}{RGB}{94,60,153}
\definecolor{PuOr-4-M}{RGB}{94,60,153}
\definecolor{PuOr-5-1}{RGB}{230,97,1}
\definecolor{PuOr-5-C}{RGB}{230,97,1}
\definecolor{PuOr-5-2}{RGB}{253,184,99}
\definecolor{PuOr-5-F}{RGB}{253,184,99}
\definecolor{PuOr-5-3}{RGB}{247,247,247}
\definecolor{PuOr-5-H}{RGB}{247,247,247}
\definecolor{PuOr-5-4}{RGB}{178,171,210}
\definecolor{PuOr-5-J}{RGB}{178,171,210}
\definecolor{PuOr-5-5}{RGB}{94,60,153}
\definecolor{PuOr-5-M}{RGB}{94,60,153}
\definecolor{PuOr-6-1}{RGB}{179,88,6}
\definecolor{PuOr-6-B}{RGB}{179,88,6}
\definecolor{PuOr-6-2}{RGB}{241,163,64}
\definecolor{PuOr-6-E}{RGB}{241,163,64}
\definecolor{PuOr-6-3}{RGB}{254,224,182}
\definecolor{PuOr-6-G}{RGB}{254,224,182}
\definecolor{PuOr-6-4}{RGB}{216,218,235}
\definecolor{PuOr-6-I}{RGB}{216,218,235}
\definecolor{PuOr-6-5}{RGB}{153,142,195}
\definecolor{PuOr-6-K}{RGB}{153,142,195}
\definecolor{PuOr-6-6}{RGB}{84,39,136}
\definecolor{PuOr-6-N}{RGB}{84,39,136}
\definecolor{PuOr-7-1}{RGB}{179,88,6}
\definecolor{PuOr-7-B}{RGB}{179,88,6}
\definecolor{PuOr-7-2}{RGB}{241,163,64}
\definecolor{PuOr-7-E}{RGB}{241,163,64}
\definecolor{PuOr-7-3}{RGB}{254,224,182}
\definecolor{PuOr-7-G}{RGB}{254,224,182}
\definecolor{PuOr-7-4}{RGB}{247,247,247}
\definecolor{PuOr-7-H}{RGB}{247,247,247}
\definecolor{PuOr-7-5}{RGB}{216,218,235}
\definecolor{PuOr-7-I}{RGB}{216,218,235}
\definecolor{PuOr-7-6}{RGB}{153,142,195}
\definecolor{PuOr-7-K}{RGB}{153,142,195}
\definecolor{PuOr-7-7}{RGB}{84,39,136}
\definecolor{PuOr-7-N}{RGB}{84,39,136}
\definecolor{PuOr-8-1}{RGB}{179,88,6}
\definecolor{PuOr-8-B}{RGB}{179,88,6}
\definecolor{PuOr-8-2}{RGB}{224,130,20}
\definecolor{PuOr-8-D}{RGB}{224,130,20}
\definecolor{PuOr-8-3}{RGB}{253,184,99}
\definecolor{PuOr-8-F}{RGB}{253,184,99}
\definecolor{PuOr-8-4}{RGB}{254,224,182}
\definecolor{PuOr-8-G}{RGB}{254,224,182}
\definecolor{PuOr-8-5}{RGB}{216,218,235}
\definecolor{PuOr-8-I}{RGB}{216,218,235}
\definecolor{PuOr-8-6}{RGB}{178,171,210}
\definecolor{PuOr-8-J}{RGB}{178,171,210}
\definecolor{PuOr-8-7}{RGB}{128,115,172}
\definecolor{PuOr-8-L}{RGB}{128,115,172}
\definecolor{PuOr-8-8}{RGB}{84,39,136}
\definecolor{PuOr-8-N}{RGB}{84,39,136}
\definecolor{PuOr-9-1}{RGB}{179,88,6}
\definecolor{PuOr-9-B}{RGB}{179,88,6}
\definecolor{PuOr-9-2}{RGB}{224,130,20}
\definecolor{PuOr-9-D}{RGB}{224,130,20}
\definecolor{PuOr-9-3}{RGB}{253,184,99}
\definecolor{PuOr-9-F}{RGB}{253,184,99}
\definecolor{PuOr-9-4}{RGB}{254,224,182}
\definecolor{PuOr-9-G}{RGB}{254,224,182}
\definecolor{PuOr-9-5}{RGB}{247,247,247}
\definecolor{PuOr-9-H}{RGB}{247,247,247}
\definecolor{PuOr-9-6}{RGB}{216,218,235}
\definecolor{PuOr-9-I}{RGB}{216,218,235}
\definecolor{PuOr-9-7}{RGB}{178,171,210}
\definecolor{PuOr-9-J}{RGB}{178,171,210}
\definecolor{PuOr-9-8}{RGB}{128,115,172}
\definecolor{PuOr-9-L}{RGB}{128,115,172}
\definecolor{PuOr-9-9}{RGB}{84,39,136}
\definecolor{PuOr-9-N}{RGB}{84,39,136}
\definecolor{PuOr-10-1}{RGB}{127,59,8}
\definecolor{PuOr-10-A}{RGB}{127,59,8}
\definecolor{PuOr-10-2}{RGB}{179,88,6}
\definecolor{PuOr-10-B}{RGB}{179,88,6}
\definecolor{PuOr-10-3}{RGB}{224,130,20}
\definecolor{PuOr-10-D}{RGB}{224,130,20}
\definecolor{PuOr-10-4}{RGB}{253,184,99}
\definecolor{PuOr-10-F}{RGB}{253,184,99}
\definecolor{PuOr-10-5}{RGB}{254,224,182}
\definecolor{PuOr-10-G}{RGB}{254,224,182}
\definecolor{PuOr-10-6}{RGB}{216,218,235}
\definecolor{PuOr-10-I}{RGB}{216,218,235}
\definecolor{PuOr-10-7}{RGB}{178,171,210}
\definecolor{PuOr-10-J}{RGB}{178,171,210}
\definecolor{PuOr-10-8}{RGB}{128,115,172}
\definecolor{PuOr-10-L}{RGB}{128,115,172}
\definecolor{PuOr-10-9}{RGB}{84,39,136}
\definecolor{PuOr-10-N}{RGB}{84,39,136}
\definecolor{PuOr-10-10}{RGB}{45,0,75}
\definecolor{PuOr-10-O}{RGB}{45,0,75}
\definecolor{PuOr-11-1}{RGB}{127,59,8}
\definecolor{PuOr-11-A}{RGB}{127,59,8}
\definecolor{PuOr-11-2}{RGB}{179,88,6}
\definecolor{PuOr-11-B}{RGB}{179,88,6}
\definecolor{PuOr-11-3}{RGB}{224,130,20}
\definecolor{PuOr-11-D}{RGB}{224,130,20}
\definecolor{PuOr-11-4}{RGB}{253,184,99}
\definecolor{PuOr-11-F}{RGB}{253,184,99}
\definecolor{PuOr-11-5}{RGB}{254,224,182}
\definecolor{PuOr-11-G}{RGB}{254,224,182}
\definecolor{PuOr-11-6}{RGB}{247,247,247}
\definecolor{PuOr-11-H}{RGB}{247,247,247}
\definecolor{PuOr-11-7}{RGB}{216,218,235}
\definecolor{PuOr-11-I}{RGB}{216,218,235}
\definecolor{PuOr-11-8}{RGB}{178,171,210}
\definecolor{PuOr-11-J}{RGB}{178,171,210}
\definecolor{PuOr-11-9}{RGB}{128,115,172}
\definecolor{PuOr-11-L}{RGB}{128,115,172}
\definecolor{PuOr-11-10}{RGB}{84,39,136}
\definecolor{PuOr-11-N}{RGB}{84,39,136}
\definecolor{PuOr-11-11}{RGB}{45,0,75}
\definecolor{PuOr-11-O}{RGB}{45,0,75}
\definecolor{BrBG-3-1}{RGB}{216,179,101}
\definecolor{BrBG-3-E}{RGB}{216,179,101}
\definecolor{BrBG-3-2}{RGB}{245,245,245}
\definecolor{BrBG-3-H}{RGB}{245,245,245}
\definecolor{BrBG-3-3}{RGB}{90,180,172}
\definecolor{BrBG-3-K}{RGB}{90,180,172}
\definecolor{BrBG-4-1}{RGB}{166,97,26}
\definecolor{BrBG-4-C}{RGB}{166,97,26}
\definecolor{BrBG-4-2}{RGB}{223,194,125}
\definecolor{BrBG-4-F}{RGB}{223,194,125}
\definecolor{BrBG-4-3}{RGB}{128,205,193}
\definecolor{BrBG-4-J}{RGB}{128,205,193}
\definecolor{BrBG-4-4}{RGB}{1,133,113}
\definecolor{BrBG-4-M}{RGB}{1,133,113}
\definecolor{BrBG-5-1}{RGB}{166,97,26}
\definecolor{BrBG-5-C}{RGB}{166,97,26}
\definecolor{BrBG-5-2}{RGB}{223,194,125}
\definecolor{BrBG-5-F}{RGB}{223,194,125}
\definecolor{BrBG-5-3}{RGB}{245,245,245}
\definecolor{BrBG-5-H}{RGB}{245,245,245}
\definecolor{BrBG-5-4}{RGB}{128,205,193}
\definecolor{BrBG-5-J}{RGB}{128,205,193}
\definecolor{BrBG-5-5}{RGB}{1,133,113}
\definecolor{BrBG-5-M}{RGB}{1,133,113}
\definecolor{BrBG-6-1}{RGB}{140,81,10}
\definecolor{BrBG-6-B}{RGB}{140,81,10}
\definecolor{BrBG-6-2}{RGB}{216,179,101}
\definecolor{BrBG-6-E}{RGB}{216,179,101}
\definecolor{BrBG-6-3}{RGB}{246,232,195}
\definecolor{BrBG-6-G}{RGB}{246,232,195}
\definecolor{BrBG-6-4}{RGB}{199,234,229}
\definecolor{BrBG-6-I}{RGB}{199,234,229}
\definecolor{BrBG-6-5}{RGB}{90,180,172}
\definecolor{BrBG-6-K}{RGB}{90,180,172}
\definecolor{BrBG-6-6}{RGB}{1,102,94}
\definecolor{BrBG-6-N}{RGB}{1,102,94}
\definecolor{BrBG-7-1}{RGB}{140,81,10}
\definecolor{BrBG-7-B}{RGB}{140,81,10}
\definecolor{BrBG-7-2}{RGB}{216,179,101}
\definecolor{BrBG-7-E}{RGB}{216,179,101}
\definecolor{BrBG-7-3}{RGB}{246,232,195}
\definecolor{BrBG-7-G}{RGB}{246,232,195}
\definecolor{BrBG-7-4}{RGB}{245,245,245}
\definecolor{BrBG-7-H}{RGB}{245,245,245}
\definecolor{BrBG-7-5}{RGB}{199,234,229}
\definecolor{BrBG-7-I}{RGB}{199,234,229}
\definecolor{BrBG-7-6}{RGB}{90,180,172}
\definecolor{BrBG-7-K}{RGB}{90,180,172}
\definecolor{BrBG-7-7}{RGB}{1,102,94}
\definecolor{BrBG-7-N}{RGB}{1,102,94}
\definecolor{BrBG-8-1}{RGB}{140,81,10}
\definecolor{BrBG-8-B}{RGB}{140,81,10}
\definecolor{BrBG-8-2}{RGB}{191,129,45}
\definecolor{BrBG-8-D}{RGB}{191,129,45}
\definecolor{BrBG-8-3}{RGB}{223,194,125}
\definecolor{BrBG-8-F}{RGB}{223,194,125}
\definecolor{BrBG-8-4}{RGB}{246,232,195}
\definecolor{BrBG-8-G}{RGB}{246,232,195}
\definecolor{BrBG-8-5}{RGB}{199,234,229}
\definecolor{BrBG-8-I}{RGB}{199,234,229}
\definecolor{BrBG-8-6}{RGB}{128,205,193}
\definecolor{BrBG-8-J}{RGB}{128,205,193}
\definecolor{BrBG-8-7}{RGB}{53,151,143}
\definecolor{BrBG-8-L}{RGB}{53,151,143}
\definecolor{BrBG-8-8}{RGB}{1,102,94}
\definecolor{BrBG-8-N}{RGB}{1,102,94}
\definecolor{BrBG-9-1}{RGB}{140,81,10}
\definecolor{BrBG-9-B}{RGB}{140,81,10}
\definecolor{BrBG-9-2}{RGB}{191,129,45}
\definecolor{BrBG-9-D}{RGB}{191,129,45}
\definecolor{BrBG-9-3}{RGB}{223,194,125}
\definecolor{BrBG-9-F}{RGB}{223,194,125}
\definecolor{BrBG-9-4}{RGB}{246,232,195}
\definecolor{BrBG-9-G}{RGB}{246,232,195}
\definecolor{BrBG-9-5}{RGB}{245,245,245}
\definecolor{BrBG-9-H}{RGB}{245,245,245}
\definecolor{BrBG-9-6}{RGB}{199,234,229}
\definecolor{BrBG-9-I}{RGB}{199,234,229}
\definecolor{BrBG-9-7}{RGB}{128,205,193}
\definecolor{BrBG-9-J}{RGB}{128,205,193}
\definecolor{BrBG-9-8}{RGB}{53,151,143}
\definecolor{BrBG-9-L}{RGB}{53,151,143}
\definecolor{BrBG-9-9}{RGB}{1,102,94}
\definecolor{BrBG-9-N}{RGB}{1,102,94}
\definecolor{BrBG-10-1}{RGB}{84,48,5}
\definecolor{BrBG-10-A}{RGB}{84,48,5}
\definecolor{BrBG-10-2}{RGB}{140,81,10}
\definecolor{BrBG-10-B}{RGB}{140,81,10}
\definecolor{BrBG-10-3}{RGB}{191,129,45}
\definecolor{BrBG-10-D}{RGB}{191,129,45}
\definecolor{BrBG-10-4}{RGB}{223,194,125}
\definecolor{BrBG-10-F}{RGB}{223,194,125}
\definecolor{BrBG-10-5}{RGB}{246,232,195}
\definecolor{BrBG-10-G}{RGB}{246,232,195}
\definecolor{BrBG-10-6}{RGB}{199,234,229}
\definecolor{BrBG-10-I}{RGB}{199,234,229}
\definecolor{BrBG-10-7}{RGB}{128,205,193}
\definecolor{BrBG-10-J}{RGB}{128,205,193}
\definecolor{BrBG-10-8}{RGB}{53,151,143}
\definecolor{BrBG-10-L}{RGB}{53,151,143}
\definecolor{BrBG-10-9}{RGB}{1,102,94}
\definecolor{BrBG-10-N}{RGB}{1,102,94}
\definecolor{BrBG-10-10}{RGB}{0,60,48}
\definecolor{BrBG-10-O}{RGB}{0,60,48}
\definecolor{BrBG-11-1}{RGB}{84,48,5}
\definecolor{BrBG-11-A}{RGB}{84,48,5}
\definecolor{BrBG-11-2}{RGB}{140,81,10}
\definecolor{BrBG-11-B}{RGB}{140,81,10}
\definecolor{BrBG-11-3}{RGB}{191,129,45}
\definecolor{BrBG-11-D}{RGB}{191,129,45}
\definecolor{BrBG-11-4}{RGB}{223,194,125}
\definecolor{BrBG-11-F}{RGB}{223,194,125}
\definecolor{BrBG-11-5}{RGB}{246,232,195}
\definecolor{BrBG-11-G}{RGB}{246,232,195}
\definecolor{BrBG-11-6}{RGB}{245,245,245}
\definecolor{BrBG-11-H}{RGB}{245,245,245}
\definecolor{BrBG-11-7}{RGB}{199,234,229}
\definecolor{BrBG-11-I}{RGB}{199,234,229}
\definecolor{BrBG-11-8}{RGB}{128,205,193}
\definecolor{BrBG-11-J}{RGB}{128,205,193}
\definecolor{BrBG-11-9}{RGB}{53,151,143}
\definecolor{BrBG-11-L}{RGB}{53,151,143}
\definecolor{BrBG-11-10}{RGB}{1,102,94}
\definecolor{BrBG-11-N}{RGB}{1,102,94}
\definecolor{BrBG-11-11}{RGB}{0,60,48}
\definecolor{BrBG-11-O}{RGB}{0,60,48}
\definecolor{PRGn-3-1}{RGB}{175,141,195}
\definecolor{PRGn-3-E}{RGB}{175,141,195}
\definecolor{PRGn-3-2}{RGB}{247,247,247}
\definecolor{PRGn-3-H}{RGB}{247,247,247}
\definecolor{PRGn-3-3}{RGB}{127,191,123}
\definecolor{PRGn-3-K}{RGB}{127,191,123}
\definecolor{PRGn-4-1}{RGB}{123,50,148}
\definecolor{PRGn-4-C}{RGB}{123,50,148}
\definecolor{PRGn-4-2}{RGB}{194,165,207}
\definecolor{PRGn-4-F}{RGB}{194,165,207}
\definecolor{PRGn-4-3}{RGB}{166,219,160}
\definecolor{PRGn-4-J}{RGB}{166,219,160}
\definecolor{PRGn-4-4}{RGB}{0,136,55}
\definecolor{PRGn-4-M}{RGB}{0,136,55}
\definecolor{PRGn-5-1}{RGB}{123,50,148}
\definecolor{PRGn-5-C}{RGB}{123,50,148}
\definecolor{PRGn-5-2}{RGB}{194,165,207}
\definecolor{PRGn-5-F}{RGB}{194,165,207}
\definecolor{PRGn-5-3}{RGB}{247,247,247}
\definecolor{PRGn-5-H}{RGB}{247,247,247}
\definecolor{PRGn-5-4}{RGB}{166,219,160}
\definecolor{PRGn-5-J}{RGB}{166,219,160}
\definecolor{PRGn-5-5}{RGB}{0,136,55}
\definecolor{PRGn-5-M}{RGB}{0,136,55}
\definecolor{PRGn-6-1}{RGB}{118,42,131}
\definecolor{PRGn-6-B}{RGB}{118,42,131}
\definecolor{PRGn-6-2}{RGB}{175,141,195}
\definecolor{PRGn-6-E}{RGB}{175,141,195}
\definecolor{PRGn-6-3}{RGB}{231,212,232}
\definecolor{PRGn-6-G}{RGB}{231,212,232}
\definecolor{PRGn-6-4}{RGB}{217,240,211}
\definecolor{PRGn-6-I}{RGB}{217,240,211}
\definecolor{PRGn-6-5}{RGB}{127,191,123}
\definecolor{PRGn-6-K}{RGB}{127,191,123}
\definecolor{PRGn-6-6}{RGB}{27,120,55}
\definecolor{PRGn-6-N}{RGB}{27,120,55}
\definecolor{PRGn-7-1}{RGB}{118,42,131}
\definecolor{PRGn-7-B}{RGB}{118,42,131}
\definecolor{PRGn-7-2}{RGB}{175,141,195}
\definecolor{PRGn-7-E}{RGB}{175,141,195}
\definecolor{PRGn-7-3}{RGB}{231,212,232}
\definecolor{PRGn-7-G}{RGB}{231,212,232}
\definecolor{PRGn-7-4}{RGB}{247,247,247}
\definecolor{PRGn-7-H}{RGB}{247,247,247}
\definecolor{PRGn-7-5}{RGB}{217,240,211}
\definecolor{PRGn-7-I}{RGB}{217,240,211}
\definecolor{PRGn-7-6}{RGB}{127,191,123}
\definecolor{PRGn-7-K}{RGB}{127,191,123}
\definecolor{PRGn-7-7}{RGB}{27,120,55}
\definecolor{PRGn-7-N}{RGB}{27,120,55}
\definecolor{PRGn-8-1}{RGB}{118,42,131}
\definecolor{PRGn-8-B}{RGB}{118,42,131}
\definecolor{PRGn-8-2}{RGB}{153,112,171}
\definecolor{PRGn-8-D}{RGB}{153,112,171}
\definecolor{PRGn-8-3}{RGB}{194,165,207}
\definecolor{PRGn-8-F}{RGB}{194,165,207}
\definecolor{PRGn-8-4}{RGB}{231,212,232}
\definecolor{PRGn-8-G}{RGB}{231,212,232}
\definecolor{PRGn-8-5}{RGB}{217,240,211}
\definecolor{PRGn-8-I}{RGB}{217,240,211}
\definecolor{PRGn-8-6}{RGB}{166,219,160}
\definecolor{PRGn-8-J}{RGB}{166,219,160}
\definecolor{PRGn-8-7}{RGB}{90,174,97}
\definecolor{PRGn-8-L}{RGB}{90,174,97}
\definecolor{PRGn-8-8}{RGB}{27,120,55}
\definecolor{PRGn-8-N}{RGB}{27,120,55}
\definecolor{PRGn-9-1}{RGB}{118,42,131}
\definecolor{PRGn-9-B}{RGB}{118,42,131}
\definecolor{PRGn-9-2}{RGB}{153,112,171}
\definecolor{PRGn-9-D}{RGB}{153,112,171}
\definecolor{PRGn-9-3}{RGB}{194,165,207}
\definecolor{PRGn-9-F}{RGB}{194,165,207}
\definecolor{PRGn-9-4}{RGB}{231,212,232}
\definecolor{PRGn-9-G}{RGB}{231,212,232}
\definecolor{PRGn-9-5}{RGB}{247,247,247}
\definecolor{PRGn-9-H}{RGB}{247,247,247}
\definecolor{PRGn-9-6}{RGB}{217,240,211}
\definecolor{PRGn-9-I}{RGB}{217,240,211}
\definecolor{PRGn-9-7}{RGB}{166,219,160}
\definecolor{PRGn-9-J}{RGB}{166,219,160}
\definecolor{PRGn-9-8}{RGB}{90,174,97}
\definecolor{PRGn-9-L}{RGB}{90,174,97}
\definecolor{PRGn-9-9}{RGB}{27,120,55}
\definecolor{PRGn-9-N}{RGB}{27,120,55}
\definecolor{PRGn-10-1}{RGB}{64,0,75}
\definecolor{PRGn-10-A}{RGB}{64,0,75}
\definecolor{PRGn-10-2}{RGB}{118,42,131}
\definecolor{PRGn-10-B}{RGB}{118,42,131}
\definecolor{PRGn-10-3}{RGB}{153,112,171}
\definecolor{PRGn-10-D}{RGB}{153,112,171}
\definecolor{PRGn-10-4}{RGB}{194,165,207}
\definecolor{PRGn-10-F}{RGB}{194,165,207}
\definecolor{PRGn-10-5}{RGB}{231,212,232}
\definecolor{PRGn-10-G}{RGB}{231,212,232}
\definecolor{PRGn-10-6}{RGB}{217,240,211}
\definecolor{PRGn-10-I}{RGB}{217,240,211}
\definecolor{PRGn-10-7}{RGB}{166,219,160}
\definecolor{PRGn-10-J}{RGB}{166,219,160}
\definecolor{PRGn-10-8}{RGB}{90,174,97}
\definecolor{PRGn-10-L}{RGB}{90,174,97}
\definecolor{PRGn-10-9}{RGB}{27,120,55}
\definecolor{PRGn-10-N}{RGB}{27,120,55}
\definecolor{PRGn-10-10}{RGB}{0,68,27}
\definecolor{PRGn-10-O}{RGB}{0,68,27}
\definecolor{PRGn-11-1}{RGB}{64,0,75}
\definecolor{PRGn-11-A}{RGB}{64,0,75}
\definecolor{PRGn-11-2}{RGB}{118,42,131}
\definecolor{PRGn-11-B}{RGB}{118,42,131}
\definecolor{PRGn-11-3}{RGB}{153,112,171}
\definecolor{PRGn-11-D}{RGB}{153,112,171}
\definecolor{PRGn-11-4}{RGB}{194,165,207}
\definecolor{PRGn-11-F}{RGB}{194,165,207}
\definecolor{PRGn-11-5}{RGB}{231,212,232}
\definecolor{PRGn-11-G}{RGB}{231,212,232}
\definecolor{PRGn-11-6}{RGB}{247,247,247}
\definecolor{PRGn-11-H}{RGB}{247,247,247}
\definecolor{PRGn-11-7}{RGB}{217,240,211}
\definecolor{PRGn-11-I}{RGB}{217,240,211}
\definecolor{PRGn-11-8}{RGB}{166,219,160}
\definecolor{PRGn-11-J}{RGB}{166,219,160}
\definecolor{PRGn-11-9}{RGB}{90,174,97}
\definecolor{PRGn-11-L}{RGB}{90,174,97}
\definecolor{PRGn-11-10}{RGB}{27,120,55}
\definecolor{PRGn-11-N}{RGB}{27,120,55}
\definecolor{PRGn-11-11}{RGB}{0,68,27}
\definecolor{PRGn-11-O}{RGB}{0,68,27}
\definecolor{PiYG-3-1}{RGB}{233,163,201}
\definecolor{PiYG-3-E}{RGB}{233,163,201}
\definecolor{PiYG-3-2}{RGB}{247,247,247}
\definecolor{PiYG-3-H}{RGB}{247,247,247}
\definecolor{PiYG-3-3}{RGB}{161,215,106}
\definecolor{PiYG-3-K}{RGB}{161,215,106}
\definecolor{PiYG-4-1}{RGB}{208,28,139}
\definecolor{PiYG-4-C}{RGB}{208,28,139}
\definecolor{PiYG-4-2}{RGB}{241,182,218}
\definecolor{PiYG-4-F}{RGB}{241,182,218}
\definecolor{PiYG-4-3}{RGB}{184,225,134}
\definecolor{PiYG-4-J}{RGB}{184,225,134}
\definecolor{PiYG-4-4}{RGB}{77,172,38}
\definecolor{PiYG-4-M}{RGB}{77,172,38}
\definecolor{PiYG-5-1}{RGB}{208,28,139}
\definecolor{PiYG-5-C}{RGB}{208,28,139}
\definecolor{PiYG-5-2}{RGB}{241,182,218}
\definecolor{PiYG-5-F}{RGB}{241,182,218}
\definecolor{PiYG-5-3}{RGB}{247,247,247}
\definecolor{PiYG-5-H}{RGB}{247,247,247}
\definecolor{PiYG-5-4}{RGB}{184,225,134}
\definecolor{PiYG-5-J}{RGB}{184,225,134}
\definecolor{PiYG-5-5}{RGB}{77,172,38}
\definecolor{PiYG-5-M}{RGB}{77,172,38}
\definecolor{PiYG-6-1}{RGB}{197,27,125}
\definecolor{PiYG-6-B}{RGB}{197,27,125}
\definecolor{PiYG-6-2}{RGB}{233,163,201}
\definecolor{PiYG-6-E}{RGB}{233,163,201}
\definecolor{PiYG-6-3}{RGB}{253,224,239}
\definecolor{PiYG-6-G}{RGB}{253,224,239}
\definecolor{PiYG-6-4}{RGB}{230,245,208}
\definecolor{PiYG-6-I}{RGB}{230,245,208}
\definecolor{PiYG-6-5}{RGB}{161,215,106}
\definecolor{PiYG-6-K}{RGB}{161,215,106}
\definecolor{PiYG-6-6}{RGB}{77,146,33}
\definecolor{PiYG-6-N}{RGB}{77,146,33}
\definecolor{PiYG-7-1}{RGB}{197,27,125}
\definecolor{PiYG-7-B}{RGB}{197,27,125}
\definecolor{PiYG-7-2}{RGB}{233,163,201}
\definecolor{PiYG-7-E}{RGB}{233,163,201}
\definecolor{PiYG-7-3}{RGB}{253,224,239}
\definecolor{PiYG-7-G}{RGB}{253,224,239}
\definecolor{PiYG-7-4}{RGB}{247,247,247}
\definecolor{PiYG-7-H}{RGB}{247,247,247}
\definecolor{PiYG-7-5}{RGB}{230,245,208}
\definecolor{PiYG-7-I}{RGB}{230,245,208}
\definecolor{PiYG-7-6}{RGB}{161,215,106}
\definecolor{PiYG-7-K}{RGB}{161,215,106}
\definecolor{PiYG-7-7}{RGB}{77,146,33}
\definecolor{PiYG-7-N}{RGB}{77,146,33}
\definecolor{PiYG-8-1}{RGB}{197,27,125}
\definecolor{PiYG-8-B}{RGB}{197,27,125}
\definecolor{PiYG-8-2}{RGB}{222,119,174}
\definecolor{PiYG-8-D}{RGB}{222,119,174}
\definecolor{PiYG-8-3}{RGB}{241,182,218}
\definecolor{PiYG-8-F}{RGB}{241,182,218}
\definecolor{PiYG-8-4}{RGB}{253,224,239}
\definecolor{PiYG-8-G}{RGB}{253,224,239}
\definecolor{PiYG-8-5}{RGB}{230,245,208}
\definecolor{PiYG-8-I}{RGB}{230,245,208}
\definecolor{PiYG-8-6}{RGB}{184,225,134}
\definecolor{PiYG-8-J}{RGB}{184,225,134}
\definecolor{PiYG-8-7}{RGB}{127,188,65}
\definecolor{PiYG-8-L}{RGB}{127,188,65}
\definecolor{PiYG-8-8}{RGB}{77,146,33}
\definecolor{PiYG-8-N}{RGB}{77,146,33}
\definecolor{PiYG-9-1}{RGB}{197,27,125}
\definecolor{PiYG-9-B}{RGB}{197,27,125}
\definecolor{PiYG-9-2}{RGB}{222,119,174}
\definecolor{PiYG-9-D}{RGB}{222,119,174}
\definecolor{PiYG-9-3}{RGB}{241,182,218}
\definecolor{PiYG-9-F}{RGB}{241,182,218}
\definecolor{PiYG-9-4}{RGB}{253,224,239}
\definecolor{PiYG-9-G}{RGB}{253,224,239}
\definecolor{PiYG-9-5}{RGB}{247,247,247}
\definecolor{PiYG-9-H}{RGB}{247,247,247}
\definecolor{PiYG-9-6}{RGB}{230,245,208}
\definecolor{PiYG-9-I}{RGB}{230,245,208}
\definecolor{PiYG-9-7}{RGB}{184,225,134}
\definecolor{PiYG-9-J}{RGB}{184,225,134}
\definecolor{PiYG-9-8}{RGB}{127,188,65}
\definecolor{PiYG-9-L}{RGB}{127,188,65}
\definecolor{PiYG-9-9}{RGB}{77,146,33}
\definecolor{PiYG-9-N}{RGB}{77,146,33}
\definecolor{PiYG-10-1}{RGB}{142,1,82}
\definecolor{PiYG-10-A}{RGB}{142,1,82}
\definecolor{PiYG-10-2}{RGB}{197,27,125}
\definecolor{PiYG-10-B}{RGB}{197,27,125}
\definecolor{PiYG-10-3}{RGB}{222,119,174}
\definecolor{PiYG-10-D}{RGB}{222,119,174}
\definecolor{PiYG-10-4}{RGB}{241,182,218}
\definecolor{PiYG-10-F}{RGB}{241,182,218}
\definecolor{PiYG-10-5}{RGB}{253,224,239}
\definecolor{PiYG-10-G}{RGB}{253,224,239}
\definecolor{PiYG-10-6}{RGB}{230,245,208}
\definecolor{PiYG-10-I}{RGB}{230,245,208}
\definecolor{PiYG-10-7}{RGB}{184,225,134}
\definecolor{PiYG-10-J}{RGB}{184,225,134}
\definecolor{PiYG-10-8}{RGB}{127,188,65}
\definecolor{PiYG-10-L}{RGB}{127,188,65}
\definecolor{PiYG-10-9}{RGB}{77,146,33}
\definecolor{PiYG-10-N}{RGB}{77,146,33}
\definecolor{PiYG-10-10}{RGB}{39,100,25}
\definecolor{PiYG-10-O}{RGB}{39,100,25}
\definecolor{PiYG-11-1}{RGB}{142,1,82}
\definecolor{PiYG-11-A}{RGB}{142,1,82}
\definecolor{PiYG-11-2}{RGB}{197,27,125}
\definecolor{PiYG-11-B}{RGB}{197,27,125}
\definecolor{PiYG-11-3}{RGB}{222,119,174}
\definecolor{PiYG-11-D}{RGB}{222,119,174}
\definecolor{PiYG-11-4}{RGB}{241,182,218}
\definecolor{PiYG-11-F}{RGB}{241,182,218}
\definecolor{PiYG-11-5}{RGB}{253,224,239}
\definecolor{PiYG-11-G}{RGB}{253,224,239}
\definecolor{PiYG-11-6}{RGB}{247,247,247}
\definecolor{PiYG-11-H}{RGB}{247,247,247}
\definecolor{PiYG-11-7}{RGB}{230,245,208}
\definecolor{PiYG-11-I}{RGB}{230,245,208}
\definecolor{PiYG-11-8}{RGB}{184,225,134}
\definecolor{PiYG-11-J}{RGB}{184,225,134}
\definecolor{PiYG-11-9}{RGB}{127,188,65}
\definecolor{PiYG-11-L}{RGB}{127,188,65}
\definecolor{PiYG-11-10}{RGB}{77,146,33}
\definecolor{PiYG-11-N}{RGB}{77,146,33}
\definecolor{PiYG-11-11}{RGB}{39,100,25}
\definecolor{PiYG-11-O}{RGB}{39,100,25}
\definecolor{RdBu-3-1}{RGB}{239,138,98}
\definecolor{RdBu-3-E}{RGB}{239,138,98}
\definecolor{RdBu-3-2}{RGB}{247,247,247}
\definecolor{RdBu-3-H}{RGB}{247,247,247}
\definecolor{RdBu-3-3}{RGB}{103,169,207}
\definecolor{RdBu-3-K}{RGB}{103,169,207}
\definecolor{RdBu-4-1}{RGB}{202,0,32}
\definecolor{RdBu-4-C}{RGB}{202,0,32}
\definecolor{RdBu-4-2}{RGB}{244,165,130}
\definecolor{RdBu-4-F}{RGB}{244,165,130}
\definecolor{RdBu-4-3}{RGB}{146,197,222}
\definecolor{RdBu-4-J}{RGB}{146,197,222}
\definecolor{RdBu-4-4}{RGB}{5,113,176}
\definecolor{RdBu-4-M}{RGB}{5,113,176}
\definecolor{RdBu-5-1}{RGB}{202,0,32}
\definecolor{RdBu-5-C}{RGB}{202,0,32}
\definecolor{RdBu-5-2}{RGB}{244,165,130}
\definecolor{RdBu-5-F}{RGB}{244,165,130}
\definecolor{RdBu-5-3}{RGB}{247,247,247}
\definecolor{RdBu-5-H}{RGB}{247,247,247}
\definecolor{RdBu-5-4}{RGB}{146,197,222}
\definecolor{RdBu-5-J}{RGB}{146,197,222}
\definecolor{RdBu-5-5}{RGB}{5,113,176}
\definecolor{RdBu-5-M}{RGB}{5,113,176}
\definecolor{RdBu-6-1}{RGB}{178,24,43}
\definecolor{RdBu-6-B}{RGB}{178,24,43}
\definecolor{RdBu-6-2}{RGB}{239,138,98}
\definecolor{RdBu-6-E}{RGB}{239,138,98}
\definecolor{RdBu-6-3}{RGB}{253,219,199}
\definecolor{RdBu-6-G}{RGB}{253,219,199}
\definecolor{RdBu-6-4}{RGB}{209,229,240}
\definecolor{RdBu-6-I}{RGB}{209,229,240}
\definecolor{RdBu-6-5}{RGB}{103,169,207}
\definecolor{RdBu-6-K}{RGB}{103,169,207}
\definecolor{RdBu-6-6}{RGB}{33,102,172}
\definecolor{RdBu-6-N}{RGB}{33,102,172}
\definecolor{RdBu-7-1}{RGB}{178,24,43}
\definecolor{RdBu-7-B}{RGB}{178,24,43}
\definecolor{RdBu-7-2}{RGB}{239,138,98}
\definecolor{RdBu-7-E}{RGB}{239,138,98}
\definecolor{RdBu-7-3}{RGB}{253,219,199}
\definecolor{RdBu-7-G}{RGB}{253,219,199}
\definecolor{RdBu-7-4}{RGB}{247,247,247}
\definecolor{RdBu-7-H}{RGB}{247,247,247}
\definecolor{RdBu-7-5}{RGB}{209,229,240}
\definecolor{RdBu-7-I}{RGB}{209,229,240}
\definecolor{RdBu-7-6}{RGB}{103,169,207}
\definecolor{RdBu-7-K}{RGB}{103,169,207}
\definecolor{RdBu-7-7}{RGB}{33,102,172}
\definecolor{RdBu-7-N}{RGB}{33,102,172}
\definecolor{RdBu-8-1}{RGB}{178,24,43}
\definecolor{RdBu-8-B}{RGB}{178,24,43}
\definecolor{RdBu-8-2}{RGB}{214,96,77}
\definecolor{RdBu-8-D}{RGB}{214,96,77}
\definecolor{RdBu-8-3}{RGB}{244,165,130}
\definecolor{RdBu-8-F}{RGB}{244,165,130}
\definecolor{RdBu-8-4}{RGB}{253,219,199}
\definecolor{RdBu-8-G}{RGB}{253,219,199}
\definecolor{RdBu-8-5}{RGB}{209,229,240}
\definecolor{RdBu-8-I}{RGB}{209,229,240}
\definecolor{RdBu-8-6}{RGB}{146,197,222}
\definecolor{RdBu-8-J}{RGB}{146,197,222}
\definecolor{RdBu-8-7}{RGB}{67,147,195}
\definecolor{RdBu-8-L}{RGB}{67,147,195}
\definecolor{RdBu-8-8}{RGB}{33,102,172}
\definecolor{RdBu-8-N}{RGB}{33,102,172}
\definecolor{RdBu-9-1}{RGB}{178,24,43}
\definecolor{RdBu-9-B}{RGB}{178,24,43}
\definecolor{RdBu-9-2}{RGB}{214,96,77}
\definecolor{RdBu-9-D}{RGB}{214,96,77}
\definecolor{RdBu-9-3}{RGB}{244,165,130}
\definecolor{RdBu-9-F}{RGB}{244,165,130}
\definecolor{RdBu-9-4}{RGB}{253,219,199}
\definecolor{RdBu-9-G}{RGB}{253,219,199}
\definecolor{RdBu-9-5}{RGB}{247,247,247}
\definecolor{RdBu-9-H}{RGB}{247,247,247}
\definecolor{RdBu-9-6}{RGB}{209,229,240}
\definecolor{RdBu-9-I}{RGB}{209,229,240}
\definecolor{RdBu-9-7}{RGB}{146,197,222}
\definecolor{RdBu-9-J}{RGB}{146,197,222}
\definecolor{RdBu-9-8}{RGB}{67,147,195}
\definecolor{RdBu-9-L}{RGB}{67,147,195}
\definecolor{RdBu-9-9}{RGB}{33,102,172}
\definecolor{RdBu-9-N}{RGB}{33,102,172}
\definecolor{RdBu-10-1}{RGB}{103,0,31}
\definecolor{RdBu-10-A}{RGB}{103,0,31}
\definecolor{RdBu-10-2}{RGB}{178,24,43}
\definecolor{RdBu-10-B}{RGB}{178,24,43}
\definecolor{RdBu-10-3}{RGB}{214,96,77}
\definecolor{RdBu-10-D}{RGB}{214,96,77}
\definecolor{RdBu-10-4}{RGB}{244,165,130}
\definecolor{RdBu-10-F}{RGB}{244,165,130}
\definecolor{RdBu-10-5}{RGB}{253,219,199}
\definecolor{RdBu-10-G}{RGB}{253,219,199}
\definecolor{RdBu-10-6}{RGB}{209,229,240}
\definecolor{RdBu-10-I}{RGB}{209,229,240}
\definecolor{RdBu-10-7}{RGB}{146,197,222}
\definecolor{RdBu-10-J}{RGB}{146,197,222}
\definecolor{RdBu-10-8}{RGB}{67,147,195}
\definecolor{RdBu-10-L}{RGB}{67,147,195}
\definecolor{RdBu-10-9}{RGB}{33,102,172}
\definecolor{RdBu-10-N}{RGB}{33,102,172}
\definecolor{RdBu-10-10}{RGB}{5,48,97}
\definecolor{RdBu-10-O}{RGB}{5,48,97}
\definecolor{RdBu-11-1}{RGB}{103,0,31}
\definecolor{RdBu-11-A}{RGB}{103,0,31}
\definecolor{RdBu-11-2}{RGB}{178,24,43}
\definecolor{RdBu-11-B}{RGB}{178,24,43}
\definecolor{RdBu-11-3}{RGB}{214,96,77}
\definecolor{RdBu-11-D}{RGB}{214,96,77}
\definecolor{RdBu-11-4}{RGB}{244,165,130}
\definecolor{RdBu-11-F}{RGB}{244,165,130}
\definecolor{RdBu-11-5}{RGB}{253,219,199}
\definecolor{RdBu-11-G}{RGB}{253,219,199}
\definecolor{RdBu-11-6}{RGB}{247,247,247}
\definecolor{RdBu-11-H}{RGB}{247,247,247}
\definecolor{RdBu-11-7}{RGB}{209,229,240}
\definecolor{RdBu-11-I}{RGB}{209,229,240}
\definecolor{RdBu-11-8}{RGB}{146,197,222}
\definecolor{RdBu-11-J}{RGB}{146,197,222}
\definecolor{RdBu-11-9}{RGB}{67,147,195}
\definecolor{RdBu-11-L}{RGB}{67,147,195}
\definecolor{RdBu-11-10}{RGB}{33,102,172}
\definecolor{RdBu-11-N}{RGB}{33,102,172}
\definecolor{RdBu-11-11}{RGB}{5,48,97}
\definecolor{RdBu-11-O}{RGB}{5,48,97}
\definecolor{RdGy-3-1}{RGB}{239,138,98}
\definecolor{RdGy-3-E}{RGB}{239,138,98}
\definecolor{RdGy-3-2}{RGB}{255,255,255}
\definecolor{RdGy-3-H}{RGB}{255,255,255}
\definecolor{RdGy-3-3}{RGB}{153,153,153}
\definecolor{RdGy-3-K}{RGB}{153,153,153}
\definecolor{RdGy-4-1}{RGB}{202,0,32}
\definecolor{RdGy-4-C}{RGB}{202,0,32}
\definecolor{RdGy-4-2}{RGB}{244,165,130}
\definecolor{RdGy-4-F}{RGB}{244,165,130}
\definecolor{RdGy-4-3}{RGB}{186,186,186}
\definecolor{RdGy-4-J}{RGB}{186,186,186}
\definecolor{RdGy-4-4}{RGB}{64,64,64}
\definecolor{RdGy-4-M}{RGB}{64,64,64}
\definecolor{RdGy-5-1}{RGB}{202,0,32}
\definecolor{RdGy-5-C}{RGB}{202,0,32}
\definecolor{RdGy-5-2}{RGB}{244,165,130}
\definecolor{RdGy-5-F}{RGB}{244,165,130}
\definecolor{RdGy-5-3}{RGB}{255,255,255}
\definecolor{RdGy-5-H}{RGB}{255,255,255}
\definecolor{RdGy-5-4}{RGB}{186,186,186}
\definecolor{RdGy-5-J}{RGB}{186,186,186}
\definecolor{RdGy-5-5}{RGB}{64,64,64}
\definecolor{RdGy-5-M}{RGB}{64,64,64}
\definecolor{RdGy-6-1}{RGB}{178,24,43}
\definecolor{RdGy-6-B}{RGB}{178,24,43}
\definecolor{RdGy-6-2}{RGB}{239,138,98}
\definecolor{RdGy-6-E}{RGB}{239,138,98}
\definecolor{RdGy-6-3}{RGB}{253,219,199}
\definecolor{RdGy-6-G}{RGB}{253,219,199}
\definecolor{RdGy-6-4}{RGB}{224,224,224}
\definecolor{RdGy-6-I}{RGB}{224,224,224}
\definecolor{RdGy-6-5}{RGB}{153,153,153}
\definecolor{RdGy-6-K}{RGB}{153,153,153}
\definecolor{RdGy-6-6}{RGB}{77,77,77}
\definecolor{RdGy-6-N}{RGB}{77,77,77}
\definecolor{RdGy-7-1}{RGB}{178,24,43}
\definecolor{RdGy-7-B}{RGB}{178,24,43}
\definecolor{RdGy-7-2}{RGB}{239,138,98}
\definecolor{RdGy-7-E}{RGB}{239,138,98}
\definecolor{RdGy-7-3}{RGB}{253,219,199}
\definecolor{RdGy-7-G}{RGB}{253,219,199}
\definecolor{RdGy-7-4}{RGB}{255,255,255}
\definecolor{RdGy-7-H}{RGB}{255,255,255}
\definecolor{RdGy-7-5}{RGB}{224,224,224}
\definecolor{RdGy-7-I}{RGB}{224,224,224}
\definecolor{RdGy-7-6}{RGB}{153,153,153}
\definecolor{RdGy-7-K}{RGB}{153,153,153}
\definecolor{RdGy-7-7}{RGB}{77,77,77}
\definecolor{RdGy-7-N}{RGB}{77,77,77}
\definecolor{RdGy-8-1}{RGB}{178,24,43}
\definecolor{RdGy-8-B}{RGB}{178,24,43}
\definecolor{RdGy-8-2}{RGB}{214,96,77}
\definecolor{RdGy-8-D}{RGB}{214,96,77}
\definecolor{RdGy-8-3}{RGB}{244,165,130}
\definecolor{RdGy-8-F}{RGB}{244,165,130}
\definecolor{RdGy-8-4}{RGB}{253,219,199}
\definecolor{RdGy-8-G}{RGB}{253,219,199}
\definecolor{RdGy-8-5}{RGB}{224,224,224}
\definecolor{RdGy-8-I}{RGB}{224,224,224}
\definecolor{RdGy-8-6}{RGB}{186,186,186}
\definecolor{RdGy-8-J}{RGB}{186,186,186}
\definecolor{RdGy-8-7}{RGB}{135,135,135}
\definecolor{RdGy-8-L}{RGB}{135,135,135}
\definecolor{RdGy-8-8}{RGB}{77,77,77}
\definecolor{RdGy-8-N}{RGB}{77,77,77}
\definecolor{RdGy-9-1}{RGB}{178,24,43}
\definecolor{RdGy-9-B}{RGB}{178,24,43}
\definecolor{RdGy-9-2}{RGB}{214,96,77}
\definecolor{RdGy-9-D}{RGB}{214,96,77}
\definecolor{RdGy-9-3}{RGB}{244,165,130}
\definecolor{RdGy-9-F}{RGB}{244,165,130}
\definecolor{RdGy-9-4}{RGB}{253,219,199}
\definecolor{RdGy-9-G}{RGB}{253,219,199}
\definecolor{RdGy-9-5}{RGB}{255,255,255}
\definecolor{RdGy-9-H}{RGB}{255,255,255}
\definecolor{RdGy-9-6}{RGB}{224,224,224}
\definecolor{RdGy-9-I}{RGB}{224,224,224}
\definecolor{RdGy-9-7}{RGB}{186,186,186}
\definecolor{RdGy-9-J}{RGB}{186,186,186}
\definecolor{RdGy-9-8}{RGB}{135,135,135}
\definecolor{RdGy-9-L}{RGB}{135,135,135}
\definecolor{RdGy-9-9}{RGB}{77,77,77}
\definecolor{RdGy-9-N}{RGB}{77,77,77}
\definecolor{RdGy-10-1}{RGB}{103,0,31}
\definecolor{RdGy-10-A}{RGB}{103,0,31}
\definecolor{RdGy-10-2}{RGB}{178,24,43}
\definecolor{RdGy-10-B}{RGB}{178,24,43}
\definecolor{RdGy-10-3}{RGB}{214,96,77}
\definecolor{RdGy-10-D}{RGB}{214,96,77}
\definecolor{RdGy-10-4}{RGB}{244,165,130}
\definecolor{RdGy-10-F}{RGB}{244,165,130}
\definecolor{RdGy-10-5}{RGB}{253,219,199}
\definecolor{RdGy-10-G}{RGB}{253,219,199}
\definecolor{RdGy-10-6}{RGB}{224,224,224}
\definecolor{RdGy-10-I}{RGB}{224,224,224}
\definecolor{RdGy-10-7}{RGB}{186,186,186}
\definecolor{RdGy-10-J}{RGB}{186,186,186}
\definecolor{RdGy-10-8}{RGB}{135,135,135}
\definecolor{RdGy-10-L}{RGB}{135,135,135}
\definecolor{RdGy-10-9}{RGB}{77,77,77}
\definecolor{RdGy-10-N}{RGB}{77,77,77}
\definecolor{RdGy-10-10}{RGB}{26,26,26}
\definecolor{RdGy-10-O}{RGB}{26,26,26}
\definecolor{RdGy-11-1}{RGB}{103,0,31}
\definecolor{RdGy-11-A}{RGB}{103,0,31}
\definecolor{RdGy-11-2}{RGB}{178,24,43}
\definecolor{RdGy-11-B}{RGB}{178,24,43}
\definecolor{RdGy-11-3}{RGB}{214,96,77}
\definecolor{RdGy-11-D}{RGB}{214,96,77}
\definecolor{RdGy-11-4}{RGB}{244,165,130}
\definecolor{RdGy-11-F}{RGB}{244,165,130}
\definecolor{RdGy-11-5}{RGB}{253,219,199}
\definecolor{RdGy-11-G}{RGB}{253,219,199}
\definecolor{RdGy-11-6}{RGB}{255,255,255}
\definecolor{RdGy-11-H}{RGB}{255,255,255}
\definecolor{RdGy-11-7}{RGB}{224,224,224}
\definecolor{RdGy-11-I}{RGB}{224,224,224}
\definecolor{RdGy-11-8}{RGB}{186,186,186}
\definecolor{RdGy-11-J}{RGB}{186,186,186}
\definecolor{RdGy-11-9}{RGB}{135,135,135}
\definecolor{RdGy-11-L}{RGB}{135,135,135}
\definecolor{RdGy-11-10}{RGB}{77,77,77}
\definecolor{RdGy-11-N}{RGB}{77,77,77}
\definecolor{RdGy-11-11}{RGB}{26,26,26}
\definecolor{RdGy-11-O}{RGB}{26,26,26}
\definecolor{RdYlBu-3-1}{RGB}{252,141,89}
\definecolor{RdYlBu-3-E}{RGB}{252,141,89}
\definecolor{RdYlBu-3-2}{RGB}{255,255,191}
\definecolor{RdYlBu-3-H}{RGB}{255,255,191}
\definecolor{RdYlBu-3-3}{RGB}{145,191,219}
\definecolor{RdYlBu-3-K}{RGB}{145,191,219}
\definecolor{RdYlBu-4-1}{RGB}{215,25,28}
\definecolor{RdYlBu-4-C}{RGB}{215,25,28}
\definecolor{RdYlBu-4-2}{RGB}{253,174,97}
\definecolor{RdYlBu-4-F}{RGB}{253,174,97}
\definecolor{RdYlBu-4-3}{RGB}{171,217,233}
\definecolor{RdYlBu-4-J}{RGB}{171,217,233}
\definecolor{RdYlBu-4-4}{RGB}{44,123,182}
\definecolor{RdYlBu-4-M}{RGB}{44,123,182}
\definecolor{RdYlBu-5-1}{RGB}{215,25,28}
\definecolor{RdYlBu-5-C}{RGB}{215,25,28}
\definecolor{RdYlBu-5-2}{RGB}{253,174,97}
\definecolor{RdYlBu-5-F}{RGB}{253,174,97}
\definecolor{RdYlBu-5-3}{RGB}{255,255,191}
\definecolor{RdYlBu-5-H}{RGB}{255,255,191}
\definecolor{RdYlBu-5-4}{RGB}{171,217,233}
\definecolor{RdYlBu-5-J}{RGB}{171,217,233}
\definecolor{RdYlBu-5-5}{RGB}{44,123,182}
\definecolor{RdYlBu-5-M}{RGB}{44,123,182}
\definecolor{RdYlBu-6-1}{RGB}{215,48,39}
\definecolor{RdYlBu-6-B}{RGB}{215,48,39}
\definecolor{RdYlBu-6-2}{RGB}{252,141,89}
\definecolor{RdYlBu-6-E}{RGB}{252,141,89}
\definecolor{RdYlBu-6-3}{RGB}{254,224,144}
\definecolor{RdYlBu-6-G}{RGB}{254,224,144}
\definecolor{RdYlBu-6-4}{RGB}{224,243,248}
\definecolor{RdYlBu-6-I}{RGB}{224,243,248}
\definecolor{RdYlBu-6-5}{RGB}{145,191,219}
\definecolor{RdYlBu-6-K}{RGB}{145,191,219}
\definecolor{RdYlBu-6-6}{RGB}{69,117,180}
\definecolor{RdYlBu-6-N}{RGB}{69,117,180}
\definecolor{RdYlBu-7-1}{RGB}{215,48,39}
\definecolor{RdYlBu-7-B}{RGB}{215,48,39}
\definecolor{RdYlBu-7-2}{RGB}{252,141,89}
\definecolor{RdYlBu-7-E}{RGB}{252,141,89}
\definecolor{RdYlBu-7-3}{RGB}{254,224,144}
\definecolor{RdYlBu-7-G}{RGB}{254,224,144}
\definecolor{RdYlBu-7-4}{RGB}{255,255,191}
\definecolor{RdYlBu-7-H}{RGB}{255,255,191}
\definecolor{RdYlBu-7-5}{RGB}{224,243,248}
\definecolor{RdYlBu-7-I}{RGB}{224,243,248}
\definecolor{RdYlBu-7-6}{RGB}{145,191,219}
\definecolor{RdYlBu-7-K}{RGB}{145,191,219}
\definecolor{RdYlBu-7-7}{RGB}{69,117,180}
\definecolor{RdYlBu-7-N}{RGB}{69,117,180}
\definecolor{RdYlBu-8-1}{RGB}{215,48,39}
\definecolor{RdYlBu-8-B}{RGB}{215,48,39}
\definecolor{RdYlBu-8-2}{RGB}{244,109,67}
\definecolor{RdYlBu-8-D}{RGB}{244,109,67}
\definecolor{RdYlBu-8-3}{RGB}{253,174,97}
\definecolor{RdYlBu-8-F}{RGB}{253,174,97}
\definecolor{RdYlBu-8-4}{RGB}{254,224,144}
\definecolor{RdYlBu-8-G}{RGB}{254,224,144}
\definecolor{RdYlBu-8-5}{RGB}{224,243,248}
\definecolor{RdYlBu-8-I}{RGB}{224,243,248}
\definecolor{RdYlBu-8-6}{RGB}{171,217,233}
\definecolor{RdYlBu-8-J}{RGB}{171,217,233}
\definecolor{RdYlBu-8-7}{RGB}{116,173,209}
\definecolor{RdYlBu-8-L}{RGB}{116,173,209}
\definecolor{RdYlBu-8-8}{RGB}{69,117,180}
\definecolor{RdYlBu-8-N}{RGB}{69,117,180}
\definecolor{RdYlBu-9-1}{RGB}{215,48,39}
\definecolor{RdYlBu-9-B}{RGB}{215,48,39}
\definecolor{RdYlBu-9-2}{RGB}{244,109,67}
\definecolor{RdYlBu-9-D}{RGB}{244,109,67}
\definecolor{RdYlBu-9-3}{RGB}{253,174,97}
\definecolor{RdYlBu-9-F}{RGB}{253,174,97}
\definecolor{RdYlBu-9-4}{RGB}{254,224,144}
\definecolor{RdYlBu-9-G}{RGB}{254,224,144}
\definecolor{RdYlBu-9-5}{RGB}{255,255,191}
\definecolor{RdYlBu-9-H}{RGB}{255,255,191}
\definecolor{RdYlBu-9-6}{RGB}{224,243,248}
\definecolor{RdYlBu-9-I}{RGB}{224,243,248}
\definecolor{RdYlBu-9-7}{RGB}{171,217,233}
\definecolor{RdYlBu-9-J}{RGB}{171,217,233}
\definecolor{RdYlBu-9-8}{RGB}{116,173,209}
\definecolor{RdYlBu-9-L}{RGB}{116,173,209}
\definecolor{RdYlBu-9-9}{RGB}{69,117,180}
\definecolor{RdYlBu-9-N}{RGB}{69,117,180}
\definecolor{RdYlBu-10-1}{RGB}{165,0,38}
\definecolor{RdYlBu-10-A}{RGB}{165,0,38}
\definecolor{RdYlBu-10-2}{RGB}{215,48,39}
\definecolor{RdYlBu-10-B}{RGB}{215,48,39}
\definecolor{RdYlBu-10-3}{RGB}{244,109,67}
\definecolor{RdYlBu-10-D}{RGB}{244,109,67}
\definecolor{RdYlBu-10-4}{RGB}{253,174,97}
\definecolor{RdYlBu-10-F}{RGB}{253,174,97}
\definecolor{RdYlBu-10-5}{RGB}{254,224,144}
\definecolor{RdYlBu-10-G}{RGB}{254,224,144}
\definecolor{RdYlBu-10-6}{RGB}{224,243,248}
\definecolor{RdYlBu-10-I}{RGB}{224,243,248}
\definecolor{RdYlBu-10-7}{RGB}{171,217,233}
\definecolor{RdYlBu-10-J}{RGB}{171,217,233}
\definecolor{RdYlBu-10-8}{RGB}{116,173,209}
\definecolor{RdYlBu-10-L}{RGB}{116,173,209}
\definecolor{RdYlBu-10-9}{RGB}{69,117,180}
\definecolor{RdYlBu-10-N}{RGB}{69,117,180}
\definecolor{RdYlBu-10-10}{RGB}{49,54,149}
\definecolor{RdYlBu-10-O}{RGB}{49,54,149}
\definecolor{RdYlBu-11-1}{RGB}{165,0,38}
\definecolor{RdYlBu-11-A}{RGB}{165,0,38}
\definecolor{RdYlBu-11-2}{RGB}{215,48,39}
\definecolor{RdYlBu-11-B}{RGB}{215,48,39}
\definecolor{RdYlBu-11-3}{RGB}{244,109,67}
\definecolor{RdYlBu-11-D}{RGB}{244,109,67}
\definecolor{RdYlBu-11-4}{RGB}{253,174,97}
\definecolor{RdYlBu-11-F}{RGB}{253,174,97}
\definecolor{RdYlBu-11-5}{RGB}{254,224,144}
\definecolor{RdYlBu-11-G}{RGB}{254,224,144}
\definecolor{RdYlBu-11-6}{RGB}{255,255,191}
\definecolor{RdYlBu-11-H}{RGB}{255,255,191}
\definecolor{RdYlBu-11-7}{RGB}{224,243,248}
\definecolor{RdYlBu-11-I}{RGB}{224,243,248}
\definecolor{RdYlBu-11-8}{RGB}{171,217,233}
\definecolor{RdYlBu-11-J}{RGB}{171,217,233}
\definecolor{RdYlBu-11-9}{RGB}{116,173,209}
\definecolor{RdYlBu-11-L}{RGB}{116,173,209}
\definecolor{RdYlBu-11-10}{RGB}{69,117,180}
\definecolor{RdYlBu-11-N}{RGB}{69,117,180}
\definecolor{RdYlBu-11-11}{RGB}{49,54,149}
\definecolor{RdYlBu-11-O}{RGB}{49,54,149}
\definecolor{Spectral-3-1}{RGB}{252,141,89}
\definecolor{Spectral-3-E}{RGB}{252,141,89}
\definecolor{Spectral-3-2}{RGB}{255,255,191}
\definecolor{Spectral-3-H}{RGB}{255,255,191}
\definecolor{Spectral-3-3}{RGB}{153,213,148}
\definecolor{Spectral-3-K}{RGB}{153,213,148}
\definecolor{Spectral-4-1}{RGB}{215,25,28}
\definecolor{Spectral-4-C}{RGB}{215,25,28}
\definecolor{Spectral-4-2}{RGB}{253,174,97}
\definecolor{Spectral-4-F}{RGB}{253,174,97}
\definecolor{Spectral-4-3}{RGB}{171,221,164}
\definecolor{Spectral-4-J}{RGB}{171,221,164}
\definecolor{Spectral-4-4}{RGB}{43,131,186}
\definecolor{Spectral-4-M}{RGB}{43,131,186}
\definecolor{Spectral-5-1}{RGB}{215,25,28}
\definecolor{Spectral-5-C}{RGB}{215,25,28}
\definecolor{Spectral-5-2}{RGB}{253,174,97}
\definecolor{Spectral-5-F}{RGB}{253,174,97}
\definecolor{Spectral-5-3}{RGB}{255,255,191}
\definecolor{Spectral-5-H}{RGB}{255,255,191}
\definecolor{Spectral-5-4}{RGB}{171,221,164}
\definecolor{Spectral-5-J}{RGB}{171,221,164}
\definecolor{Spectral-5-5}{RGB}{43,131,186}
\definecolor{Spectral-5-M}{RGB}{43,131,186}
\definecolor{Spectral-6-1}{RGB}{213,62,79}
\definecolor{Spectral-6-B}{RGB}{213,62,79}
\definecolor{Spectral-6-2}{RGB}{252,141,89}
\definecolor{Spectral-6-E}{RGB}{252,141,89}
\definecolor{Spectral-6-3}{RGB}{254,224,139}
\definecolor{Spectral-6-G}{RGB}{254,224,139}
\definecolor{Spectral-6-4}{RGB}{230,245,152}
\definecolor{Spectral-6-I}{RGB}{230,245,152}
\definecolor{Spectral-6-5}{RGB}{153,213,148}
\definecolor{Spectral-6-K}{RGB}{153,213,148}
\definecolor{Spectral-6-6}{RGB}{50,136,189}
\definecolor{Spectral-6-N}{RGB}{50,136,189}
\definecolor{Spectral-7-1}{RGB}{213,62,79}
\definecolor{Spectral-7-B}{RGB}{213,62,79}
\definecolor{Spectral-7-2}{RGB}{252,141,89}
\definecolor{Spectral-7-E}{RGB}{252,141,89}
\definecolor{Spectral-7-3}{RGB}{254,224,139}
\definecolor{Spectral-7-G}{RGB}{254,224,139}
\definecolor{Spectral-7-4}{RGB}{255,255,191}
\definecolor{Spectral-7-H}{RGB}{255,255,191}
\definecolor{Spectral-7-5}{RGB}{230,245,152}
\definecolor{Spectral-7-I}{RGB}{230,245,152}
\definecolor{Spectral-7-6}{RGB}{153,213,148}
\definecolor{Spectral-7-K}{RGB}{153,213,148}
\definecolor{Spectral-7-7}{RGB}{50,136,189}
\definecolor{Spectral-7-N}{RGB}{50,136,189}
\definecolor{Spectral-8-1}{RGB}{213,62,79}
\definecolor{Spectral-8-B}{RGB}{213,62,79}
\definecolor{Spectral-8-2}{RGB}{244,109,67}
\definecolor{Spectral-8-D}{RGB}{244,109,67}
\definecolor{Spectral-8-3}{RGB}{253,174,97}
\definecolor{Spectral-8-F}{RGB}{253,174,97}
\definecolor{Spectral-8-4}{RGB}{254,224,139}
\definecolor{Spectral-8-G}{RGB}{254,224,139}
\definecolor{Spectral-8-5}{RGB}{230,245,152}
\definecolor{Spectral-8-I}{RGB}{230,245,152}
\definecolor{Spectral-8-6}{RGB}{171,221,164}
\definecolor{Spectral-8-J}{RGB}{171,221,164}
\definecolor{Spectral-8-7}{RGB}{102,194,165}
\definecolor{Spectral-8-L}{RGB}{102,194,165}
\definecolor{Spectral-8-8}{RGB}{50,136,189}
\definecolor{Spectral-8-N}{RGB}{50,136,189}
\definecolor{Spectral-9-1}{RGB}{213,62,79}
\definecolor{Spectral-9-B}{RGB}{213,62,79}
\definecolor{Spectral-9-2}{RGB}{244,109,67}
\definecolor{Spectral-9-D}{RGB}{244,109,67}
\definecolor{Spectral-9-3}{RGB}{253,174,97}
\definecolor{Spectral-9-F}{RGB}{253,174,97}
\definecolor{Spectral-9-4}{RGB}{254,224,139}
\definecolor{Spectral-9-G}{RGB}{254,224,139}
\definecolor{Spectral-9-5}{RGB}{255,255,191}
\definecolor{Spectral-9-H}{RGB}{255,255,191}
\definecolor{Spectral-9-6}{RGB}{230,245,152}
\definecolor{Spectral-9-I}{RGB}{230,245,152}
\definecolor{Spectral-9-7}{RGB}{171,221,164}
\definecolor{Spectral-9-J}{RGB}{171,221,164}
\definecolor{Spectral-9-8}{RGB}{102,194,165}
\definecolor{Spectral-9-L}{RGB}{102,194,165}
\definecolor{Spectral-9-9}{RGB}{50,136,189}
\definecolor{Spectral-9-N}{RGB}{50,136,189}
\definecolor{Spectral-10-1}{RGB}{158,1,66}
\definecolor{Spectral-10-A}{RGB}{158,1,66}
\definecolor{Spectral-10-2}{RGB}{213,62,79}
\definecolor{Spectral-10-B}{RGB}{213,62,79}
\definecolor{Spectral-10-3}{RGB}{244,109,67}
\definecolor{Spectral-10-D}{RGB}{244,109,67}
\definecolor{Spectral-10-4}{RGB}{253,174,97}
\definecolor{Spectral-10-F}{RGB}{253,174,97}
\definecolor{Spectral-10-5}{RGB}{254,224,139}
\definecolor{Spectral-10-G}{RGB}{254,224,139}
\definecolor{Spectral-10-6}{RGB}{230,245,152}
\definecolor{Spectral-10-I}{RGB}{230,245,152}
\definecolor{Spectral-10-7}{RGB}{171,221,164}
\definecolor{Spectral-10-J}{RGB}{171,221,164}
\definecolor{Spectral-10-8}{RGB}{102,194,165}
\definecolor{Spectral-10-L}{RGB}{102,194,165}
\definecolor{Spectral-10-9}{RGB}{50,136,189}
\definecolor{Spectral-10-N}{RGB}{50,136,189}
\definecolor{Spectral-10-10}{RGB}{94,79,162}
\definecolor{Spectral-10-O}{RGB}{94,79,162}
\definecolor{Spectral-11-1}{RGB}{158,1,66}
\definecolor{Spectral-11-A}{RGB}{158,1,66}
\definecolor{Spectral-11-2}{RGB}{213,62,79}
\definecolor{Spectral-11-B}{RGB}{213,62,79}
\definecolor{Spectral-11-3}{RGB}{244,109,67}
\definecolor{Spectral-11-D}{RGB}{244,109,67}
\definecolor{Spectral-11-4}{RGB}{253,174,97}
\definecolor{Spectral-11-F}{RGB}{253,174,97}
\definecolor{Spectral-11-5}{RGB}{254,224,139}
\definecolor{Spectral-11-G}{RGB}{254,224,139}
\definecolor{Spectral-11-6}{RGB}{255,255,191}
\definecolor{Spectral-11-H}{RGB}{255,255,191}
\definecolor{Spectral-11-7}{RGB}{230,245,152}
\definecolor{Spectral-11-I}{RGB}{230,245,152}
\definecolor{Spectral-11-8}{RGB}{171,221,164}
\definecolor{Spectral-11-J}{RGB}{171,221,164}
\definecolor{Spectral-11-9}{RGB}{102,194,165}
\definecolor{Spectral-11-L}{RGB}{102,194,165}
\definecolor{Spectral-11-10}{RGB}{50,136,189}
\definecolor{Spectral-11-N}{RGB}{50,136,189}
\definecolor{Spectral-11-11}{RGB}{94,79,162}
\definecolor{Spectral-11-O}{RGB}{94,79,162}
\definecolor{RdYlGn-3-1}{RGB}{252,141,89}
\definecolor{RdYlGn-3-E}{RGB}{252,141,89}
\definecolor{RdYlGn-3-2}{RGB}{255,255,191}
\definecolor{RdYlGn-3-H}{RGB}{255,255,191}
\definecolor{RdYlGn-3-3}{RGB}{145,207,96}
\definecolor{RdYlGn-3-K}{RGB}{145,207,96}
\definecolor{RdYlGn-4-1}{RGB}{215,25,28}
\definecolor{RdYlGn-4-C}{RGB}{215,25,28}
\definecolor{RdYlGn-4-2}{RGB}{253,174,97}
\definecolor{RdYlGn-4-F}{RGB}{253,174,97}
\definecolor{RdYlGn-4-3}{RGB}{166,217,106}
\definecolor{RdYlGn-4-J}{RGB}{166,217,106}
\definecolor{RdYlGn-4-4}{RGB}{26,150,65}
\definecolor{RdYlGn-4-M}{RGB}{26,150,65}
\definecolor{RdYlGn-5-1}{RGB}{215,25,28}
\definecolor{RdYlGn-5-C}{RGB}{215,25,28}
\definecolor{RdYlGn-5-2}{RGB}{253,174,97}
\definecolor{RdYlGn-5-F}{RGB}{253,174,97}
\definecolor{RdYlGn-5-3}{RGB}{255,255,191}
\definecolor{RdYlGn-5-H}{RGB}{255,255,191}
\definecolor{RdYlGn-5-4}{RGB}{166,217,106}
\definecolor{RdYlGn-5-J}{RGB}{166,217,106}
\definecolor{RdYlGn-5-5}{RGB}{26,150,65}
\definecolor{RdYlGn-5-M}{RGB}{26,150,65}
\definecolor{RdYlGn-6-1}{RGB}{215,48,39}
\definecolor{RdYlGn-6-B}{RGB}{215,48,39}
\definecolor{RdYlGn-6-2}{RGB}{252,141,89}
\definecolor{RdYlGn-6-E}{RGB}{252,141,89}
\definecolor{RdYlGn-6-3}{RGB}{254,224,139}
\definecolor{RdYlGn-6-G}{RGB}{254,224,139}
\definecolor{RdYlGn-6-4}{RGB}{217,239,139}
\definecolor{RdYlGn-6-I}{RGB}{217,239,139}
\definecolor{RdYlGn-6-5}{RGB}{145,207,96}
\definecolor{RdYlGn-6-K}{RGB}{145,207,96}
\definecolor{RdYlGn-6-6}{RGB}{26,152,80}
\definecolor{RdYlGn-6-N}{RGB}{26,152,80}
\definecolor{RdYlGn-7-1}{RGB}{215,48,39}
\definecolor{RdYlGn-7-B}{RGB}{215,48,39}
\definecolor{RdYlGn-7-2}{RGB}{252,141,89}
\definecolor{RdYlGn-7-E}{RGB}{252,141,89}
\definecolor{RdYlGn-7-3}{RGB}{254,224,139}
\definecolor{RdYlGn-7-G}{RGB}{254,224,139}
\definecolor{RdYlGn-7-4}{RGB}{255,255,191}
\definecolor{RdYlGn-7-H}{RGB}{255,255,191}
\definecolor{RdYlGn-7-5}{RGB}{217,239,139}
\definecolor{RdYlGn-7-I}{RGB}{217,239,139}
\definecolor{RdYlGn-7-6}{RGB}{145,207,96}
\definecolor{RdYlGn-7-K}{RGB}{145,207,96}
\definecolor{RdYlGn-7-7}{RGB}{26,152,80}
\definecolor{RdYlGn-7-N}{RGB}{26,152,80}
\definecolor{RdYlGn-8-1}{RGB}{215,48,39}
\definecolor{RdYlGn-8-B}{RGB}{215,48,39}
\definecolor{RdYlGn-8-2}{RGB}{244,109,67}
\definecolor{RdYlGn-8-D}{RGB}{244,109,67}
\definecolor{RdYlGn-8-3}{RGB}{253,174,97}
\definecolor{RdYlGn-8-F}{RGB}{253,174,97}
\definecolor{RdYlGn-8-4}{RGB}{254,224,139}
\definecolor{RdYlGn-8-G}{RGB}{254,224,139}
\definecolor{RdYlGn-8-5}{RGB}{217,239,139}
\definecolor{RdYlGn-8-I}{RGB}{217,239,139}
\definecolor{RdYlGn-8-6}{RGB}{166,217,106}
\definecolor{RdYlGn-8-J}{RGB}{166,217,106}
\definecolor{RdYlGn-8-7}{RGB}{102,189,99}
\definecolor{RdYlGn-8-L}{RGB}{102,189,99}
\definecolor{RdYlGn-8-8}{RGB}{26,152,80}
\definecolor{RdYlGn-8-N}{RGB}{26,152,80}
\definecolor{RdYlGn-9-1}{RGB}{215,48,39}
\definecolor{RdYlGn-9-B}{RGB}{215,48,39}
\definecolor{RdYlGn-9-2}{RGB}{244,109,67}
\definecolor{RdYlGn-9-D}{RGB}{244,109,67}
\definecolor{RdYlGn-9-3}{RGB}{253,174,97}
\definecolor{RdYlGn-9-F}{RGB}{253,174,97}
\definecolor{RdYlGn-9-4}{RGB}{254,224,139}
\definecolor{RdYlGn-9-G}{RGB}{254,224,139}
\definecolor{RdYlGn-9-5}{RGB}{255,255,191}
\definecolor{RdYlGn-9-H}{RGB}{255,255,191}
\definecolor{RdYlGn-9-6}{RGB}{217,239,139}
\definecolor{RdYlGn-9-I}{RGB}{217,239,139}
\definecolor{RdYlGn-9-7}{RGB}{166,217,106}
\definecolor{RdYlGn-9-J}{RGB}{166,217,106}
\definecolor{RdYlGn-9-8}{RGB}{102,189,99}
\definecolor{RdYlGn-9-L}{RGB}{102,189,99}
\definecolor{RdYlGn-9-9}{RGB}{26,152,80}
\definecolor{RdYlGn-9-N}{RGB}{26,152,80}
\definecolor{RdYlGn-10-1}{RGB}{165,0,38}
\definecolor{RdYlGn-10-A}{RGB}{165,0,38}
\definecolor{RdYlGn-10-2}{RGB}{215,48,39}
\definecolor{RdYlGn-10-B}{RGB}{215,48,39}
\definecolor{RdYlGn-10-3}{RGB}{244,109,67}
\definecolor{RdYlGn-10-D}{RGB}{244,109,67}
\definecolor{RdYlGn-10-4}{RGB}{253,174,97}
\definecolor{RdYlGn-10-F}{RGB}{253,174,97}
\definecolor{RdYlGn-10-5}{RGB}{254,224,139}
\definecolor{RdYlGn-10-G}{RGB}{254,224,139}
\definecolor{RdYlGn-10-6}{RGB}{217,239,139}
\definecolor{RdYlGn-10-I}{RGB}{217,239,139}
\definecolor{RdYlGn-10-7}{RGB}{166,217,106}
\definecolor{RdYlGn-10-J}{RGB}{166,217,106}
\definecolor{RdYlGn-10-8}{RGB}{102,189,99}
\definecolor{RdYlGn-10-L}{RGB}{102,189,99}
\definecolor{RdYlGn-10-9}{RGB}{26,152,80}
\definecolor{RdYlGn-10-N}{RGB}{26,152,80}
\definecolor{RdYlGn-10-10}{RGB}{0,104,55}
\definecolor{RdYlGn-10-O}{RGB}{0,104,55}
\definecolor{RdYlGn-11-1}{RGB}{165,0,38}
\definecolor{RdYlGn-11-A}{RGB}{165,0,38}
\definecolor{RdYlGn-11-2}{RGB}{215,48,39}
\definecolor{RdYlGn-11-B}{RGB}{215,48,39}
\definecolor{RdYlGn-11-3}{RGB}{244,109,67}
\definecolor{RdYlGn-11-D}{RGB}{244,109,67}
\definecolor{RdYlGn-11-4}{RGB}{253,174,97}
\definecolor{RdYlGn-11-F}{RGB}{253,174,97}
\definecolor{RdYlGn-11-5}{RGB}{254,224,139}
\definecolor{RdYlGn-11-G}{RGB}{254,224,139}
\definecolor{RdYlGn-11-6}{RGB}{255,255,191}
\definecolor{RdYlGn-11-H}{RGB}{255,255,191}
\definecolor{RdYlGn-11-7}{RGB}{217,239,139}
\definecolor{RdYlGn-11-I}{RGB}{217,239,139}
\definecolor{RdYlGn-11-8}{RGB}{166,217,106}
\definecolor{RdYlGn-11-J}{RGB}{166,217,106}
\definecolor{RdYlGn-11-9}{RGB}{102,189,99}
\definecolor{RdYlGn-11-L}{RGB}{102,189,99}
\definecolor{RdYlGn-11-10}{RGB}{26,152,80}
\definecolor{RdYlGn-11-N}{RGB}{26,152,80}
\definecolor{RdYlGn-11-11}{RGB}{0,104,55}
\definecolor{RdYlGn-11-O}{RGB}{0,104,55}
\definecolor{Set3-3-1}{RGB}{141,211,199}
\definecolor{Set3-3-A}{RGB}{141,211,199}
\definecolor{Set3-3-2}{RGB}{255,255,179}
\definecolor{Set3-3-B}{RGB}{255,255,179}
\definecolor{Set3-3-3}{RGB}{190,186,218}
\definecolor{Set3-3-C}{RGB}{190,186,218}
\definecolor{Set3-4-1}{RGB}{141,211,199}
\definecolor{Set3-4-A}{RGB}{141,211,199}
\definecolor{Set3-4-2}{RGB}{255,255,179}
\definecolor{Set3-4-B}{RGB}{255,255,179}
\definecolor{Set3-4-3}{RGB}{190,186,218}
\definecolor{Set3-4-C}{RGB}{190,186,218}
\definecolor{Set3-4-4}{RGB}{251,128,114}
\definecolor{Set3-4-D}{RGB}{251,128,114}
\definecolor{Set3-5-1}{RGB}{141,211,199}
\definecolor{Set3-5-A}{RGB}{141,211,199}
\definecolor{Set3-5-2}{RGB}{255,255,179}
\definecolor{Set3-5-B}{RGB}{255,255,179}
\definecolor{Set3-5-3}{RGB}{190,186,218}
\definecolor{Set3-5-C}{RGB}{190,186,218}
\definecolor{Set3-5-4}{RGB}{251,128,114}
\definecolor{Set3-5-D}{RGB}{251,128,114}
\definecolor{Set3-5-5}{RGB}{128,177,211}
\definecolor{Set3-5-E}{RGB}{128,177,211}
\definecolor{Set3-6-1}{RGB}{141,211,199}
\definecolor{Set3-6-A}{RGB}{141,211,199}
\definecolor{Set3-6-2}{RGB}{255,255,179}
\definecolor{Set3-6-B}{RGB}{255,255,179}
\definecolor{Set3-6-3}{RGB}{190,186,218}
\definecolor{Set3-6-C}{RGB}{190,186,218}
\definecolor{Set3-6-4}{RGB}{251,128,114}
\definecolor{Set3-6-D}{RGB}{251,128,114}
\definecolor{Set3-6-5}{RGB}{128,177,211}
\definecolor{Set3-6-E}{RGB}{128,177,211}
\definecolor{Set3-6-6}{RGB}{253,180,98}
\definecolor{Set3-6-F}{RGB}{253,180,98}
\definecolor{Set3-7-1}{RGB}{141,211,199}
\definecolor{Set3-7-A}{RGB}{141,211,199}
\definecolor{Set3-7-2}{RGB}{255,255,179}
\definecolor{Set3-7-B}{RGB}{255,255,179}
\definecolor{Set3-7-3}{RGB}{190,186,218}
\definecolor{Set3-7-C}{RGB}{190,186,218}
\definecolor{Set3-7-4}{RGB}{251,128,114}
\definecolor{Set3-7-D}{RGB}{251,128,114}
\definecolor{Set3-7-5}{RGB}{128,177,211}
\definecolor{Set3-7-E}{RGB}{128,177,211}
\definecolor{Set3-7-6}{RGB}{253,180,98}
\definecolor{Set3-7-F}{RGB}{253,180,98}
\definecolor{Set3-7-7}{RGB}{179,222,105}
\definecolor{Set3-7-G}{RGB}{179,222,105}
\definecolor{Set3-8-1}{RGB}{141,211,199}
\definecolor{Set3-8-A}{RGB}{141,211,199}
\definecolor{Set3-8-2}{RGB}{255,255,179}
\definecolor{Set3-8-B}{RGB}{255,255,179}
\definecolor{Set3-8-3}{RGB}{190,186,218}
\definecolor{Set3-8-C}{RGB}{190,186,218}
\definecolor{Set3-8-4}{RGB}{251,128,114}
\definecolor{Set3-8-D}{RGB}{251,128,114}
\definecolor{Set3-8-5}{RGB}{128,177,211}
\definecolor{Set3-8-E}{RGB}{128,177,211}
\definecolor{Set3-8-6}{RGB}{253,180,98}
\definecolor{Set3-8-F}{RGB}{253,180,98}
\definecolor{Set3-8-7}{RGB}{179,222,105}
\definecolor{Set3-8-G}{RGB}{179,222,105}
\definecolor{Set3-8-8}{RGB}{252,205,229}
\definecolor{Set3-8-H}{RGB}{252,205,229}
\definecolor{Set3-9-1}{RGB}{141,211,199}
\definecolor{Set3-9-A}{RGB}{141,211,199}
\definecolor{Set3-9-2}{RGB}{255,255,179}
\definecolor{Set3-9-B}{RGB}{255,255,179}
\definecolor{Set3-9-3}{RGB}{190,186,218}
\definecolor{Set3-9-C}{RGB}{190,186,218}
\definecolor{Set3-9-4}{RGB}{251,128,114}
\definecolor{Set3-9-D}{RGB}{251,128,114}
\definecolor{Set3-9-5}{RGB}{128,177,211}
\definecolor{Set3-9-E}{RGB}{128,177,211}
\definecolor{Set3-9-6}{RGB}{253,180,98}
\definecolor{Set3-9-F}{RGB}{253,180,98}
\definecolor{Set3-9-7}{RGB}{179,222,105}
\definecolor{Set3-9-G}{RGB}{179,222,105}
\definecolor{Set3-9-8}{RGB}{252,205,229}
\definecolor{Set3-9-H}{RGB}{252,205,229}
\definecolor{Set3-9-9}{RGB}{217,217,217}
\definecolor{Set3-9-I}{RGB}{217,217,217}
\definecolor{Set3-10-1}{RGB}{141,211,199}
\definecolor{Set3-10-A}{RGB}{141,211,199}
\definecolor{Set3-10-2}{RGB}{255,255,179}
\definecolor{Set3-10-B}{RGB}{255,255,179}
\definecolor{Set3-10-3}{RGB}{190,186,218}
\definecolor{Set3-10-C}{RGB}{190,186,218}
\definecolor{Set3-10-4}{RGB}{251,128,114}
\definecolor{Set3-10-D}{RGB}{251,128,114}
\definecolor{Set3-10-5}{RGB}{128,177,211}
\definecolor{Set3-10-E}{RGB}{128,177,211}
\definecolor{Set3-10-6}{RGB}{253,180,98}
\definecolor{Set3-10-F}{RGB}{253,180,98}
\definecolor{Set3-10-7}{RGB}{179,222,105}
\definecolor{Set3-10-G}{RGB}{179,222,105}
\definecolor{Set3-10-8}{RGB}{252,205,229}
\definecolor{Set3-10-H}{RGB}{252,205,229}
\definecolor{Set3-10-9}{RGB}{217,217,217}
\definecolor{Set3-10-I}{RGB}{217,217,217}
\definecolor{Set3-10-10}{RGB}{188,128,189}
\definecolor{Set3-10-J}{RGB}{188,128,189}
\definecolor{Set3-11-1}{RGB}{141,211,199}
\definecolor{Set3-11-A}{RGB}{141,211,199}
\definecolor{Set3-11-2}{RGB}{255,255,179}
\definecolor{Set3-11-B}{RGB}{255,255,179}
\definecolor{Set3-11-3}{RGB}{190,186,218}
\definecolor{Set3-11-C}{RGB}{190,186,218}
\definecolor{Set3-11-4}{RGB}{251,128,114}
\definecolor{Set3-11-D}{RGB}{251,128,114}
\definecolor{Set3-11-5}{RGB}{128,177,211}
\definecolor{Set3-11-E}{RGB}{128,177,211}
\definecolor{Set3-11-6}{RGB}{253,180,98}
\definecolor{Set3-11-F}{RGB}{253,180,98}
\definecolor{Set3-11-7}{RGB}{179,222,105}
\definecolor{Set3-11-G}{RGB}{179,222,105}
\definecolor{Set3-11-8}{RGB}{252,205,229}
\definecolor{Set3-11-H}{RGB}{252,205,229}
\definecolor{Set3-11-9}{RGB}{217,217,217}
\definecolor{Set3-11-I}{RGB}{217,217,217}
\definecolor{Set3-11-10}{RGB}{188,128,189}
\definecolor{Set3-11-J}{RGB}{188,128,189}
\definecolor{Set3-11-11}{RGB}{204,235,197}
\definecolor{Set3-11-K}{RGB}{204,235,197}
\definecolor{Set3-12-1}{RGB}{141,211,199}
\definecolor{Set3-12-A}{RGB}{141,211,199}
\definecolor{Set3-12-2}{RGB}{255,255,179}
\definecolor{Set3-12-B}{RGB}{255,255,179}
\definecolor{Set3-12-3}{RGB}{190,186,218}
\definecolor{Set3-12-C}{RGB}{190,186,218}
\definecolor{Set3-12-4}{RGB}{251,128,114}
\definecolor{Set3-12-D}{RGB}{251,128,114}
\definecolor{Set3-12-5}{RGB}{128,177,211}
\definecolor{Set3-12-E}{RGB}{128,177,211}
\definecolor{Set3-12-6}{RGB}{253,180,98}
\definecolor{Set3-12-F}{RGB}{253,180,98}
\definecolor{Set3-12-7}{RGB}{179,222,105}
\definecolor{Set3-12-G}{RGB}{179,222,105}
\definecolor{Set3-12-8}{RGB}{252,205,229}
\definecolor{Set3-12-H}{RGB}{252,205,229}
\definecolor{Set3-12-9}{RGB}{217,217,217}
\definecolor{Set3-12-I}{RGB}{217,217,217}
\definecolor{Set3-12-10}{RGB}{188,128,189}
\definecolor{Set3-12-J}{RGB}{188,128,189}
\definecolor{Set3-12-11}{RGB}{204,235,197}
\definecolor{Set3-12-K}{RGB}{204,235,197}
\definecolor{Set3-12-12}{RGB}{255,237,111}
\definecolor{Set3-12-L}{RGB}{255,237,111}
\definecolor{Pastel1-3-1}{RGB}{251,180,174}
\definecolor{Pastel1-3-A}{RGB}{251,180,174}
\definecolor{Pastel1-3-2}{RGB}{179,205,227}
\definecolor{Pastel1-3-B}{RGB}{179,205,227}
\definecolor{Pastel1-3-3}{RGB}{204,235,197}
\definecolor{Pastel1-3-C}{RGB}{204,235,197}
\definecolor{Pastel1-4-1}{RGB}{251,180,174}
\definecolor{Pastel1-4-A}{RGB}{251,180,174}
\definecolor{Pastel1-4-2}{RGB}{179,205,227}
\definecolor{Pastel1-4-B}{RGB}{179,205,227}
\definecolor{Pastel1-4-3}{RGB}{204,235,197}
\definecolor{Pastel1-4-C}{RGB}{204,235,197}
\definecolor{Pastel1-4-4}{RGB}{222,203,228}
\definecolor{Pastel1-4-D}{RGB}{222,203,228}
\definecolor{Pastel1-5-1}{RGB}{251,180,174}
\definecolor{Pastel1-5-A}{RGB}{251,180,174}
\definecolor{Pastel1-5-2}{RGB}{179,205,227}
\definecolor{Pastel1-5-B}{RGB}{179,205,227}
\definecolor{Pastel1-5-3}{RGB}{204,235,197}
\definecolor{Pastel1-5-C}{RGB}{204,235,197}
\definecolor{Pastel1-5-4}{RGB}{222,203,228}
\definecolor{Pastel1-5-D}{RGB}{222,203,228}
\definecolor{Pastel1-5-5}{RGB}{254,217,166}
\definecolor{Pastel1-5-E}{RGB}{254,217,166}
\definecolor{Pastel1-6-1}{RGB}{251,180,174}
\definecolor{Pastel1-6-A}{RGB}{251,180,174}
\definecolor{Pastel1-6-2}{RGB}{179,205,227}
\definecolor{Pastel1-6-B}{RGB}{179,205,227}
\definecolor{Pastel1-6-3}{RGB}{204,235,197}
\definecolor{Pastel1-6-C}{RGB}{204,235,197}
\definecolor{Pastel1-6-4}{RGB}{222,203,228}
\definecolor{Pastel1-6-D}{RGB}{222,203,228}
\definecolor{Pastel1-6-5}{RGB}{254,217,166}
\definecolor{Pastel1-6-E}{RGB}{254,217,166}
\definecolor{Pastel1-6-6}{RGB}{255,255,204}
\definecolor{Pastel1-6-F}{RGB}{255,255,204}
\definecolor{Pastel1-7-1}{RGB}{251,180,174}
\definecolor{Pastel1-7-A}{RGB}{251,180,174}
\definecolor{Pastel1-7-2}{RGB}{179,205,227}
\definecolor{Pastel1-7-B}{RGB}{179,205,227}
\definecolor{Pastel1-7-3}{RGB}{204,235,197}
\definecolor{Pastel1-7-C}{RGB}{204,235,197}
\definecolor{Pastel1-7-4}{RGB}{222,203,228}
\definecolor{Pastel1-7-D}{RGB}{222,203,228}
\definecolor{Pastel1-7-5}{RGB}{254,217,166}
\definecolor{Pastel1-7-E}{RGB}{254,217,166}
\definecolor{Pastel1-7-6}{RGB}{255,255,204}
\definecolor{Pastel1-7-F}{RGB}{255,255,204}
\definecolor{Pastel1-7-7}{RGB}{229,216,189}
\definecolor{Pastel1-7-G}{RGB}{229,216,189}
\definecolor{Pastel1-8-1}{RGB}{251,180,174}
\definecolor{Pastel1-8-A}{RGB}{251,180,174}
\definecolor{Pastel1-8-2}{RGB}{179,205,227}
\definecolor{Pastel1-8-B}{RGB}{179,205,227}
\definecolor{Pastel1-8-3}{RGB}{204,235,197}
\definecolor{Pastel1-8-C}{RGB}{204,235,197}
\definecolor{Pastel1-8-4}{RGB}{222,203,228}
\definecolor{Pastel1-8-D}{RGB}{222,203,228}
\definecolor{Pastel1-8-5}{RGB}{254,217,166}
\definecolor{Pastel1-8-E}{RGB}{254,217,166}
\definecolor{Pastel1-8-6}{RGB}{255,255,204}
\definecolor{Pastel1-8-F}{RGB}{255,255,204}
\definecolor{Pastel1-8-7}{RGB}{229,216,189}
\definecolor{Pastel1-8-G}{RGB}{229,216,189}
\definecolor{Pastel1-8-8}{RGB}{253,218,236}
\definecolor{Pastel1-8-H}{RGB}{253,218,236}
\definecolor{Pastel1-9-1}{RGB}{251,180,174}
\definecolor{Pastel1-9-A}{RGB}{251,180,174}
\definecolor{Pastel1-9-2}{RGB}{179,205,227}
\definecolor{Pastel1-9-B}{RGB}{179,205,227}
\definecolor{Pastel1-9-3}{RGB}{204,235,197}
\definecolor{Pastel1-9-C}{RGB}{204,235,197}
\definecolor{Pastel1-9-4}{RGB}{222,203,228}
\definecolor{Pastel1-9-D}{RGB}{222,203,228}
\definecolor{Pastel1-9-5}{RGB}{254,217,166}
\definecolor{Pastel1-9-E}{RGB}{254,217,166}
\definecolor{Pastel1-9-6}{RGB}{255,255,204}
\definecolor{Pastel1-9-F}{RGB}{255,255,204}
\definecolor{Pastel1-9-7}{RGB}{229,216,189}
\definecolor{Pastel1-9-G}{RGB}{229,216,189}
\definecolor{Pastel1-9-8}{RGB}{253,218,236}
\definecolor{Pastel1-9-H}{RGB}{253,218,236}
\definecolor{Pastel1-9-9}{RGB}{242,242,242}
\definecolor{Pastel1-9-I}{RGB}{242,242,242}
\definecolor{Set1-3-1}{RGB}{228,26,28}
\definecolor{Set1-3-A}{RGB}{228,26,28}
\definecolor{Set1-3-2}{RGB}{55,126,184}
\definecolor{Set1-3-B}{RGB}{55,126,184}
\definecolor{Set1-3-3}{RGB}{77,175,74}
\definecolor{Set1-3-C}{RGB}{77,175,74}
\definecolor{Set1-4-1}{RGB}{228,26,28}
\definecolor{Set1-4-A}{RGB}{228,26,28}
\definecolor{Set1-4-2}{RGB}{55,126,184}
\definecolor{Set1-4-B}{RGB}{55,126,184}
\definecolor{Set1-4-3}{RGB}{77,175,74}
\definecolor{Set1-4-C}{RGB}{77,175,74}
\definecolor{Set1-4-4}{RGB}{152,78,163}
\definecolor{Set1-4-D}{RGB}{152,78,163}
\definecolor{Set1-5-1}{RGB}{228,26,28}
\definecolor{Set1-5-A}{RGB}{228,26,28}
\definecolor{Set1-5-2}{RGB}{55,126,184}
\definecolor{Set1-5-B}{RGB}{55,126,184}
\definecolor{Set1-5-3}{RGB}{77,175,74}
\definecolor{Set1-5-C}{RGB}{77,175,74}
\definecolor{Set1-5-4}{RGB}{152,78,163}
\definecolor{Set1-5-D}{RGB}{152,78,163}
\definecolor{Set1-5-5}{RGB}{255,127,0}
\definecolor{Set1-5-E}{RGB}{255,127,0}
\definecolor{Set1-6-1}{RGB}{228,26,28}
\definecolor{Set1-6-A}{RGB}{228,26,28}
\definecolor{Set1-6-2}{RGB}{55,126,184}
\definecolor{Set1-6-B}{RGB}{55,126,184}
\definecolor{Set1-6-3}{RGB}{77,175,74}
\definecolor{Set1-6-C}{RGB}{77,175,74}
\definecolor{Set1-6-4}{RGB}{152,78,163}
\definecolor{Set1-6-D}{RGB}{152,78,163}
\definecolor{Set1-6-5}{RGB}{255,127,0}
\definecolor{Set1-6-E}{RGB}{255,127,0}
\definecolor{Set1-6-6}{RGB}{255,255,51}
\definecolor{Set1-6-F}{RGB}{255,255,51}
\definecolor{Set1-7-1}{RGB}{228,26,28}
\definecolor{Set1-7-A}{RGB}{228,26,28}
\definecolor{Set1-7-2}{RGB}{55,126,184}
\definecolor{Set1-7-B}{RGB}{55,126,184}
\definecolor{Set1-7-3}{RGB}{77,175,74}
\definecolor{Set1-7-C}{RGB}{77,175,74}
\definecolor{Set1-7-4}{RGB}{152,78,163}
\definecolor{Set1-7-D}{RGB}{152,78,163}
\definecolor{Set1-7-5}{RGB}{255,127,0}
\definecolor{Set1-7-E}{RGB}{255,127,0}
\definecolor{Set1-7-6}{RGB}{255,255,51}
\definecolor{Set1-7-F}{RGB}{255,255,51}
\definecolor{Set1-7-7}{RGB}{166,86,40}
\definecolor{Set1-7-G}{RGB}{166,86,40}
\definecolor{Set1-8-1}{RGB}{228,26,28}
\definecolor{Set1-8-A}{RGB}{228,26,28}
\definecolor{Set1-8-2}{RGB}{55,126,184}
\definecolor{Set1-8-B}{RGB}{55,126,184}
\definecolor{Set1-8-3}{RGB}{77,175,74}
\definecolor{Set1-8-C}{RGB}{77,175,74}
\definecolor{Set1-8-4}{RGB}{152,78,163}
\definecolor{Set1-8-D}{RGB}{152,78,163}
\definecolor{Set1-8-5}{RGB}{255,127,0}
\definecolor{Set1-8-E}{RGB}{255,127,0}
\definecolor{Set1-8-6}{RGB}{255,255,51}
\definecolor{Set1-8-F}{RGB}{255,255,51}
\definecolor{Set1-8-7}{RGB}{166,86,40}
\definecolor{Set1-8-G}{RGB}{166,86,40}
\definecolor{Set1-8-8}{RGB}{247,129,191}
\definecolor{Set1-8-H}{RGB}{247,129,191}
\definecolor{Set1-9-1}{RGB}{228,26,28}
\definecolor{Set1-9-A}{RGB}{228,26,28}
\definecolor{Set1-9-2}{RGB}{55,126,184}
\definecolor{Set1-9-B}{RGB}{55,126,184}
\definecolor{Set1-9-3}{RGB}{77,175,74}
\definecolor{Set1-9-C}{RGB}{77,175,74}
\definecolor{Set1-9-4}{RGB}{152,78,163}
\definecolor{Set1-9-D}{RGB}{152,78,163}
\definecolor{Set1-9-5}{RGB}{255,127,0}
\definecolor{Set1-9-E}{RGB}{255,127,0}
\definecolor{Set1-9-6}{RGB}{255,255,51}
\definecolor{Set1-9-F}{RGB}{255,255,51}
\definecolor{Set1-9-7}{RGB}{166,86,40}
\definecolor{Set1-9-G}{RGB}{166,86,40}
\definecolor{Set1-9-8}{RGB}{247,129,191}
\definecolor{Set1-9-H}{RGB}{247,129,191}
\definecolor{Set1-9-9}{RGB}{153,153,153}
\definecolor{Set1-9-I}{RGB}{153,153,153}
\definecolor{Pastel2-3-1}{RGB}{179,226,205}
\definecolor{Pastel2-3-A}{RGB}{179,226,205}
\definecolor{Pastel2-3-2}{RGB}{253,205,172}
\definecolor{Pastel2-3-B}{RGB}{253,205,172}
\definecolor{Pastel2-3-3}{RGB}{203,213,232}
\definecolor{Pastel2-3-C}{RGB}{203,213,232}
\definecolor{Pastel2-4-1}{RGB}{179,226,205}
\definecolor{Pastel2-4-A}{RGB}{179,226,205}
\definecolor{Pastel2-4-2}{RGB}{253,205,172}
\definecolor{Pastel2-4-B}{RGB}{253,205,172}
\definecolor{Pastel2-4-3}{RGB}{203,213,232}
\definecolor{Pastel2-4-C}{RGB}{203,213,232}
\definecolor{Pastel2-4-4}{RGB}{244,202,228}
\definecolor{Pastel2-4-D}{RGB}{244,202,228}
\definecolor{Pastel2-5-1}{RGB}{179,226,205}
\definecolor{Pastel2-5-A}{RGB}{179,226,205}
\definecolor{Pastel2-5-2}{RGB}{253,205,172}
\definecolor{Pastel2-5-B}{RGB}{253,205,172}
\definecolor{Pastel2-5-3}{RGB}{203,213,232}
\definecolor{Pastel2-5-C}{RGB}{203,213,232}
\definecolor{Pastel2-5-4}{RGB}{244,202,228}
\definecolor{Pastel2-5-D}{RGB}{244,202,228}
\definecolor{Pastel2-5-5}{RGB}{230,245,201}
\definecolor{Pastel2-5-E}{RGB}{230,245,201}
\definecolor{Pastel2-6-1}{RGB}{179,226,205}
\definecolor{Pastel2-6-A}{RGB}{179,226,205}
\definecolor{Pastel2-6-2}{RGB}{253,205,172}
\definecolor{Pastel2-6-B}{RGB}{253,205,172}
\definecolor{Pastel2-6-3}{RGB}{203,213,232}
\definecolor{Pastel2-6-C}{RGB}{203,213,232}
\definecolor{Pastel2-6-4}{RGB}{244,202,228}
\definecolor{Pastel2-6-D}{RGB}{244,202,228}
\definecolor{Pastel2-6-5}{RGB}{230,245,201}
\definecolor{Pastel2-6-E}{RGB}{230,245,201}
\definecolor{Pastel2-6-6}{RGB}{255,242,174}
\definecolor{Pastel2-6-F}{RGB}{255,242,174}
\definecolor{Pastel2-7-1}{RGB}{179,226,205}
\definecolor{Pastel2-7-A}{RGB}{179,226,205}
\definecolor{Pastel2-7-2}{RGB}{253,205,172}
\definecolor{Pastel2-7-B}{RGB}{253,205,172}
\definecolor{Pastel2-7-3}{RGB}{203,213,232}
\definecolor{Pastel2-7-C}{RGB}{203,213,232}
\definecolor{Pastel2-7-4}{RGB}{244,202,228}
\definecolor{Pastel2-7-D}{RGB}{244,202,228}
\definecolor{Pastel2-7-5}{RGB}{230,245,201}
\definecolor{Pastel2-7-E}{RGB}{230,245,201}
\definecolor{Pastel2-7-6}{RGB}{255,242,174}
\definecolor{Pastel2-7-F}{RGB}{255,242,174}
\definecolor{Pastel2-7-7}{RGB}{241,226,204}
\definecolor{Pastel2-7-G}{RGB}{241,226,204}
\definecolor{Pastel2-8-1}{RGB}{179,226,205}
\definecolor{Pastel2-8-A}{RGB}{179,226,205}
\definecolor{Pastel2-8-2}{RGB}{253,205,172}
\definecolor{Pastel2-8-B}{RGB}{253,205,172}
\definecolor{Pastel2-8-3}{RGB}{203,213,232}
\definecolor{Pastel2-8-C}{RGB}{203,213,232}
\definecolor{Pastel2-8-4}{RGB}{244,202,228}
\definecolor{Pastel2-8-D}{RGB}{244,202,228}
\definecolor{Pastel2-8-5}{RGB}{230,245,201}
\definecolor{Pastel2-8-E}{RGB}{230,245,201}
\definecolor{Pastel2-8-6}{RGB}{255,242,174}
\definecolor{Pastel2-8-F}{RGB}{255,242,174}
\definecolor{Pastel2-8-7}{RGB}{241,226,204}
\definecolor{Pastel2-8-G}{RGB}{241,226,204}
\definecolor{Pastel2-8-8}{RGB}{204,204,204}
\definecolor{Pastel2-8-H}{RGB}{204,204,204}
\definecolor{Set2-3-1}{RGB}{102,194,165}
\definecolor{Set2-3-A}{RGB}{102,194,165}
\definecolor{Set2-3-2}{RGB}{252,141,98}
\definecolor{Set2-3-B}{RGB}{252,141,98}
\definecolor{Set2-3-3}{RGB}{141,160,203}
\definecolor{Set2-3-C}{RGB}{141,160,203}
\definecolor{Set2-4-1}{RGB}{102,194,165}
\definecolor{Set2-4-A}{RGB}{102,194,165}
\definecolor{Set2-4-2}{RGB}{252,141,98}
\definecolor{Set2-4-B}{RGB}{252,141,98}
\definecolor{Set2-4-3}{RGB}{141,160,203}
\definecolor{Set2-4-C}{RGB}{141,160,203}
\definecolor{Set2-4-4}{RGB}{231,138,195}
\definecolor{Set2-4-D}{RGB}{231,138,195}
\definecolor{Set2-5-1}{RGB}{102,194,165}
\definecolor{Set2-5-A}{RGB}{102,194,165}
\definecolor{Set2-5-2}{RGB}{252,141,98}
\definecolor{Set2-5-B}{RGB}{252,141,98}
\definecolor{Set2-5-3}{RGB}{141,160,203}
\definecolor{Set2-5-C}{RGB}{141,160,203}
\definecolor{Set2-5-4}{RGB}{231,138,195}
\definecolor{Set2-5-D}{RGB}{231,138,195}
\definecolor{Set2-5-5}{RGB}{166,216,84}
\definecolor{Set2-5-E}{RGB}{166,216,84}
\definecolor{Set2-6-1}{RGB}{102,194,165}
\definecolor{Set2-6-A}{RGB}{102,194,165}
\definecolor{Set2-6-2}{RGB}{252,141,98}
\definecolor{Set2-6-B}{RGB}{252,141,98}
\definecolor{Set2-6-3}{RGB}{141,160,203}
\definecolor{Set2-6-C}{RGB}{141,160,203}
\definecolor{Set2-6-4}{RGB}{231,138,195}
\definecolor{Set2-6-D}{RGB}{231,138,195}
\definecolor{Set2-6-5}{RGB}{166,216,84}
\definecolor{Set2-6-E}{RGB}{166,216,84}
\definecolor{Set2-6-6}{RGB}{255,217,47}
\definecolor{Set2-6-F}{RGB}{255,217,47}
\definecolor{Set2-7-1}{RGB}{102,194,165}
\definecolor{Set2-7-A}{RGB}{102,194,165}
\definecolor{Set2-7-2}{RGB}{252,141,98}
\definecolor{Set2-7-B}{RGB}{252,141,98}
\definecolor{Set2-7-3}{RGB}{141,160,203}
\definecolor{Set2-7-C}{RGB}{141,160,203}
\definecolor{Set2-7-4}{RGB}{231,138,195}
\definecolor{Set2-7-D}{RGB}{231,138,195}
\definecolor{Set2-7-5}{RGB}{166,216,84}
\definecolor{Set2-7-E}{RGB}{166,216,84}
\definecolor{Set2-7-6}{RGB}{255,217,47}
\definecolor{Set2-7-F}{RGB}{255,217,47}
\definecolor{Set2-7-7}{RGB}{229,196,148}
\definecolor{Set2-7-G}{RGB}{229,196,148}
\definecolor{Set2-8-1}{RGB}{102,194,165}
\definecolor{Set2-8-A}{RGB}{102,194,165}
\definecolor{Set2-8-2}{RGB}{252,141,98}
\definecolor{Set2-8-B}{RGB}{252,141,98}
\definecolor{Set2-8-3}{RGB}{141,160,203}
\definecolor{Set2-8-C}{RGB}{141,160,203}
\definecolor{Set2-8-4}{RGB}{231,138,195}
\definecolor{Set2-8-D}{RGB}{231,138,195}
\definecolor{Set2-8-5}{RGB}{166,216,84}
\definecolor{Set2-8-E}{RGB}{166,216,84}
\definecolor{Set2-8-6}{RGB}{255,217,47}
\definecolor{Set2-8-F}{RGB}{255,217,47}
\definecolor{Set2-8-7}{RGB}{229,196,148}
\definecolor{Set2-8-G}{RGB}{229,196,148}
\definecolor{Set2-8-8}{RGB}{179,179,179}
\definecolor{Set2-8-H}{RGB}{179,179,179}
\definecolor{Dark2-3-1}{RGB}{27,158,119}
\definecolor{Dark2-3-A}{RGB}{27,158,119}
\definecolor{Dark2-3-2}{RGB}{217,95,2}
\definecolor{Dark2-3-B}{RGB}{217,95,2}
\definecolor{Dark2-3-3}{RGB}{117,112,179}
\definecolor{Dark2-3-C}{RGB}{117,112,179}
\definecolor{Dark2-4-1}{RGB}{27,158,119}
\definecolor{Dark2-4-A}{RGB}{27,158,119}
\definecolor{Dark2-4-2}{RGB}{217,95,2}
\definecolor{Dark2-4-B}{RGB}{217,95,2}
\definecolor{Dark2-4-3}{RGB}{117,112,179}
\definecolor{Dark2-4-C}{RGB}{117,112,179}
\definecolor{Dark2-4-4}{RGB}{231,41,138}
\definecolor{Dark2-4-D}{RGB}{231,41,138}
\definecolor{Dark2-5-1}{RGB}{27,158,119}
\definecolor{Dark2-5-A}{RGB}{27,158,119}
\definecolor{Dark2-5-2}{RGB}{217,95,2}
\definecolor{Dark2-5-B}{RGB}{217,95,2}
\definecolor{Dark2-5-3}{RGB}{117,112,179}
\definecolor{Dark2-5-C}{RGB}{117,112,179}
\definecolor{Dark2-5-4}{RGB}{231,41,138}
\definecolor{Dark2-5-D}{RGB}{231,41,138}
\definecolor{Dark2-5-5}{RGB}{102,166,30}
\definecolor{Dark2-5-E}{RGB}{102,166,30}
\definecolor{Dark2-6-1}{RGB}{27,158,119}
\definecolor{Dark2-6-A}{RGB}{27,158,119}
\definecolor{Dark2-6-2}{RGB}{217,95,2}
\definecolor{Dark2-6-B}{RGB}{217,95,2}
\definecolor{Dark2-6-3}{RGB}{117,112,179}
\definecolor{Dark2-6-C}{RGB}{117,112,179}
\definecolor{Dark2-6-4}{RGB}{231,41,138}
\definecolor{Dark2-6-D}{RGB}{231,41,138}
\definecolor{Dark2-6-5}{RGB}{102,166,30}
\definecolor{Dark2-6-E}{RGB}{102,166,30}
\definecolor{Dark2-6-6}{RGB}{230,171,2}
\definecolor{Dark2-6-F}{RGB}{230,171,2}
\definecolor{Dark2-7-1}{RGB}{27,158,119}
\definecolor{Dark2-7-A}{RGB}{27,158,119}
\definecolor{Dark2-7-2}{RGB}{217,95,2}
\definecolor{Dark2-7-B}{RGB}{217,95,2}
\definecolor{Dark2-7-3}{RGB}{117,112,179}
\definecolor{Dark2-7-C}{RGB}{117,112,179}
\definecolor{Dark2-7-4}{RGB}{231,41,138}
\definecolor{Dark2-7-D}{RGB}{231,41,138}
\definecolor{Dark2-7-5}{RGB}{102,166,30}
\definecolor{Dark2-7-E}{RGB}{102,166,30}
\definecolor{Dark2-7-6}{RGB}{230,171,2}
\definecolor{Dark2-7-F}{RGB}{230,171,2}
\definecolor{Dark2-7-7}{RGB}{166,118,29}
\definecolor{Dark2-7-G}{RGB}{166,118,29}
\definecolor{Dark2-8-1}{RGB}{27,158,119}
\definecolor{Dark2-8-A}{RGB}{27,158,119}
\definecolor{Dark2-8-2}{RGB}{217,95,2}
\definecolor{Dark2-8-B}{RGB}{217,95,2}
\definecolor{Dark2-8-3}{RGB}{117,112,179}
\definecolor{Dark2-8-C}{RGB}{117,112,179}
\definecolor{Dark2-8-4}{RGB}{231,41,138}
\definecolor{Dark2-8-D}{RGB}{231,41,138}
\definecolor{Dark2-8-5}{RGB}{102,166,30}
\definecolor{Dark2-8-E}{RGB}{102,166,30}
\definecolor{Dark2-8-6}{RGB}{230,171,2}
\definecolor{Dark2-8-F}{RGB}{230,171,2}
\definecolor{Dark2-8-7}{RGB}{166,118,29}
\definecolor{Dark2-8-G}{RGB}{166,118,29}
\definecolor{Dark2-8-8}{RGB}{102,102,102}
\definecolor{Dark2-8-H}{RGB}{102,102,102}
\definecolor{Paired-3-1}{RGB}{166,206,227}
\definecolor{Paired-3-A}{RGB}{166,206,227}
\definecolor{Paired-3-2}{RGB}{31,120,180}
\definecolor{Paired-3-B}{RGB}{31,120,180}
\definecolor{Paired-3-3}{RGB}{178,223,138}
\definecolor{Paired-3-C}{RGB}{178,223,138}
\definecolor{Paired-4-1}{RGB}{166,206,227}
\definecolor{Paired-4-A}{RGB}{166,206,227}
\definecolor{Paired-4-2}{RGB}{31,120,180}
\definecolor{Paired-4-B}{RGB}{31,120,180}
\definecolor{Paired-4-3}{RGB}{178,223,138}
\definecolor{Paired-4-C}{RGB}{178,223,138}
\definecolor{Paired-4-4}{RGB}{51,160,44}
\definecolor{Paired-4-D}{RGB}{51,160,44}
\definecolor{Paired-5-1}{RGB}{166,206,227}
\definecolor{Paired-5-A}{RGB}{166,206,227}
\definecolor{Paired-5-2}{RGB}{31,120,180}
\definecolor{Paired-5-B}{RGB}{31,120,180}
\definecolor{Paired-5-3}{RGB}{178,223,138}
\definecolor{Paired-5-C}{RGB}{178,223,138}
\definecolor{Paired-5-4}{RGB}{51,160,44}
\definecolor{Paired-5-D}{RGB}{51,160,44}
\definecolor{Paired-5-5}{RGB}{251,154,153}
\definecolor{Paired-5-E}{RGB}{251,154,153}
\definecolor{Paired-6-1}{RGB}{166,206,227}
\definecolor{Paired-6-A}{RGB}{166,206,227}
\definecolor{Paired-6-2}{RGB}{31,120,180}
\definecolor{Paired-6-B}{RGB}{31,120,180}
\definecolor{Paired-6-3}{RGB}{178,223,138}
\definecolor{Paired-6-C}{RGB}{178,223,138}
\definecolor{Paired-6-4}{RGB}{51,160,44}
\definecolor{Paired-6-D}{RGB}{51,160,44}
\definecolor{Paired-6-5}{RGB}{251,154,153}
\definecolor{Paired-6-E}{RGB}{251,154,153}
\definecolor{Paired-6-6}{RGB}{227,26,28}
\definecolor{Paired-6-F}{RGB}{227,26,28}
\definecolor{Paired-7-1}{RGB}{166,206,227}
\definecolor{Paired-7-A}{RGB}{166,206,227}
\definecolor{Paired-7-2}{RGB}{31,120,180}
\definecolor{Paired-7-B}{RGB}{31,120,180}
\definecolor{Paired-7-3}{RGB}{178,223,138}
\definecolor{Paired-7-C}{RGB}{178,223,138}
\definecolor{Paired-7-4}{RGB}{51,160,44}
\definecolor{Paired-7-D}{RGB}{51,160,44}
\definecolor{Paired-7-5}{RGB}{251,154,153}
\definecolor{Paired-7-E}{RGB}{251,154,153}
\definecolor{Paired-7-6}{RGB}{227,26,28}
\definecolor{Paired-7-F}{RGB}{227,26,28}
\definecolor{Paired-7-7}{RGB}{253,191,111}
\definecolor{Paired-7-G}{RGB}{253,191,111}
\definecolor{Paired-8-1}{RGB}{166,206,227}
\definecolor{Paired-8-A}{RGB}{166,206,227}
\definecolor{Paired-8-2}{RGB}{31,120,180}
\definecolor{Paired-8-B}{RGB}{31,120,180}
\definecolor{Paired-8-3}{RGB}{178,223,138}
\definecolor{Paired-8-C}{RGB}{178,223,138}
\definecolor{Paired-8-4}{RGB}{51,160,44}
\definecolor{Paired-8-D}{RGB}{51,160,44}
\definecolor{Paired-8-5}{RGB}{251,154,153}
\definecolor{Paired-8-E}{RGB}{251,154,153}
\definecolor{Paired-8-6}{RGB}{227,26,28}
\definecolor{Paired-8-F}{RGB}{227,26,28}
\definecolor{Paired-8-7}{RGB}{253,191,111}
\definecolor{Paired-8-G}{RGB}{253,191,111}
\definecolor{Paired-8-8}{RGB}{255,127,0}
\definecolor{Paired-8-H}{RGB}{255,127,0}
\definecolor{Paired-9-1}{RGB}{166,206,227}
\definecolor{Paired-9-A}{RGB}{166,206,227}
\definecolor{Paired-9-2}{RGB}{31,120,180}
\definecolor{Paired-9-B}{RGB}{31,120,180}
\definecolor{Paired-9-3}{RGB}{178,223,138}
\definecolor{Paired-9-C}{RGB}{178,223,138}
\definecolor{Paired-9-4}{RGB}{51,160,44}
\definecolor{Paired-9-D}{RGB}{51,160,44}
\definecolor{Paired-9-5}{RGB}{251,154,153}
\definecolor{Paired-9-E}{RGB}{251,154,153}
\definecolor{Paired-9-6}{RGB}{227,26,28}
\definecolor{Paired-9-F}{RGB}{227,26,28}
\definecolor{Paired-9-7}{RGB}{253,191,111}
\definecolor{Paired-9-G}{RGB}{253,191,111}
\definecolor{Paired-9-8}{RGB}{255,127,0}
\definecolor{Paired-9-H}{RGB}{255,127,0}
\definecolor{Paired-9-9}{RGB}{202,178,214}
\definecolor{Paired-9-I}{RGB}{202,178,214}
\definecolor{Paired-10-1}{RGB}{166,206,227}
\definecolor{Paired-10-A}{RGB}{166,206,227}
\definecolor{Paired-10-2}{RGB}{31,120,180}
\definecolor{Paired-10-B}{RGB}{31,120,180}
\definecolor{Paired-10-3}{RGB}{178,223,138}
\definecolor{Paired-10-C}{RGB}{178,223,138}
\definecolor{Paired-10-4}{RGB}{51,160,44}
\definecolor{Paired-10-D}{RGB}{51,160,44}
\definecolor{Paired-10-5}{RGB}{251,154,153}
\definecolor{Paired-10-E}{RGB}{251,154,153}
\definecolor{Paired-10-6}{RGB}{227,26,28}
\definecolor{Paired-10-F}{RGB}{227,26,28}
\definecolor{Paired-10-7}{RGB}{253,191,111}
\definecolor{Paired-10-G}{RGB}{253,191,111}
\definecolor{Paired-10-8}{RGB}{255,127,0}
\definecolor{Paired-10-H}{RGB}{255,127,0}
\definecolor{Paired-10-9}{RGB}{202,178,214}
\definecolor{Paired-10-I}{RGB}{202,178,214}
\definecolor{Paired-10-10}{RGB}{106,61,154}
\definecolor{Paired-10-J}{RGB}{106,61,154}
\definecolor{Paired-11-1}{RGB}{166,206,227}
\definecolor{Paired-11-A}{RGB}{166,206,227}
\definecolor{Paired-11-2}{RGB}{31,120,180}
\definecolor{Paired-11-B}{RGB}{31,120,180}
\definecolor{Paired-11-3}{RGB}{178,223,138}
\definecolor{Paired-11-C}{RGB}{178,223,138}
\definecolor{Paired-11-4}{RGB}{51,160,44}
\definecolor{Paired-11-D}{RGB}{51,160,44}
\definecolor{Paired-11-5}{RGB}{251,154,153}
\definecolor{Paired-11-E}{RGB}{251,154,153}
\definecolor{Paired-11-6}{RGB}{227,26,28}
\definecolor{Paired-11-F}{RGB}{227,26,28}
\definecolor{Paired-11-7}{RGB}{253,191,111}
\definecolor{Paired-11-G}{RGB}{253,191,111}
\definecolor{Paired-11-8}{RGB}{255,127,0}
\definecolor{Paired-11-H}{RGB}{255,127,0}
\definecolor{Paired-11-9}{RGB}{202,178,214}
\definecolor{Paired-11-I}{RGB}{202,178,214}
\definecolor{Paired-11-10}{RGB}{106,61,154}
\definecolor{Paired-11-J}{RGB}{106,61,154}
\definecolor{Paired-11-11}{RGB}{255,255,153}
\definecolor{Paired-11-K}{RGB}{255,255,153}
\definecolor{Paired-12-1}{RGB}{166,206,227}
\definecolor{Paired-12-A}{RGB}{166,206,227}
\definecolor{Paired-12-2}{RGB}{31,120,180}
\definecolor{Paired-12-B}{RGB}{31,120,180}
\definecolor{Paired-12-3}{RGB}{178,223,138}
\definecolor{Paired-12-C}{RGB}{178,223,138}
\definecolor{Paired-12-4}{RGB}{51,160,44}
\definecolor{Paired-12-D}{RGB}{51,160,44}
\definecolor{Paired-12-5}{RGB}{251,154,153}
\definecolor{Paired-12-E}{RGB}{251,154,153}
\definecolor{Paired-12-6}{RGB}{227,26,28}
\definecolor{Paired-12-F}{RGB}{227,26,28}
\definecolor{Paired-12-7}{RGB}{253,191,111}
\definecolor{Paired-12-G}{RGB}{253,191,111}
\definecolor{Paired-12-8}{RGB}{255,127,0}
\definecolor{Paired-12-H}{RGB}{255,127,0}
\definecolor{Paired-12-9}{RGB}{202,178,214}
\definecolor{Paired-12-I}{RGB}{202,178,214}
\definecolor{Paired-12-10}{RGB}{106,61,154}
\definecolor{Paired-12-J}{RGB}{106,61,154}
\definecolor{Paired-12-11}{RGB}{255,255,153}
\definecolor{Paired-12-K}{RGB}{255,255,153}
\definecolor{Paired-12-12}{RGB}{177,89,40}
\definecolor{Paired-12-L}{RGB}{177,89,40}
\definecolor{Accent-3-1}{RGB}{127,201,127}
\definecolor{Accent-3-A}{RGB}{127,201,127}
\definecolor{Accent-3-2}{RGB}{190,174,212}
\definecolor{Accent-3-B}{RGB}{190,174,212}
\definecolor{Accent-3-3}{RGB}{253,192,134}
\definecolor{Accent-3-C}{RGB}{253,192,134}
\definecolor{Accent-4-1}{RGB}{127,201,127}
\definecolor{Accent-4-A}{RGB}{127,201,127}
\definecolor{Accent-4-2}{RGB}{190,174,212}
\definecolor{Accent-4-B}{RGB}{190,174,212}
\definecolor{Accent-4-3}{RGB}{253,192,134}
\definecolor{Accent-4-C}{RGB}{253,192,134}
\definecolor{Accent-4-4}{RGB}{255,255,153}
\definecolor{Accent-4-D}{RGB}{255,255,153}
\definecolor{Accent-5-1}{RGB}{127,201,127}
\definecolor{Accent-5-A}{RGB}{127,201,127}
\definecolor{Accent-5-2}{RGB}{190,174,212}
\definecolor{Accent-5-B}{RGB}{190,174,212}
\definecolor{Accent-5-3}{RGB}{253,192,134}
\definecolor{Accent-5-C}{RGB}{253,192,134}
\definecolor{Accent-5-4}{RGB}{255,255,153}
\definecolor{Accent-5-D}{RGB}{255,255,153}
\definecolor{Accent-5-5}{RGB}{56,108,176}
\definecolor{Accent-5-E}{RGB}{56,108,176}
\definecolor{Accent-6-1}{RGB}{127,201,127}
\definecolor{Accent-6-A}{RGB}{127,201,127}
\definecolor{Accent-6-2}{RGB}{190,174,212}
\definecolor{Accent-6-B}{RGB}{190,174,212}
\definecolor{Accent-6-3}{RGB}{253,192,134}
\definecolor{Accent-6-C}{RGB}{253,192,134}
\definecolor{Accent-6-4}{RGB}{255,255,153}
\definecolor{Accent-6-D}{RGB}{255,255,153}
\definecolor{Accent-6-5}{RGB}{56,108,176}
\definecolor{Accent-6-E}{RGB}{56,108,176}
\definecolor{Accent-6-6}{RGB}{240,2,127}
\definecolor{Accent-6-F}{RGB}{240,2,127}
\definecolor{Accent-7-1}{RGB}{127,201,127}
\definecolor{Accent-7-A}{RGB}{127,201,127}
\definecolor{Accent-7-2}{RGB}{190,174,212}
\definecolor{Accent-7-B}{RGB}{190,174,212}
\definecolor{Accent-7-3}{RGB}{253,192,134}
\definecolor{Accent-7-C}{RGB}{253,192,134}
\definecolor{Accent-7-4}{RGB}{255,255,153}
\definecolor{Accent-7-D}{RGB}{255,255,153}
\definecolor{Accent-7-5}{RGB}{56,108,176}
\definecolor{Accent-7-E}{RGB}{56,108,176}
\definecolor{Accent-7-6}{RGB}{240,2,127}
\definecolor{Accent-7-F}{RGB}{240,2,127}
\definecolor{Accent-7-7}{RGB}{191,91,23}
\definecolor{Accent-7-G}{RGB}{191,91,23}
\definecolor{Accent-8-1}{RGB}{127,201,127}
\definecolor{Accent-8-A}{RGB}{127,201,127}
\definecolor{Accent-8-2}{RGB}{190,174,212}
\definecolor{Accent-8-B}{RGB}{190,174,212}
\definecolor{Accent-8-3}{RGB}{253,192,134}
\definecolor{Accent-8-C}{RGB}{253,192,134}
\definecolor{Accent-8-4}{RGB}{255,255,153}
\definecolor{Accent-8-D}{RGB}{255,255,153}
\definecolor{Accent-8-5}{RGB}{56,108,176}
\definecolor{Accent-8-E}{RGB}{56,108,176}
\definecolor{Accent-8-6}{RGB}{240,2,127}
\definecolor{Accent-8-F}{RGB}{240,2,127}
\definecolor{Accent-8-7}{RGB}{191,91,23}
\definecolor{Accent-8-G}{RGB}{191,91,23}
\definecolor{Accent-8-8}{RGB}{102,102,102}
\definecolor{Accent-8-H}{RGB}{102,102,102}

%% file: figs/plot_B.tex
\input{figs/tikzlibrarycolorbrewer.code.tex}
\def\minipagewidth{.3\textwidth}
\def\colorfirstplot{Paired-6-4}
\def\colorsecondplot{Paired-6-6}

\newcommand\addnmiplotB[7]{
\begin{tikzpicture}
\begin{axis}[
		enlarge x limits=false,
		xmin=0,
		xmax=1,
		xtick= {0, 0.5, 1},
		ytick= {0, 0.5, 0.7, 0.8, 0.9, 1}, 
		ymin=0.6,
		ymax=1.05,
		scale only axis, 
		width=100pt, 
	#2		
	]
\addplot+[draw=Paired-6-2,
				mark=none,
				line width=1pt,
				style=solid]				
					table[x=beta,y=NMI] {#1};
#3;
			
#4			 ;

\addplot+[draw=Paired-12-8,
				mark=none,
				line width=1pt,
				style=solid]				
					table[x=beta,y=NMI#5] {#1};
\lega;					
\addplot+[draw=Paired-12-10,
				mark=none,
				line width=1pt,
				style=solid]				
					table[x=beta,y=NMI#6] {#1};
\legb					
\addplot+[draw=Paired-12-12,
				mark=none,
				line width=1pt,
				style=solid]				
					table[x=beta,y=NMI#7] {#1};			
\legc;									

\addplot+[draw=Paired-6-2,
				mark=none,
				line width=1pt,
				style=dashed]				
					table[x=beta,y=NMIMOD] {#1};
						
\end{axis}			
					
\end{tikzpicture}
}

\newcommand\addstabilityplotB[2]{
\begin{tikzpicture}
\begin{semilogyaxis}[
		yticklabel style = {color=\colorfirstplot},
		enlarge x limits=false,
		 axis y line*=left,
		 xticklabel=\empty,
		xmin=0,
		xmax=1,
		ymin=1000,
		ymax=100000,
		scale only axis, 
		width=100pt, 
	#2		
	]

\addplot+[draw=\colorfirstplot,
				mark=none,
				line width=1pt,
				style=solid]				
					table[x=beta,y=NC2] {#1};

\addplot+[draw=\colorfirstplot,
				mark=none,
				line width=1pt,
				style=dashed]				
					table[x=beta,y=NC2T] {#1};

\end{semilogyaxis}

\begin{axis}[
		yticklabel style = {color=\colorsecondplot},
		enlarge x limits=false,
		 axis y line*=right,
		xtick= {0, 0.5, 1},
		xmin=0,
		xmax=1,
		ymin=0,
		ymax=400,
		scale only axis, 
		width=100pt, 
	#2		
	]

\addplot+[draw=\colorsecondplot,
				mark=none,
				line width=1pt,
				style=solid]				
					table[x=beta,y=C] {#1};

\addplot+[draw=\colorsecondplot,
				mark=none,
				line width=1pt,
				style=dashed]				
					table[x=beta,y=CT] {#1};

\end{axis}

\end{tikzpicture}
}

\newcommand\tmplegend[3]{%
    \def\lega{#1}%
    \def\legb{#2}%
    \def\legc{#3}%
}

\begin{figure}[t]
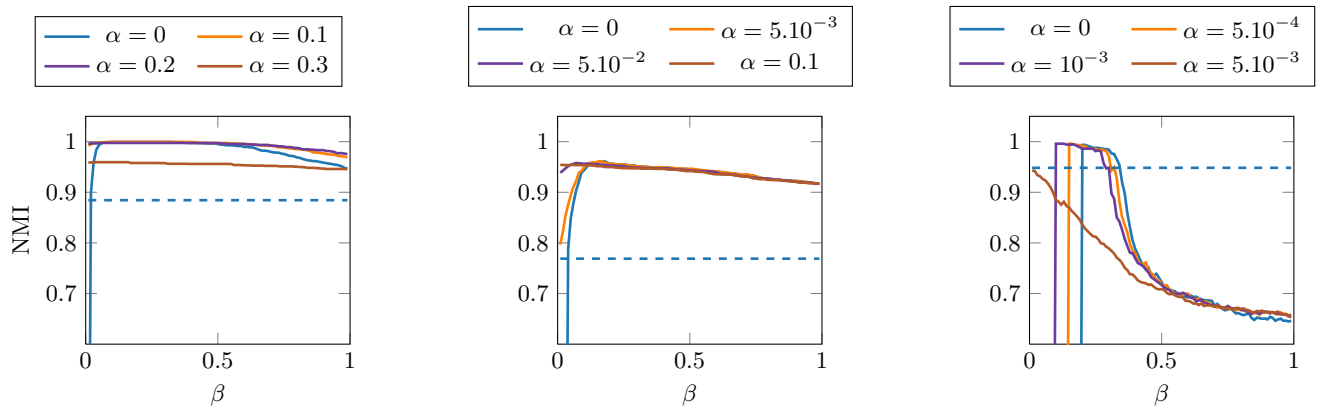

\centering
\begin{minipage}[b]{\minipagewidth}
  \centering
  \pgfplotslegendfromname{legend_multi_plot_1}
  \vskip10pt
			 \tmplegend 	{\addlegendentry{$\alpha=0.1$}} 
			 						{\addlegendentry{$\alpha=0.2$}} 
			 						{\addlegendentry{$\alpha=0.3$}}
  \addnmiplotB{figs/G1_B.dat}{ylabel={\small {NMI}}, xlabel =$\beta$, legend columns=2, legend style={/tikz/every even column/.append style={column sep=0.2cm}, legend to name={legend_multi_plot_1}}}{\addlegendentry{$\alpha=0$}}{}{0.1}{0.2}{0.3}  
\end{minipage}
\hfill
\begin{minipage}[b]{\minipagewidth}
  \centering
    \pgfplotslegendfromname{legend_multi_plot_2}
  \vskip10pt
			 \tmplegend 	{\addlegendentry{$\alpha=5.10^{-3}$}} 
			 						{\addlegendentry{$\alpha=5.10^{-2}$}} 
			 						{\addlegendentry{$\alpha=0.1$}}
    \addnmiplotB{figs/G2_B.dat}{ylabel=\phantom{\small {NMI}},  xlabel =$\beta$, legend columns=2, legend style={/tikz/every even column/.append style={column sep=0.2cm}, legend to name={legend_multi_plot_2}}} {\addlegendentry{$\alpha=0$}}{}{0.005}{0.05}{0.1} 
\end{minipage}
\hfill
\begin{minipage}[b]{\minipagewidth}
  \centering
    \pgfplotslegendfromname{legend_multi_plot_3}
  	\vskip10pt
			 \tmplegend 	{\addlegendentry{$\alpha=5.10^{-4}$}} 
			 						{\addlegendentry{$\alpha=10^{-3}$}} 
			 						{\addlegendentry{$\alpha=5.10^{-3}$}}
	\addnmiplotB{figs/G3_B.dat}{ylabel=\phantom{\small {NMI}},  xlabel =$\beta$, legend columns=2, legend style={/tikz/every even column/.append style={column sep=0.2cm}, legend to name={legend_multi_plot_3}}} {\addlegendentry{$\alpha=0$}}{}{0.0005}{0.001}{0.005} %
\end{minipage}

\caption{(color online) Quality of the partitions obtained using $f_B$, as a function of $\beta$ for different values of $\alpha$. As previously, the quality is measured using the normalized mutual information (NMI) of the partitions extracted with respect to the ``ground truth'' partition (solid line). The NMI obtained by the partition extracted using the modularity is also provided (dashed line). The columns correspond respectively to easy, medium and hard graphs as in Figure~\ref{fig:results}. The quality obtained by $f_A$ (i.e.~$\alpha=0$) can be outperformed with a fine tuning of the parameter $\alpha$ but the performance is highly sensitive to this parameter.
This includes interdependencies between $\alpha$ and $\beta$ as can be seen in the third column.}
\label{fig:results_B}
\end{figure}